\pdfoutput=1  
\documentclass[a4paper, UKenglish, cleveref, autoref, thm-restate]{lipics-v2021}

\usepackage{mathrsfs} 

\usepackage{amsmath}
\usepackage{amssymb}
\usepackage{subcaption}
\usepackage[utf8]{inputenc}

\usepackage{tikz}
\usetikzlibrary{arrows, automata, petri, backgrounds, positioning, shapes, patterns, calc, decorations.pathreplacing, decorations.pathmorphing, decorations.markings, fit, fadings, chains, scopes}

\tikzset{every path/.style={->,thick}}
\tikzstyle{envplace}=[circle, thick, draw=black, fill=white, minimum size=15pt, transform shape]
\tikzstyle{sysplace}=[circle, thick, draw=black, fill=black!20, minimum size=15pt ,transform shape]
\tikzstyle{transition}=[rectangle, thick, draw=black, fill=white, minimum size=10pt, transform shape]
\tikzstyle{token}=[circle, thick, fill=black, inner sep=1.5pt, transform shape]
\tikzstyle{specialSys}=[double=black!20, minimum size=13.5pt]
\tikzstyle{specialEnv}=[double, minimum size=13.5pt]

\newcommand{\refLemma}[1]{Lemma~\ref{lem:#1}}

\newcommand{\refDef}[1]{Def.~\ref{def:#1}}
\newcommand{\refFig}[1]{Fig.~\ref{fig:#1}}
\newcommand{\refAppendix}[1]{App.~\ref{sec:#1}}
\newcommand{\refSection}[1]{Sec.~\ref{sec:#1}}

\newcommand{\True}{\ensuremath\mathit{true}}
\newcommand{\False}{\ensuremath\mathit{false}}
\newcommand{\End}{\ensuremath\mathit{end}}
\newcommand{\pomF}{\ensuremath\mathbb{P}}
\newcommand{\pom}[1]{\pomF\left(#1\right)}

\newcommand{\pNet}{\ensuremath\mathcal{N}}
\newcommand{\pl}{\ensuremath\mathcal{P}}
\newcommand{\places}{\pl}
\newcommand{\tr}{\ensuremath\mathcal{T}}
\newcommand{\transitions}{\tr}
\newcommand{\fl}{\ensuremath\mathcal{F}}
\newcommand{\flow}{\fl}
\newcommand{\init}{\ensuremath\mathit{In}}

\newcommand{\pre}[2]{\mathit{pre}^{#1}(#2)}
\newcommand{\post}[2]{\mathit{post}^{#1}(#2)}
\newcommand{\petriNet}{\ensuremath\pNet=(\pl,\tr,\fl,\init)}
\newcommand{\conflict}[2]{{#1}\,\sharp\,{#2}}
\newcommand{\firable}[1]{\ensuremath [#1\rangle}

\newcommand{\brapro}{\ensuremath\iota}
\newcommand{\unf}{\ensuremath\brapro_U}
\newcommand{\past}{\ensuremath\mathit{past}}

\newcommand{\pGame}{\ensuremath\mathcal{G}}
\newcommand{\plS}{\ensuremath\pl_S}
\newcommand{\plE}{\ensuremath\pl_E}
\newcommand{\badplaces}{\ensuremath\pl_B}
\newcommand{\goodmarkings}{\ensuremath\mathcal{M}_G}
\newcommand{\badmarkings}{\ensuremath\mathcal{M}_B}
\newcommand{\pSpecial}{\ensuremath\mathcal{W}}
\newcommand{\petriGame}{\ensuremath\pGame=(\plS,\plE,\tr,\fl,\init,\pSpecial)}

\newcommand{\pGameBadMarkings}{\ensuremath\pGame=(\plS,\plE,\tr,\fl,\init,\badmarkings)}

\newcommand{\sysstrat}{\ensuremath{\sigma}}
\newcommand{\maxSys}{\ensuremath\text{max}_{S}}

\newcommand{\BR}{\ensuremath\mathit{BM}}

\newcommand{\bbD}{\ensuremath\mathbb{D}}
\newcommand{\reach}{\ensuremath\mathcal{R}}

\newcommand{\TPgame}{\ensuremath\mathbb{G}}
\newcommand{\TPstates}{\ensuremath V}
\newcommand{\TPsysstates}{\ensuremath\TPstates_0}
\newcommand{\TPenvstates}{\ensuremath\TPstates_1}
\newcommand{\TPinit}{\ensuremath I}
\newcommand{\TPedges}{\ensuremath E}
\newcommand{\TPfinal}{\ensuremath F}
\newcommand{\BuchiGame}{\ensuremath\TPgame = (\TPstates, \TPsysstates, \TPenvstates, \TPinit, \TPedges, \TPfinal)}
\newcommand{\TPoneState}{\ensuremath\mathcal{V}}
\newcommand{\TPdecset}{\ensuremath \mathit{DT}}

\newcommand{\TPtypeTwoMarking}{\ensuremath T2M}
\newcommand{\TPbackRulesGen}{\ensuremath \mathit{SBM}}

\newcommand{\DSid}{\ensuremath\mathit{id}}
\newcommand{\DSpl}{\ensuremath\mathit{pl}}
\newcommand{\DSdec}{\ensuremath\mathit{dec}}
\newcommand{\DStypetwo}{\ensuremath t2}
\newcommand{\DSlastmcut}{\ensuremath \mathit{lmc}}
\newcommand{\DSmarking}{\ensuremath\mathscr{M}}

\newcommand{\decisionsets}{\ensuremath\mathcal{D}}
\newcommand{\decsetgame}{\ensuremath\mathscr{D}}
\newcommand{\backwardrules}{\ensuremath\mathcal{B}}
\newcommand{\sucTT}{\ensuremath\mathit{suc}_{t2}}
\newcommand{\sucMC}{\ensuremath\mathit{suc}_{\mathit{mcut}}}
\newcommand{\sucS}{\ensuremath\mathit{suc}_{S}}
\newcommand{\sucDEC}{\ensuremath\mathit{suc}_{\mathit{dec}}}

\newcommand{\tfl}{\ensuremath\Upsilon}
\newcommand{\FLstart}{\ensuremath\rhd}
\newcommand{\FLend}{\ensuremath\lhd}

\newcommand{\lett}{\ensuremath\mathit{letter}}
\newcommand{\labels}{\ensuremath\mathscr{L}}
\newcommand{\cl}{\ensuremath\mathit{co}}

\newcommand{\id}{\ensuremath\mathit{id}}
\newcommand{\oneend}{\ensuremath \#_1}
\newcommand{\twoend}{\ensuremath \#_2}
\newcommand{\inde}{\ensuremath\mathit{index}}

\newcommand{\typeTwo}{NES\nobreakdash-}
\newcommand{\TypeTwo}{NES-}
\newcommand{\MODT}{MOD\nobreakdash-3 }

\makeatletter
\newcommand{\parseNodeDistance}{%
{%
\def\tikz@lib@place@single@factor{1}%
\expandafter\tikz@lib@place@parse@nums\expandafter{\tikz@node@distance}%
\xdef\nodeDisX{\the\pgf@x}%
\xdef\nodeDisY{\the\pgf@y}%
}%
}
\makeatother

\definecolor{ganttGreen}{HTML}{2a8728}

\title{Global Winning Conditions in Synthesis of Distributed Systems with Causal Memory\\(Full Version)}
\titlerunning{Global Winning Conditions in Synthesis of Distributed Systems with Causal Memory}

\author{Bernd Finkbeiner}{CISPA Helmholtz Center for Information Security, Saarbrücken, Germany}{finkbeiner@cispa.de}{https://orcid.org/0000-0002-4280-8441}{}
\author{Manuel Gieseking}{University of Oldenburg, Oldenburg, Germany}{gieseking@informatik.uni-oldenburg.de}{https://orcid.org/0000-0001-9073-3002}{}
\author{Jesko Hecking-Harbusch}{CISPA Helmholtz Center for Information Security, Saarbrücken, Germany}{jesko.hecking-harbusch@de.bosch.com}{https://orcid.org/0000-0003-2076-617X}{}
\author{Ernst-Rüdiger Olderog}{University of Oldenburg, Oldenburg, Germany}{olderog@informatik.uni-oldenburg.de}{https://orcid.org/0000-0002-3600-2046}{}

\authorrunning{B.\ Finkbeiner, M.\ Gieseking, J.\ Hecking-Harbusch, E.-R.\ Olderog}

\Copyright{Bernd Finkbeiner, Manuel Gieseking, Jesko Hecking-Harbusch, Ernst-Rüdiger Olderog}

\ccsdesc{Theory of computation}

\keywords{Synthesis, distributed systems, reactive systems, Petri games, decidability}

\category{}

\relatedversion{A conference version of the paper is available~\cite{DBLP:conf/csl/FinkbeinerGHO22}.}

\supplement{}
\funding{Supported by the German Research Foundation (DFG) Grant Petri Games (392735815) and Collaborative Research Center “Foundations of Perspicuous Software Systems” (TRR 248, 389792660), and by the European Research Council (ERC) Grant OSARES (683300).}
\acknowledgements{}

\pdfoutput=1 
\nolinenumbers 

\hideLIPIcs  

\EventEditors{Florin Manea and Alex Simpson}
\EventNoEds{2}
\EventLongTitle{30th EACSL Annual Conference on Computer Science Logic (CSL 2022)}
\EventShortTitle{CSL 2022}
\EventAcronym{CSL}
\EventYear{2022}
\EventDate{February 14--19, 2022}
\EventLocation{G\"{o}ttingen, Germany (Virtual Conference)}
\EventLogo{}
\SeriesVolume{216}
\ArticleNo{36}

\begin{document}

\maketitle

\begin{abstract}
In the synthesis of distributed systems, we automate the development of distributed programs and hardware by automatically deriving correct implementations from formal specifications. 
For synchronous distributed systems, the synthesis problem is well known to be undecidable. 
For asynchronous systems, the boundary between decidable and undecidable synthesis problems is a long-standing open question. 
We study the problem in the setting of Petri games, a framework for distributed systems where asynchronous processes are equipped with causal memory. 
Petri games extend Petri nets with a distinction between system places and environment places. 
The components of a distributed system are the players of the game, represented as tokens that exchange information during each synchronization. 
Previous decidability results for this model are limited to local winning conditions, i.e., conditions that only refer to individual components. 

In this paper, we consider global winning conditions such as mutual exclusion, i.e., conditions that refer to the state of all components. 
We provide decidability and undecidability results for global winning conditions. 
First, we prove for winning conditions given as bad markings that it is decidable whether a winning strategy for the system players exists in Petri games with a bounded number of system players and one environment player. 
Second, we prove for winning conditions that refer to both good and bad markings that it is undecidable whether a winning strategy for the system players exists in Petri games with at least two system players and one environment player. 
Our results thus show that, on the one hand, it is indeed possible to use global safety specifications like mutual exclusion in the synthesis of distributed systems. 
However, on the other hand, adding global liveness specifications results in an undecidable synthesis problem for almost all Petri games.
\end{abstract}

\section{Introduction}\label{sec:introduction}

The \emph{synthesis problem} probes whether there exists an implementation for a formal specification and derives such an implementation if it exists.
This approach automates the creation of systems.
Engineers can think on a more abstract level about \emph{what} a system should achieve instead of \emph{how} the system should achieve its goal. 
The synthesis problem for a system consisting of one component interacting with its environment is often encoded as a two-player game with complete observation between the \emph{system} player and the \emph{environment} player (cf.\ \cite{DBLP:conf/cav/JobstmannGWB07, DBLP:conf/tacas/Ehlers11, DBLP:conf/cav/BohyBFJR12, DBLP:conf/tacas/FaymonvilleFRT17, LTLsynt, DBLP:journals/acta/LuttenbergerMS20, DBLP:journals/corr/abs-1904-07736}).
The system player tries to satisfy the winning condition of the game while the environment player tries to violate it.
A winning strategy for the system player is a correct implementation as it encodes the system's reaction to all environment behaviors.

The synthesis problem for \emph{distributed systems} aims to derive a correct implementation for every concurrent component of a distributed system. 
Each component can interact with its environment.
Distributed systems can be differentiated depending on whether the components progress \emph{synchronously} or \emph{asynchronously}.
For synchronous distributed systems, the synthesis problem is well known to be undecidable, as observed by Pnueli and Rosner~\cite{DBLP:conf/focs/PnueliR90}. 
For asynchronous distributed systems with causal memory, the boundary between decidable and undecidable synthesis problems is a long-standing open question~\cite{DBLP:conf/icalp/Muscholl15, DBLP:journals/iandc/FinkbeinerO17}.
For the synthesis of asynchronous distributed systems, the memory model changes compared to the two-player game and the number of players increases to encode the different components of the system.
In distributed systems, components observe only their local surroundings.
This can be encoded by \emph{causal memory}~\cite{DBLP:conf/fsttcs/GastinLZ04, DBLP:conf/fsttcs/MadhusudanTY05, DBLP:conf/icalp/GenestGMW13}:
Two players share no information while they run concurrently; during every synchronization, however, they exchange their entire local histories, including all of their previous synchronizations with other players.

In this paper, we consider \emph{reactive systems}, i.e., the components continually interact with their environment.
\emph{Control games}~\cite{DBLP:conf/icalp/GenestGMW13} based on \emph{asynchronous automata}~\cite{DBLP:journals/ita/Zielonka87} and \emph{Petri games}~\cite{DBLP:journals/iandc/FinkbeinerO17} are formalisms for the synthesis of asynchronous distributed reactive systems with causal memory.
We focus on Petri games.
Here, several system players play against several environment players in a Petri net.
Tokens represent players and places either belong to the system or to the environment, resulting in a distribution of system and environment players.
Deciding the existence of a winning strategy for the system players is EXPTIME-complete for Petri games with a bounded number of system players, one environment player, and bad places as local winning condition~\cite{DBLP:journals/iandc/FinkbeinerO17}.
This also holds for Petri games with a bounded number of environment players, one system player, and bad markings as global winning condition~\cite{DBLP:conf/fsttcs/FinkbeinerG17}.
Local winning conditions cannot express global properties like mutual exclusion.

We consider global winning conditions and contribute decidability and undecidability results regarding the synthesis of asynchronous distributed reactive systems with causal memory.
In the first part of this paper, we prove that it is decidable whether a winning strategy for the system players exists in Petri games with a bounded number of system players, one environment player, and \emph{bad markings} as global winning condition.
Bad markings are a \emph{safety} winning condition in the sense that they define markings as bad that the system players have to avoid in order to win the Petri game.
Decidability is achieved by a reduction to a two-player game with complete observation and a Büchi winning condition.
In the two-player game, it is encoded that transitions with the environment player fire as late as possible, i.e., transitions \emph{without} the environment player fire before transitions \emph{with} it.
This order of transitions encodes causal memory~\cite{DBLP:journals/iandc/FinkbeinerO17}.
For every sequential play of the two-player game, we need to check that no bad marking is reached for the different orders of fired concurrent transitions.
The causal history of system players can grow infinitely large.
We show that the finite causal history of each system player until its last synchronization with the environment player can be stored finitely and suffices to find bad markings.

In the second part of this paper, we investigate whether global winning conditions beyond bad markings are decidable.
We report on two undecidability results to further underline the significance of our decidability result:
We prove that it is undecidable whether a winning strategy exists for the system players in Petri games with at least two system players, one environment player, and \emph{good and bad markings} as winning condition.
For this winning condition, no bad marking should be reached \emph{until} a good marking is reached, which can be expressed in linear-time temporal logic (LTL)~\cite{DBLP:conf/focs/Pnueli77}.
Notice that it is not required to terminate in a good marking.
Good markings can be used to simulate the undecidable synchronous setting of Pnueli and Rosner~\cite{DBLP:conf/focs/PnueliR90} in the asynchronous setting of Petri games.
This is realized by identifying executions as good if players deviate too much from the synchronous setting.
Next, we prove that it is undecidable whether a winning strategy exists for the system players in  Petri games with \emph{good markings} and at least three players, out of which one is an environment player and each of the other two can change between being a system and an environment player.
Good markings are a \emph{liveness} winning condition in the sense that they define markings as good, one of which the system players have to reach in order to win the Petri game.
Here, bad markings from the first undecidability result are encoded by repeatedly changing all players to environment players. 
With these results, we obtain an overview regarding decidability and undecidability for global winning conditions. 

\subparagraph{Related Work}
A formal connection exists between Petri games and \emph{control games}~\cite{DBLP:conf/icalp/GenestGMW13} based on \emph{asynchronous automata}~\cite{DBLP:journals/ita/Zielonka87}:
Petri games can be translated into control games and vice versa, at an exponential blow-up in each direction~\cite{DBLP:conf/concur/BeutnerFH19}.
This translates decidability in acyclic communication architectures~\cite{DBLP:conf/icalp/GenestGMW13}, originally obtained for control games, to Petri games, and decidability in single-process systems~\cite{DBLP:conf/fsttcs/FinkbeinerG17}, originally obtained for Petri games, to control games.
Further decidability results exist for control games with acyclic communication architectures~\cite{DBLP:conf/fsttcs/MuschollW14}.
Decidability has also been obtained for restrictions on the dependencies of actions~\cite{DBLP:conf/fsttcs/GastinLZ04} or on the synchronization behavior~\cite{DBLP:conf/concur/MadhusudanT02, DBLP:conf/fsttcs/MadhusudanTY05} and for decomposable games~\cite{DBLP:conf/fsttcs/Gimbert17}.

The decidability result of this paper does not transfer to control games because the translation in~\cite{DBLP:conf/concur/BeutnerFH19} produces Petri games with as many system \emph{and} environment players as there are processes in the control game.
The undecidability results of this paper transfer to control games.
System players in Petri games correspond to processes with only controllable actions in control games; environment players correspond to processes with only uncontrollable actions~\cite{DBLP:conf/concur/BeutnerFH19}.
Bad markings from the first undecidability result can be simulated by additional uncontrollable actions for all processes preventing the reaching of good markings afterward.

For Petri games with several system and environment players, bounded synthesis is a semi-decision procedure to find winning strategies for the system players~\cite{DBLP:conf/birthday/Finkbeiner15, DBLP:journals/corr/Tentrup16, DBLP:conf/atva/Hecking-Harbusch19}.
Bounded synthesis and the reduction for bad places are implemented in the tool \textsc{AdamSYNT}~\cite{DBLP:conf/cav/FinkbeinerGO15, DBLP:journals/corr/abs-1711-10637, DBLP:conf/tacas/GiesekingHY21}.

\section{Motivating Example}\label{sec:BadMarkingsMotivation}

We introduce the intuition behind Petri games and bad markings with the example in \refFig{BMfigMotivation}.
There, we search for a strategy for two power plants, which should react to the energy production of renewable sources based on the weather forecast.
A Petri game differentiates the places of a Petri net as \emph{system} places (depicted in gray) and as \emph{environment} places (depicted in white).
For example, \emph{p} is a system place whereas \emph{forecast} is an environment place.
The players of a Petri game are represented by tokens.
The type of the place, where a token is residing, dictates whether the token represents a system or an environment player.

After transition \emph{sunny} fires to indicate a sunny forecast, there are two system players in place $p$ (each representing one power plant) and one environment player in place~$s$.
Causal memory implies that both system players know what the weather forecast predicts.
They do not know whether the actual energy production is high (indicated by $s_h$ firing) or low (indicated by $s_l$ firing) producing three or two units of energy in place $w$.
Nevertheless, each power plant has to decide whether to produce two or one unit of energy in place~$k$ by $p_h$ or $p_l$ firing.
The two power plants should produce together with the renewable sources either four or five units of energy.
Therefore, any final marking resulting in a different energy production is a bad marking, i.e.,
the set of bad markings is $\{ M : \pl \rightarrow \mathbb{N} \mid (M(k) + M(w) < 4 \lor\allowbreak M(k) + M(w) > 5) \land \exists x \in \{s', c', r'\} : (M(x) = 1 \land \forall y \in \pl \setminus\{x, k, w\} : M(y) = 0) \} $.
The second conjunct ensures the marking being final by requiring that the only environment player is in one of the three environment places~$s'$, $c'$, and~$r'$ and the system players are only in system places~$k$ and~$w$. 
We assume that players always choose one of their successors.
In \refSection{pg}, the finer notion of strategies being \emph{deadlock-avoiding} is presented.

A winning strategy for the system players produces one unit of energy at both power plants for a sunny forecast, two units of energy at one power plant and one unit of energy at the other for a cloudy forecast, and two units of energy at both power plants for a rainy forecast.
The specification is expressible with the local winning condition of bad places by having transitions from each bad marking leading to a bad place.
This is only~so~as the example has no infinite behavior.
For Petri games with infinite behavior and one environment player, the global winning condition of bad markings can specify losing behavior between players without requiring their synchronization which is impossible for local winning conditions.

\begin{figure}[t]
	\centering
	\begin{tikzpicture}[label distance=-0.5mm, scale=0.9, every node/.style={transform shape}] 
		\node [envplace] (e0) [tokens = 1, label=left:\emph{forecast}]{};
		\node [envplace] (e1) [right of=e0, right of=e0, label=below:$s$]{};
		\node [envplace] (e2) [right of=e1, right of=e1, label=right:$s'$]{};
		\node [sysplace] (s0) [right of=e2, right of=e2, label=above:$w$]{};
		\node [envplace] (e3) [right of=s0, right of=s0, above of=s0, label=355:$c$]{};
		\node [envplace] (e4) [right of=e3, right of=e3, yshift=-0.5cm, label=right:$c'$]{};
		\node [envplace] (e5) [right of=s0, right of=s0, below of=s0, label=5:$r$]{};
		\node [envplace] (e6) [right of=e5, right of=e5, yshift=0.5cm, label=right:$r'$]{};
		\node [sysplace] (s1) [right of=e4, right of=e4, yshift=0.5cm, label=below:$p$]{};
		\node [sysplace] (s2) [below of=s1, below of=s1, label=right:$k$]{};

		\node [transition] (t0) [right of = e0, label =below:\emph{sunny}]{};
		\path[->] (t0)
			edge [pre] (e0)
			edge [post] (e1);
		\node [transition] (t1) [above of = t0, label =left:\emph{cloudy}]{};
		\path[->] (t1)
			edge [pre] (e0)
			edge [post] (e3);
		\node [sysplace] (s7) [opacity=0, above of=t0]{};
		\draw[->]
 			(t1.north) -- ([yshift=2.5mm]s7.north) -- ([yshift=2.5mm]s1.north) -- node[right] {2} ([yshift=0.05cm]s1.north)
		;
		\node [transition] (t2) [below of = t0, label =left:\emph{rainy}]{};
		\path[->] (t2)
			edge [pre] (e0)
			edge [post] (e5);
		\node [transition] (t3) [right of = e1, yshift=0.5cm, label =175:$s_h$]{};
		\path[->] (t3)
			edge [pre] (e1)
			edge [post] (e2)
			edge [post] node[above, xshift=0.75cm, yshift=-0.175cm] {3} ([yshift=0.1cm]s0.west);
		\node [transition] (t4) [right of = e1, yshift=-0.5cm, label =185:$s_l$]{};
		\path[->] (t4)
			edge [pre] (e1)
			edge [post] (e2)
			edge [post] node[below, xshift=0.75cm, yshift=0.175cm] {2} ([yshift=-0.1cm]s0.west);

		\node [transition] (t5) [left of = e3, yshift=-0.5cm, label =355:$c_h$]{};
		\path[->] (t5)
			edge [pre] (e3)
			edge [post] (e4)
			edge [post] node[above] {2} (s0);
		\node [transition] (t6) [right of = e3, label =right:$c_l$]{};
		\path[->] (t6)
			edge [pre] (e3)
			edge [post] (e4);
		\draw [->] (t6) |- ([xshift=0.05cm]s0.east);

		\node [transition] (t7) [left of = e5, yshift=0.5cm, label =5:$r_h$]{};
		\path[->] (t7)
			edge [pre] (e5)
			edge [post] (e6)
			edge [post] (s0);
		\node [transition] (t8) [right of = e5, label =right:$r_l$]{};
		\path[->] (t8)
			edge [pre] (e5)
			edge [post] (e6);

		\node [transition] (t9) [below of=s1, xshift=-0.5cm, label=left:$p_h$]{};
		\path[->] (t9)
			edge [pre] (s1)
			edge [post] node[left] {2} (s2);
		\node [transition] (t10) [below of=s1, xshift=0.5cm, label=right:$p_l$]{};
		\path[->] (t10)
			edge [pre] (s1)
			edge [post] (s2);

		\node [sysplace] (s8) [opacity=0, above of=e1]{};
		\node [sysplace] (s9) [opacity=0, above of=t9]{};
		\draw[->]
 			(t0.north) -- ([xshift=0.1cm, yshift=0.1cm]s8.north) -- ([yshift=0.1cm]s9.north) -- node[below left] {2} ([xshift=-0.05cm]s1.north west)
		;
		\node [transition] (t11) [opacity=0, right of=s2]{};
		\draw[->]
 			(t2.south) -- ([yshift=-0.2cm]t2.south) -- ([xshift=0.25cm, yshift=-0.2cm]t11.south) |- node[above left, xshift=-0.25cm] {2} ([xshift=0.05cm]s1.east)
		;
	\end{tikzpicture}
	\caption{Two power plants observe if \emph{sunny}, \emph{cloudy}, or \emph{rainy} weather is \emph{forecast}.
	Depending on the actual weather, renewable sources produce up to three units of energy.
	The power plants produce one or two units of energy each and have to maintain the total energy production between four and five units of energy as all final markings with different energy production are bad markings.
	}
	\label{fig:BMfigMotivation}
\end{figure}
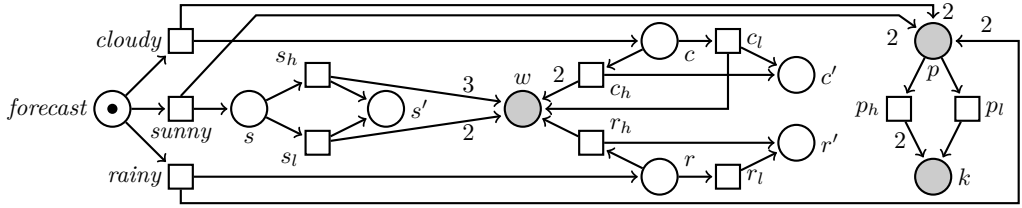

\section{Petri Nets}
\label{sec:pn}

A \emph{Petri net}~\cite{DBLP:journals/tcs/NielsenPW81, DBLP:books/sp/Reisig85a} $\petriNet$ consists of the disjoint finite sets of \emph{places} $\pl$ and of \emph{transitions}~$\tr$, the \emph{flow relation} $\fl$ as multiset over $(\pl \times \tr)\cup (\tr \times \pl)$, and the \emph{initial marking}~$\init$ as multiset over $\pl$.
For a place $p$, the \emph{precondition} is the set $\pre{}{p} = \{ t \in \tr \mid \flow(t, p) > 0 \}$ and the \emph{postcondition} is the set $\post{}{p} = \{ t \in \tr \mid \flow(p, t) > 0 \}$.
For a transition~$t$, the \emph{precondition} is the multiset over $\pl$ defined by $\pre{}{t}(p) = \flow(p, t)$ for all $p \in \pl $ and the \emph{postcondition} is the multiset over $\pl$ defined by $\post{}{t}(p) = \flow(t, p)$ for all $p \in \pl $.
States of Petri nets are represented by multisets over $\pl$, called \emph{markings}.
A marking~$M$ puts $M(p)$ \emph{tokens} in every place $p\in\pl$.
A transition $t$ is \emph{enabled} in a marking $M$ if $\pre{}{t} \subseteq M$. 
If no transition is enabled in a marking~$M$, then $M$ is called \emph{final}.
An enabled transition~$t$ can \emph{fire} in a marking~$M$ resulting in the successor marking $M' = M - \pre{}{t} + \post{}{t}$ (written $M \firable{t} M'$).
For markings~$M$ and~$M'$, we write $M \firable{t_0,\dotsc,t_{n-1}} M'$ if there exist markings $ M_0,\dotsc,M_n $ such that $M_0 = M$, $M_{n} = M'$, and $M_i \firable{t_i} M_{i+1}$ for $0 \leq i < n$.
The set of \emph{reachable markings} of $\pNet$ is defined as $\reach(\pNet) = \{M \mid \exists n \in \mathbb{N}, t_0,\dotsc, t_{n-1} \in \tr : \init \firable{t_0,\dotsc,t_{n-1}} M\}$.
A net $\pNet'$ is a \emph{subnet} of $\pNet$ (written $\pNet' \sqsubseteq \pNet$) if $\places' \subseteq \places$, $\transitions' \subseteq \transitions$, $\init' = \init$, and $\flow' = \flow \upharpoonright (\places' \times \transitions') \cup (\transitions' \times \places')$.
We enforce $\init' = \init$ to maintain all players when later defining strategies for Petri games.

We call elements in $\places \cup \transitions$ \emph{nodes}. 
For nodes $x$ and $y$, we write $x \lessdot y$ if $x\in\pre{}{y}$. 
With $\leq$, we denote the reflexive, transitive closure of $\lessdot$.
The \emph{causal past} of $x$ is $\past(x) = \{ y \mid y \leq x \}$.
Nodes $x$ and $y$ are \emph{causally related} if $x \leq y \lor y \leq x$.
They are \emph{in conflict} (written $\conflict{x}{y}$)~if,~for a place $p\in\pl$, there are distinct transitions $t_1, t_2 \in \post{}{p}$ with $t_1 \leq x \land t_2 \leq y$.
Node $x$ is in \emph{self-conflict} if $\conflict{x}{x}$.
We call $x$ and $y$ \emph{concurrent} if they are neither~causally~related~nor~in~conflict.

An \emph{occurrence net} is a Petri net where the pre- and postcondition of transitions are sets, the initial marking coincides with places without ingoing transitions,
other places have exactly one ingoing transition,
no infinite path starting from any given node and following the inverse flow relation exists,
and no transition is in self-conflict.
A \emph{homomorphism} maps nodes from $\pNet_1$ to $\pNet_2$ preserving the type of nodes and the pre- and postcondition of transitions.

A branching process \cite{DBLP:journals/acta/Engelfriet91, DBLP:journals/tcs/MeseguerMS96, DBLP:series/eatcs/EsparzaH08} describes parts of possible behaviors of a Petri net.
We use the \emph{individual token semantics}~\cite{DBLP:journals/iandc/GoltzR83}.
A \emph{branching process} of a Petri net $\pNet$ is a pair $\iota = (\pNet^\iota,\lambda^\iota)$ where $\pNet^\iota$ is an occurrence net and $\lambda^\iota :\pl^\iota\cup \tr^\iota\rightarrow\pl\cup \tr$ is a homomorphism from $\pNet^\iota$ to~$\pNet$ that is injective on transitions with the same precondition.
Intuitively, whenever a node can be reached on two distinct paths in a Petri net~$\pNet$, it is split up in the branching process of~$\pNet$.
$\lambda^\iota$ labels the nodes of $\pNet^\iota$ with the original nodes of $\pNet$.
The injectivity condition avoids additional unnecessary splits. 
The \emph{unfolding} $\unf =(\pNet^U, \lambda^U)$ of $\pNet$ is a maximal branching process: Whenever there is a set of pairwise concurrent places $C$ such that $\lambda^U[C] = \pre{\pNet}{t}$ for some transition $t \in \tr$, then there exists $t' \in \tr^{U}$ with $\lambda^U(t')=t$ and $\pre{U}{t'} = C$.

\section{Petri Games and Büchi Games}
\label{sec:pg}

A \emph{Petri game}~\cite{DBLP:journals/iandc/FinkbeinerO17, DBLP:conf/fsttcs/FinkbeinerG17} is a tuple $\petriGame$.
The places of the underlying Petri net $\petriNet$ are partitioned into \emph{system places}~$\plS$ and \emph{environment places}~$\plE$.
We call tokens on system places \emph{system players} and tokens on environment places \emph{environment players}.
The game is played by firing transitions in $\pNet$.
Players \emph{synchronize} when a joint transition fires.
Intuitively, a strategy controls the behavior of system players by deciding which transitions to allow.
Environment players are uncontrollable and transitions only dependent on environment players cannot be restricted. 
The \emph{winning condition} is given by~$\pSpecial$ as the set of \emph{bad places}~$\badplaces \subseteq \pl$, \emph{bad markings}~$\badmarkings \subseteq \reach(\pNet)$, or \emph{good markings} $\goodmarkings  \subseteq \reach(\pNet)$ or the pair of disjoint sets of \emph{good and bad markings} $(\goodmarkings,\badmarkings) \in \pom{\reach(\pNet)} \times \pom{\reach(\pNet)}$.
We depict Petri games as Petri nets and color system places gray and environment places white.
A Petri game has \emph{a bounded number of system players and one environment player} if, for a bound $k \in \mathbb{N}$, every system place contains at most $k$ tokens for all reachable markings of $\pNet$ and the sum of tokens in all environment places is exactly one for all reachable markings~of~$\pNet$.

A \emph{strategy} for $\pGame$ is a branching process $\sysstrat = (\pNet^\sysstrat, \lambda^\sysstrat)$ of $\pNet$ satisfying \emph{justified refusal}:
If there is a set of pairwise concurrent places $C$ in $\pNet^\sysstrat$ and a transition $t \in \tr$ with $\lambda^\sysstrat[C] = \pre{\pNet}{t}$, then there either is a transition $t'\in\tr^\sysstrat$ with $\lambda^\sysstrat(t')=t$ and $C = \pre{\sysstrat}{t'}$ or there is a system place $p \in C \cap (\lambda^\sysstrat)^{-1}[\plS]$ with $t \not\in \lambda^\sysstrat[\post{\sysstrat}{p}]$.
Justified refusal enforces that only system places can prohibit transitions based on their causal past:
From every situation in the game, a transition possible in the underlying net is either in the strategy or there is a system place that never allows it.
A strategy is a restriction of possible transitions in the Petri game because it is a branching process which describes subsets of the behavior of a Petri net.
We further require $\sysstrat$ to be \emph{deterministic}: For every reachable marking $M$ of $\sysstrat$ and system place $p \in M$, there is at most one transition from $\post{\sysstrat}{p}$ enabled in $M$.
Notice that $\post{\sysstrat}{p}$ can contain more than one transition as long as at most one of them is enabled in the same reachable marking.
This allows that the environment player decides between different branches of the Petri game and the system player later on reacts to every decision.

A strategy is \emph{deadlock-avoiding} if, for every final, reachable marking $M$ in the strategy, $\lambda^\sysstrat[M]$ is final as well. 
A strategy $\sysstrat$ is \emph{winning for bad places} $\pSpecial = \badplaces$ if it is deadlock-avoiding and no reachable marking in~$\sysstrat$ contains a place corresponding to a bad place.
A strategy $\sysstrat$ is \emph{winning for bad markings} $\pSpecial = \badmarkings$ if it is deadlock-avoiding and no reachable marking in~$\sysstrat$ corresponds to a bad marking.
One branching process $\brapro^1 = (\pNet^1, \lambda^1)$ is a \emph{subprocess} of another $\brapro^2 = (\pNet^2, \lambda^2)$ if $\pNet^1 \sqsubseteq \pNet^2$ and $\lambda^1 = \lambda^2 \upharpoonright (\pl^1 \cup \tr^1)$.
A \emph{play} $\pi=(\pNet^\pi,\lambda^\pi)$ is a subprocess of a strategy $\sysstrat=(\pNet^\sysstrat,\lambda^\sysstrat)$ with $\forall p \in \pl^\pi : |\post{}{p}| \leq 1$.
It is \emph{maximal} if, for each set of pairwise concurrent places~$C$ in $\pNet^\pi$ with $C = \pre{\sysstrat}{t}$ for some $t \in \tr^\sysstrat$, a place $p\in C $ and a transition $t' \in \tr^\pi$ exist with $t'\in\post{\pi}{p}$.
A \emph{complete firing sequence} of a play $\pi$ is a possibly infinite sequence of fired transitions such that each transition of $\pi$ occurs.
A strategy $\sysstrat$ is \emph{winning for good markings} $\pSpecial = \goodmarkings$~if, for all complete firing sequences $t_0 t_1 t_2 \ldots$ of all maximal plays~$\pi$ of $\sysstrat$ with $M_0 = \init^\pi$ and $M_0 \firable{t_0} M_1 \firable{t_1} M_2 \firable{t_2} \ldots$, there exists $i \geq 0$ with $\lambda^\pi[M_i] \in \goodmarkings$.
A strategy $\sysstrat$ is \emph{winning for good and bad markings} $\pSpecial = (\goodmarkings, \badmarkings) $ if, for all complete firing sequences $t_0 t_1 t_2 \ldots$ of all maximal plays $\pi$ of $\sysstrat$ with $M_0 = \init^\pi \land M_0 \firable{t_0} M_1 \firable{t_1} M_2 \firable{t_2} \ldots$, there exists $i \geq 0$ with $\lambda^\pi[M_i] \in \goodmarkings \land \forall 0 \leq j < i : \lambda^\pi[M_j]\notin \badmarkings$.
Terminating in a final marking as winning condition is different from reaching a good marking as players are not required to terminate in a good marking and can reach a bad marking afterward.

A \emph{Büchi game} has two players: \emph{Player~0} represents the \emph{system}, \emph{Player~1} the \emph{environment}.
Both act on complete information about the game arena and the play so far.
To win, Player~0 has to ensure that an accepting state is visited infinitely often.
A winning strategy for Player~0 corresponds to a correct implementation of the encoded synthesis problem.
Deciding the existence of a winning strategy can be done in polynomial time~\cite{DBLP:conf/soda/ChatterjeeH12}.

Formally, a \emph{Büchi game} $\BuchiGame$ consists of the finite set of \emph{states} $\TPstates$ partitioned into the disjoint sets of states~$\TPsysstates$ of Player~0 and of states~$\TPenvstates$ of Player~1, the \emph{initial state} $\TPinit \in \TPstates$, the \emph{edge relation} \mbox{$\TPedges \subseteq \TPstates \times \TPstates$}, and the set of \emph{accepting states} $\TPfinal \subseteq \TPstates$.
We assume that all states in a Büchi game have at least one outgoing edge.
A \emph{play} is a possibly infinite sequence of states which is constructed by letting Player~0 choose the next state from the successors in $\TPedges$ whenever the game is in a state from $\TPsysstates$ and by letting Player~1 choose otherwise.
An \emph{initial play} is a play that starts from the initial state.
A play is \emph{winning} for Player~0 if it visits at least one accepting state infinitely often.
Otherwise, the play is \emph{winning} for Player~1.
A \emph{strategy} for Player~0 is a function $ f : \TPstates^* \cdot \TPsysstates \rightarrow \TPstates $ that maps plays ending in states of Player~0 to one possible successor according to~$\TPedges$.
A play \emph{conforms} to a strategy $f$ if all successors of states in $\TPsysstates$ are chosen in accordance with~$f$.
A strategy~$f$ is winning for Player~0 if all initial plays that conform to $f$ are winning for Player~0.

\section{Decidability in Petri Games with Bad Markings}\label{sec:decPG}

We present a reduction from Petri games with a bounded number of system players, one environment player, and bad markings to Büchi games.
In the following, we give an intuition for the main concepts of the reduction, before presenting the structure of the Büchi game in the remainder of this section.
More details are in \refAppendix{decidability} and 
a running example is in \refFig{graphgame}.

Petri games use unfoldings, which can be of infinite size, to encode the causal memory of players.
By contrast, Büchi games have two players with complete information and a finite number of states.
To overcome these differences when encoding Petri games, states in the corresponding Büchi games consist of a representation of the current marking and some additional information.
Edges in the Büchi game mostly correspond to a transition firing in the Petri game.
We say that a transition fires in the Büchi game when it fires in the encoded Petri game.
Concurrency between transitions in the Petri game is encoded by having most possible interleavings in the Büchi game.
Some interleavings are left out to encode causal memory of the players in Petri games:
Causal memory is simulated in Büchi games by transitions with an environment place in their precondition firing as late as possible at \emph{mcuts}~\cite{DBLP:journals/iandc/FinkbeinerO17}.
An mcut is a situation in the Petri game where all system players have progressed maximally, i.e., the environment player can choose between all remaining possible transitions.
Mcuts can only be defined for Petri games with at most one environment player.
Player~1 in the Büchi game makes decisions only at states corresponding to mcuts.

We make two key additions to ensure that the idea of firing transition with the environment player only at mcuts can be lifted from local winning conditions to global winning conditions:
First, we add \emph{backward moves} to detect bad markings and nondeterministic decisions.
Intuitively, backward moves allow us to rewind transitions with only system players participating.
They are realized by each system player remembering its history until its last synchronization with the environment player.
In every state of the Büchi game, it is checked whether the backward moves of all system players allow us to rewind the game in such a way that a bad marking is reached or a nondeterministic decision is found.

Second, we add the so-called \emph{\typeTwo case} to handle system players playing infinitely without synchronizing with the environment player directly in the Büchi game.
The abbreviation NES stands for \emph{no more environment synchronization} and is necessary when some system players play infinitely but without synchronization with the environment player.
In~\cite{DBLP:journals/iandc/FinkbeinerO17}, this situation is called the type\nobreakdash-2 case and can be handled as a preprocessing step, because only the local winning condition of bad places is considered.
This is impossible for the global winning condition of bad markings considered in this paper.
Throughout this paper, the \typeTwo case can be disregarded by adding the restriction that each system player either terminates or synchronizes infinitely often with the environment player.

For the \typeTwo case, every system player has a three-valued flag.
As long as the system player will terminate or will synchronize with the environment player in the future, the flag should be set to \emph{negative \typeTwo status}.
When system players can play infinitely without synchronizing with the environment player, they should set their flags to \emph{positive \typeTwo status}.
After the \typeTwo case, participating system players obtain an \emph{ended \typeTwo status}, which excludes them from the remaining Büchi game.
A positive \typeTwo status triggers the \typeTwo case.
Here, the system players with positive \typeTwo status have to prove that they can play infinitely without synchronizing with the environment player.
Therefore, the usual order of all transitions without the environment player being possible until reaching an mcut is interrupted.
Instead, only system players with positive \typeTwo status are considered until their proof of playing infinitely without the environment player is successful.
If the system players with positive \typeTwo status make a mistake in their proof, then Player~0 immediately loses the Büchi game.

\subsection{States and Initial State in the Büchi Game}

Decision tuples represent players of the Petri game in states in the Büchi game.
A \emph{decision tuple} for a player consists of an \emph{identifier}, a \emph{position}, a \emph{\typeTwo status}, a \emph{decision}, and a \emph{representation of the last mcut}.
The identifier uniquely determines the player.
The position gives the current place of the player.
System players with negative \typeTwo status $\False$ claim that they will terminate or fire a transition with an environment place in its precondition and are not part of the \typeTwo case.
In the \typeTwo case, system players go from positive \typeTwo status $\True$ to ended \typeTwo status $\End$ as described previously.

The decision is either $\top$ or the set of \emph{allowed transitions} by the player.
For system players, $\top$ indicates that a decision for a set of allowed transitions is missing and has to be chosen.
The representation of the last mcut encodes the last known position of the environment player.
There can be at most as many different such positions as there are system players.
Thus, a number suffices to identify the last known mcut.
Let $\maxSys$ be the maximal number of system players in the Petri game  which are visible at the same time.
The set of \emph{system decision tuples} is $\decisionsets_S = \{ (\id, p, b, T, K) \mid \id, K \in \{1, \ldots, \maxSys\} \land p \in \plS \land b \in \{\False, \True, \End\} \land (T = \top \lor T \subseteq \post{}{p} ) \} $,
the set of \emph{environment decision tuples} is $\decisionsets_E = \{ (0, p, \False, \post{}{p}, 0) \mid p \in \plE \}$, and
the set of all \emph{decision tuples} is $\decisionsets = \decisionsets_S \cup \decisionsets_E$.

\begin{example}
In \refFig{graphgame}, a branch of the Büchi game for the Petri game in \refFig{BMfigMotivation} is shown.
States with decision tuples with positive \typeTwo status are omitted because no infinite behavior occurs.
The initial state $v_0$ has one decision tuple for the environment player in place $\mathit{forecast}$ and empty information for the \typeTwo case and the backward moves.
After Player~1 plays the edge for transition $\mathit{sunny}$ firing, state $v_1$ with three decision tuples is reached.
The decision tuples for the two system players in place $p$ have $\top$ as decision.
There are 16 combinations of decisions by the two system players, out of which four are shown.
The first system player always allows transition~$p_l$ and the second system player allows no transition in~$v_2$, only one of the two transitions $p_l$ and $p_h$ in $v_6$ and $v_8$, or both transitions in~$v_7$.
\end{example}

\begin{figure}[t]
\centering
\begin{tikzpicture}[scale=0.7, every node/.style={transform shape}, 
sys/.style={
rectangle,
very thick,
fill=lightgray,
draw,
align=center,
minimum size=7mm,
},
env/.style={
rectangle,
very thick,
fill=white,
draw,
align=center,
minimum size=7mm,
},
node distance=7mm and 2mm
]
\node[env, accepting, label={[label distance=-1mm, xshift=-.3mm, yshift=2mm]left:\small\(v_0\)}]	(0) {\small$(\{(0, \mathit{forecast}, \False, \{\mathit{sunny},\mathit{cloudy},\mathit{rainy}\}, 0)\},\emptyset, []^7)$};
\draw[<-,>=stealth'] (0) -- node[above] {} ++(-4.25cm,0); 
\node[right of=0, right of=0, right of=0, right of=0, right of=0, right of=0, right of=0, yshift=0.25cm]	(000) {\small$\ldots$};
\node[right of=0, right of=0, right of=0, right of=0, right of=0, right of=0, right of=0, yshift=-0.25cm]	(00) {\small$\ldots$};

\node[sys, below of=0, below of=0, below of=0, below of=0, left of=0, left of=0, left of=0, label={[label distance=-1mm, xshift=4mm, yshift=8mm]left:\small\(v_1\)}]	(1)	{\small$(\{\textcolor{blue}{(0, s, \False, \{s_h, s_l\}, 0),}$\\[-0.1cm] \small$\textcolor{blue}{(1, p, \False, \top, 1),}$\\[-0.1cm] \small$\textcolor{blue}{(2, p, \False, \top, 1)}\}, \emptyset, []^7)$};

\node[left of=1, left of=1, left of=1, left of=1, yshift=0.5cm]	(111) {\small$\ldots$};
\node[left of=1, left of=1, left of=1, left of=1, yshift=-0.5cm]	(1111) {\small$\ldots$};
\node[left of=1, left of=1, left of=1]	(11111) {\small$\ldots$};

\node[sys, right of=1, right of=1, right of=1, right of=1, right of=1, right of=1, label={[label distance=-1mm, xshift=4mm, yshift=8mm]left:\small\(v_2\)}]	(2)	{\small$(\{(0, s, \False, \{s_h, s_l\}, 0),$\\[-0.1cm] \small$(1, p, \False, \textcolor{blue}{\{p_l\}}, 1),$\\[-0.1cm] \small$(2, p, \False, \textcolor{blue}{\emptyset}, 1)\}, \emptyset, []^7)$};
\node[env, accepting, right of=2, right of=2, right of=2, right of=2, right of=2, right of=2, right of=2, right of=2, right of=2, xshift=0.16cm, label={[label distance=-1mm, xshift=4mm, yshift=8.5mm]left:\small\(v_3\)}]	(22)	{\small$(\{(0, s, \False, \{s_h, s_l\}, 0),\textcolor{blue}{(1, k, \False, \emptyset, 1)},$\\[-0.1cm] \small$(2, p, \False, \emptyset, 1)\}, \emptyset, \textcolor{ganttGreen}{ [ ( \{ (1, p, \False, \{p_l\}, 1 ) \},}$\\[-0.1cm] \small$\textcolor{ganttGreen}{\{ (1, k, \False, \emptyset, 1) \} ) ] }, []^{\textcolor{ganttGreen}{6}})$};
\node[sys, above of=22, above of=22, above of=22, right of=22, xshift=0.75cm, yshift=0.15cm, label={[label distance=-1mm, xshift=0mm, yshift=7mm]left:\small\(v_{5}\)}]	(222)	{\small$(\{\textcolor{blue}{(0, s', \False, \emptyset, 0)}, (1, k, \False, \emptyset, 1),$\\[-0.1cm] \small$(2, p, \False, \emptyset, 1), \textcolor{blue}{(3, w, \False, \emptyset, 2)}, $\\[-0.1cm] \small$\textcolor{blue}{(4, w, \False, \emptyset, 2)}, \textcolor{blue}{(5, w, \False, \emptyset, 2)}\}, \emptyset, $\\[-0.1cm] \small$[ ( \{ (1, p, \False, \{p_l\}, 1 ) \}, \{ (1, k, \False, \emptyset, 1) \} ) ], []^6)$};
\node[sys, above of=22, above of=22, yshift=0.2cm, left of=22, left of=22, left of=22, left of=22, left of=22, left of=22, left of=22, left of=22, left of=22, left of=22, left of=22, left of=22, xshift=-0.05cm, label={[label distance=-1mm, yshift=4mm]left:\small\(v_4\)}]	(2222)	{\small$(\{\textcolor{blue}{(0, s', \False, \emptyset, 0)}, (1, k, \False, \emptyset, 1), (2, p, \False, \emptyset, 1), \textcolor{blue}{(3, w, \False, \emptyset, 2)}, $\\[-0.1cm] \small$\textcolor{blue}{(4, w, \False, \emptyset, 2)} \}, \emptyset, [ ( \{ (1, p, \False, \{p_l\}, 1 ) \},\{ (1, k, \False, \emptyset, 1) \} ) ], []^6)$};
\node[sys, below of=1, below of=1, yshift=-0.4cm, xshift=0.145cm, label={[label distance=-1mm, xshift=4mm, yshift=8mm]left:\small\(v_6\)}]	(4)	{\small$(\{(0, s, \False, \{s_h, s_l\}, 0),$\\[-0.1cm] \small$(1, p, \False, \textcolor{blue}{\{p_l\}}, 1),$\\[-0.1cm] \small$(2, p, \False, \textcolor{blue}{\{p_l\}}, 1)\}, \emptyset, []^7)$};
\node[sys, right of=4, right of=4, right of=4, right of=4, right of=4, right of=4, right of=4, label={[label distance=-1mm, xshift=4mm, yshift=8mm]left:\small\(v_7\)}]	(5)	{\small$(\{(0, s, \False, \{s_h, s_l\}, 0),$\\[-0.1cm] \small$(1, p, \False, \textcolor{blue}{\{p_l\}}, 1),$\\[-0.1cm] \small$(2, p, \False, \textcolor{blue}{\{p_h, p_l\}}, 1)\}, \emptyset, []^7)$};
\node[sys, right of=5, right of=5, right of=5, right of=5, right of=5, right of=5, right of=5, xshift=0.045cm, label={[label distance=-1mm, xshift=4mm, yshift=8mm]left:\small\(v_8\)}]	(3)	{\small$(\{(0, s, \False, \{s_h, s_l\}, 0),$\\[-0.1cm] \small$(1, p, \False, \textcolor{blue}{\{p_l\}}, 1),$\\[-0.1cm] \small$(2, p, \False, \textcolor{blue}{\{p_h\}}, 1)\}, \emptyset, []^7)$};
\node[sys, below of=4, below of=4, below of=4, left of=4, left of=4, label={[label distance=-1mm, xshift=4mm, yshift=11.75mm]left:\small\(v_9\)}]	(7)	{\small$(\{(0, s, \False, \{s_h, s_l\}, 0),$\\[-0.1cm]\small$\textcolor{blue}{(1, k, \False, \emptyset, 1)},$\\[-0.1cm] \small$(2, p, \False, \{p_l\}, 1)\}, \emptyset,$\\[-0.1cm]\small$\textcolor{ganttGreen}{ [( \{ (1, p, \False, \{p_l\}, 1) \},}$\\[-0.1cm]\small$\textcolor{ganttGreen}{\{ (1, k, \False, \emptyset, 1) \} ) ] },[]^{\textcolor{ganttGreen}{6}})$};
\node[sys, right of=7, right of=7, right of=7, right of=7, right of=7, right of=7, label={[label distance=-1mm, xshift=5mm, yshift=11.75mm]left:\small\(v_{10}\)}]	(8)	{\small$(\{(0, s, \False, \{s_h, s_l\}, 0),$\\[-0.1cm]\small$(1, p, \False, \{p_l\}, 1),$\\[-0.1cm] \small$\textcolor{blue}{(2, k, \False, \emptyset, 1)}\}, \emptyset, $\\[-0.1cm]\small$\textcolor{ganttGreen}{[( \{ (2, p, \False, \{p_l\}, 1) \},}$\\[-0.1cm]\small$\textcolor{ganttGreen}{\{ (2, k, \False, \emptyset, 1) \} )]}, []^{\textcolor{ganttGreen}{6}})$};
\node[sys, right of=8, right of=8, right of=8, right of=8, right of=8, right of=8, label={[label distance=-1mm, xshift=5mm, yshift=11.75mm]left:\small\(v_{11}\)}]	(9)	{\small$(\{(0, s, \False, \{s_h, s_l\}, 0),$\\[-0.1cm] \small$\textcolor{blue}{(1, k, \False, \emptyset, 1)},$\\[-0.1cm] \small$(2, p, \False, \{p_h\}, 1)\}, \emptyset, $\\[-0.1cm]\small$\textcolor{ganttGreen}{ [ ( \{ (1, p, \False, \{p_l\}, 1) \}},$\\[-0.1cm]\small$\textcolor{ganttGreen}{\{ (1, k, \False, \emptyset, 1) \} ) ] }, []^{\textcolor{ganttGreen}{6}})$};
\node[sys, right of=9, right of=9, right of=9, right of=9, right of=9, right of=9, right of=9, xshift=0.5cm, label={[label distance=-1mm, xshift=5mm, yshift=11.75mm]left:\small\(v_{12}\)}]	(10)	{\small$(\{(0, s, \False, \{s_h, s_l\}, 0),$\\[-0.1cm] \small$(1, p, \False, \{p_l\}, 1),$\\[-0.1cm] \small$\textcolor{blue}{(2, k, \False, \emptyset, 1), (3, k, \False, \emptyset, 1)}\}, \emptyset,$\\[-0.1cm]\small$\textcolor{ganttGreen}{ [ ( \{ (2, p, \False, \{p_h\}, 1) \},}$\\[-0.1cm]\small$\textcolor{ganttGreen}{\{ (2, k, \False, \emptyset, 1), (3, k, \False, \emptyset, 1) \} ) ] }, []^{\textcolor{ganttGreen}{6}})$};
\node[env, accepting, below of=7, below of=7, below of=7, right of=7, right of=7, right of=7, xshift=0.37cm, label={[label distance=-1mm, xshift=5mm, yshift=8.5mm]left:\small\(v_{13}\)}]	(11)	{\small$(\{(0, s, \False, \{s_h, s_l\}, 0), \textcolor{blue}{(1, k, \False, \emptyset, 1)}, \textcolor{blue}{(2, k, \False, \emptyset, 1)}\}, $\\[-0.1cm]\small$\emptyset, \textcolor{ganttGreen}{ [ ( \{ (1, p, \False, \{p_l\}, 1) \}, \{ (1, k, \False, \emptyset, 1) \} ) ] } ,$\\[-0.1cm]\small$\textcolor{ganttGreen}{ [ ( \{ (2, p, \False, \{p_l\}, 1) \}, \{ (2, k, \False, \emptyset, 1) \} ) ] }, []^{\textcolor{ganttGreen}{5}})$};
\node[env, accepting, right of=11, right of=11, right of=11, right of=11, right of=11, right of=11, right of=11, right of=11, right of=11, right of=11, right of=11, right of=11, right of=11, label={[label distance=-1mm, xshift=5mm, yshift=8.5mm]left:\small\(v_{14}\)}]	(12)	{\small$(\{(0, s, \False, \{s_h, s_l\}, 0), \textcolor{blue}{(1, k, \False, \emptyset, 1)}, \textcolor{blue}{(2, k, \False, \emptyset, 1), } $\\[-0.1cm]\small$\textcolor{blue}{(3, k, \False, \emptyset, 1)}\}, \emptyset,  \textcolor{ganttGreen}{ [ ( \{ (1, p, \False, \{p_l\}, 1) \}, \{ (1, k, \False, \emptyset, 1) \} ) ] } ,$\\[-0.1cm]\small$\textcolor{ganttGreen}{ [ ( \{ (2, p, \False, \{p_h\}, 1) \}, \{ (2, k, \False, \emptyset, 1), (3, k, \False, \emptyset, 1) \} ) ] } , []^{\textcolor{ganttGreen}{5}})$};
\node[env, accepting, below of=11, below of=11, below of=11, yshift=-0.4cm, left of=11, left of=11, left of=11, label={[label distance=-1mm, xshift=5mm, yshift=16mm]left:\small\(v_{15}\)}] (15) {\small$(\{\textcolor{blue}{(0, s', \False, \emptyset, 0)},$\\[-0.1cm]\small$(1, k, \False, \emptyset, 1),$\\[-0.1cm]\small$(2, k, \False, \emptyset, 1),$\\[-0.1cm]\small$\textcolor{blue}{(3, w, \False, \emptyset, 2)},$\\[-0.1cm]\small$\textcolor{blue}{(4, w, \False, \emptyset, 2)},$\\[-0.1cm]\small$\textcolor{blue}{(5, w, \False, \emptyset, 2)}\}, $\\[-0.1cm]\small$\emptyset, \langle$\emph{BM} as in $v_{13}\rangle )$};
\node[env, accepting, right of=15, right of=15, right of=15, right of=15, right of=15, yshift=0.2cm, label={[label distance=-1mm, xshift=5mm, yshift=14mm]left:\small\(v_{16}\)}] (16) {\small$(\{\textcolor{blue}{(0, s', \False, \emptyset, 0)},$\\[-0.1cm]\small$(1, k, \False, \emptyset, 1),$\\[-0.1cm]\small$(2, k, \False, \emptyset, 1),$\\[-0.1cm]\small$\textcolor{blue}{(3, w, \False, \emptyset, 2)},$\\[-0.1cm]\small$\textcolor{blue}{(4, w, \False, \emptyset, 2)}\}, $\\[-0.1cm]\small $\emptyset,\langle$\emph{BM} as in $v_{13}\rangle)$};
\node[env, accepting, right of=16, right of=16, right of=16, right of=16, right of=16, right of=16, xshift=0.5cm, yshift=0.35cm, label={[label distance=-1mm, xshift=5mm, yshift=10.5mm]left:\small\(v_{17}\)}] (17) {\small$(\{(\textcolor{blue}{(0, s', \False, \emptyset, 0)}, (1, k, \False, \emptyset, 1),$\\[-0.1cm]\small$(2, k, \False, \emptyset, 1), (3, k, \False, \emptyset, 1),$\\[-0.1cm]\small$ \textcolor{blue}{(4, w, \False, \emptyset, 2)}, \textcolor{blue}{(5, w, \False, \emptyset, 2)}\}, $\\[-0.1cm]\small$\emptyset,\langle$\emph{BM} as in $v_{14}\rangle)$};
\node[env, accepting, right of=17, right of=17, right of=17, right of=17, right of=17, right of=17, right of=17, right of=17, yshift=-0.2cm, label={[label distance=-1mm, xshift=5mm, yshift=12mm]left:\small\(v_{18}\)}] (18) {\small$(\{\textcolor{blue}{(0, s', \False, \emptyset, 0)}, (1, k, \False, \emptyset, 1),$\\[-0.1cm]\small$(2, k, \False, \emptyset, 1), (3, k, \False, \emptyset, 1),$\\[-0.1cm]\small$ \textcolor{blue}{(4, w, \False, \emptyset, 2)}, \textcolor{blue}{(5, w, \False, \emptyset, 2)}, $\\[-0.1cm]\small$\textcolor{blue}{(6, w, \False, \emptyset, 2)}\}, $\\[-0.1cm]\small$\emptyset, \langle$\emph{BM} as in $v_{14}\rangle)$};

\node[env, above of=22, above of=22, left of=22, left of=22, left of=22, left of=22]	(14) {\small$F_N$};
\node[env, accepting, below of=17, below of=17, yshift=-0.2cm]	(13) {\small$F_B$};

\draw[->, >=stealth']
   (0) edge (00)
   (0) edge (000)
   (1) edge (111)
   (1) edge (1111)
   (1) edge (2)
   (1) edge (4)
   (1.310) edge (5.north west)
   (2) edge ([xshift=-0.1cm]22.west)
   (4) edge (7)
   (4) edge (8.north west)
   (3) edge (9.north east)
   (3) edge (10)
   (7) edge ([yshift=0.1cm]11.120)
   (8) edge ([yshift=0.1cm]11.60)
   (9) edge ([yshift=0.1cm]12.120)
   (10) edge ([yshift=0.1cm]12.60)
   (11) edge ([yshift=0.1cm]15.north)
   (11) edge ([yshift=0.1cm]16.north)
   (12) edge ([yshift=0.1cm]17.north)
   (12) edge ([yshift=0.1cm]18.north)
   (22) edge (222)
   (22.175) -- (2222.south east)
   (222) edge (14)
   (2222) edge (14)
   (17) edge (13)
   (13) edge [loop right] (13)
   (14) edge [loop above] (14)
;
\draw[->, >=stealth']
 (5.35) -- (14)
;
\draw[->, >=stealth']
 ([yshift=0.11cm]15.south east) -- ([xshift=-0.1cm, yshift=-0.225cm]13.west)
;
\draw[->, >=stealth']
 ([yshift=0.115cm]16.south east) -- ([xshift=-0.1cm, yshift=-0.225cm]13.north west)
;
\draw[->, >=stealth']
 (0.south) |- ([xshift=-52mm, yshift=-2mm]0.south) |- ([yshift=3mm]1.north) -- (1.north)
;
\draw[->, >=stealth']
 ([xshift=0.2mm]1.340) |- ([yshift=3mm]3.north) -- (3.north)
;
\draw[->, >=stealth']
 (18.east) -- ([xshift=2.75mm]18.east) |- ([xshift=10mm]14.south east) -- (14)
;
\end{tikzpicture}
\caption{Part of the Büchi game for the Petri game in \refFig{BMfigMotivation} is given. States of Player~0 are gray, states of Player~1 white. Most states are labeled for identification. Double squares are accepting states. Changes from previous states are blue for decision tuples and green for backward moves.
}
\label{fig:graphgame}
\end{figure}

Almost all states in the Büchi game contain decision tuples and additional information for the \typeTwo case and for backward moves.
The \emph{states in the Büchi game} are defined as $\TPstates = \TPstates_\mathit{BN} \cup \decsetgame \times (\plS \rightarrow \{0, \ldots, k\}) \times (\backwardrules^*)^{\maxSys}$ with $\TPstates_\mathit{BN} = \{F_B, F_N\}$.
Finite winning and losing behavior in the Petri game is represented in the Büchi game by the two unique states $F_B$ and $F_N$ in $\TPstates_\mathit{BN}$.
A \emph{decision marking} is a set of decision tuples corresponding to a reachable marking in the Petri game such that each identifier occurs at most once.
$\decsetgame$ is the set of all such decision markings.
The next element stores the underlying multiset over system places of the decision marking from the start of the \typeTwo case restricted to system players with positive \typeTwo status.
In the \typeTwo case, repeating this multiset proves that the system players with positive \typeTwo status can play infinitely without firing a transition with an environment place in its precondition.
This element is the empty multiset if not in the \typeTwo case.
More details are in \refSection{typetwo}.
$\backwardrules : \pom{\decisionsets_S} \times \pom{\decisionsets_S}$ is the set of backward moves to detect states corresponding to a bad marking or a nondeterministic decision.
The remaining elements are $\maxSys$ sequences of backward moves.
Each identifier in a decision tuple maps to the position of a sequence of backward moves.
More details are in \refSection{backwardrules}.

The \emph{initial state in the Büchi game} has as many decision tuples with unique identifier, \typeTwo status $\False$, $\top$ as decision, and last mcut~$1$ as there are tokens in system places in $\init$ of the Petri game and one decision tuple with identifier~$0$, \typeTwo status~$\False$, the postcondition of $p_E$ as decision, and last mcut~$0$ for the one environment place $p_E$ with one token in~$\init$.
The other parts are the empty multiset or the empty sequence of backward moves.

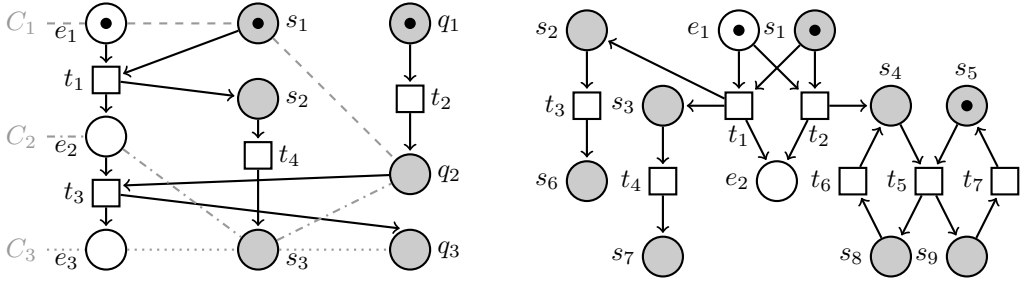
\begin{figure}[t]
	\centering
	\begin{subfigure}[t]{0.49\textwidth}
	\centering
	\begin{tikzpicture}[label distance=-0.5mm, every node/.style={transform shape}]
		\node [envplace] (e1) [tokens=1, label={[yshift=-1.5mm]left:{$e_1$}}] {};
		\node [sysplace] (s1) [tokens=1, right of=e1, right of=e1, label=right:{$s_1$}] {};
		\node [sysplace] (q1) [tokens=1, right of=s1, right of=s1, label=right:{$q_1$}] {};
		\node [envplace] (e2) [yshift=-1.5cm, label={[yshift=-1.5mm]left:{$e_2$}}] {};
		\node [sysplace] (s2) [below of=s1, label=right:{$s_2$}] {};
		\node [sysplace] (q2) [below of=q1, below of=q1, label=right:{$q_2$}] {};
		\node [envplace] (e3) [below of=e1, below of=e1, below of=e1, label={[yshift=-1.5mm]left:{$e_3$}}] {};
		\node [sysplace] (s3) [below of=s2, below of=s2, label={[yshift=-1.5mm]right:{$s_3$}}] {};
		\node [sysplace] (q3) [below of=q2, label=right:{$q_3$}] {};
		\node [transition] 	(t1) at ($(e1)!0.5!(e2)$) [label=left:{$t_1$}] {}
			edge [pre]  (e1)
			edge [pre]  (s1)
			edge [post] (e2)
			edge [post] (s2);
		\node [transition] 	(t2)  [below of = q1, label=right:{$t_2$}] {}
			edge [pre]  (q1)
			edge [post] (q2);
		\node [transition] 	(t3) at ($(e2)!0.5!(e3)$) [label=left:{$t_3$}] {}
			edge [pre]  (e2)
			edge [post] (e3)
			edge [post] ([xshift=-0.1cm]q3.north);
		\draw[thick, shorten >=1pt, ->] (q2) -- ([yshift=0.1cm]t3.east);
		\node [transition] 	(t4)  [below of = s2, yshift=0.25cm, label=right:{$t_4$}] {}
			edge [pre]  (s2)
			edge [post] (s3);
		\node (m1) [left=of e1,xshift=5mm] {\color{gray!80}$C_1$};
		\node (m2) [left=of e2,xshift=5mm] {\color{gray!80}$C_2$};
		\node (m3) [left=of e3,xshift=5mm] {\color{gray!80}$C_3$};
		\draw[-,gray!80,dashed] (m1) -- (e1) -- (s1) -- (q2);
		\draw[-,gray!80,dash dot] (m2) -- (e2) -- (s3) -- (q2);
		\draw[-,gray!80,dotted] (m3) -- (e3) -- (s3) -- (q3);
	\end{tikzpicture}
	\subcaption{The mcuts of the unfolding are $C_1$, $C_2$, and $C_3$.}
	\label{fig:mcut}
	\end{subfigure}%
	~
	\begin{subfigure}[t]{0.49\textwidth}
	\centering
	\begin{tikzpicture}[label distance=-0.5mm, every node/.style={transform shape}]
		\node [envplace] (e1) [tokens=1, label=left:{$e_1$}] {};
		\node [sysplace] (s1) [tokens=1, right of = e1, label=left:{$s_1$}] {};
		\node [envplace] (e2) [xshift=0.5cm, below of=e1, below of=e1, label=left:{$e_2$}] {};
		\node [sysplace] (s2) [below of =e1, left of = e1, label=left:{$s_3$}] {};
		\node [sysplace] (s22) [left of = e1, left of=e1, label=left:{$s_2$}] {};
		\node [sysplace] (s222) [below of = s2, below of = s2, label=left:{$s_7$}] {};
		\node [sysplace] (s2222) [below of = s22, below of = s22, label=left:{$s_6$}] {};
		\node [transition] 	(t22)  [below of = s2, label=left:{$t_4$}] {}
			edge [pre]  (s2)
			edge [post] (s222);
		\node [transition] 	(t222)  [below of = s22, label=left:{$t_3$}] {}
			edge [pre]  (s22)
			edge [post] (s2222);
		\node [sysplace] (s3) [below of=s1, right of=s1, label=above:{$s_4$}] {};
		\node [sysplace] (s4) [right of = s3, tokens=1, label=above:{$s_5$}] {};
		\node [sysplace] (s5) [below of = s3, below of = s3, label=left:{$s_8$}] {};
		\node [sysplace] (s6) [below of = s4, below of = s4, label=left:{$s_9$}] {};
		\node [transition] 	(t2)  [below of = e1, label=below:{$t_1$}] {}
			edge [pre]  (s1)
			edge [pre]  (e1)
			edge [post] (s2)
			edge [post] (s22)
			edge [post] (e2);
		\node [transition] 	(t3)  [below of = s1, label=below:{$~t_2$}] {}
			edge [pre]  (s1)
			edge [pre]  (e1)
			edge [post] (s3)
			edge [post] (e2);
		\node [transition] 	(t4)  [below of = s3, xshift=0.5cm, label=left:{$t_5$}] {}
			edge [pre]  (s3)
			edge [pre]  (s4)
			edge [post] (s5)
			edge [post] (s6);
		\node [transition] 	(t5)  [left of = t4, label=left:{$t_6$}] {}
			edge [pre]  (s5)
			edge [post] (s3);
		\node [transition] 	(t6)  [right of = t4, label=left:{$t_7$}] {}
			edge [pre]  (s6)
			edge [post] (s4);
	\end{tikzpicture}
	\subcaption{Markings containing $s_2$ and $s_7$ are bad markings.}
	\label{fig:type2}
	\end{subfigure}
	\caption{Two Petri games illustrate mcuts, backward moves, and the \typeTwo case.}
	\label{fig:mcutAndType2}
\end{figure}

\subsection{States of Player~0, States of Player~1, and Accepting States}

Causal memory in Petri games is encoded in Büchi games by letting Player~0 fix the decisions of allowed transitions for system players as early as possible and having Player~1 fire transitions with an environment place in their precondition as late as possible at mcuts.
\emph{Cuts} are markings in unfoldings.
An \emph{mcut} is a cut where all enabled transitions have an environment place in their precondition, i.e., all system players progressed maximally on their own.
With \refFig{mcut}, we illustrate mcuts.
The initial cut $\{e_1,s_1,q_1\}$ is \emph{not} an mcut as the enabled transition~$t_2$ has only the system place~$q_1$ in its precondition.
After $t_2$ fires, the cut $C_1 = \{e_1,s_1,q_2\}$ is an mcut as the only enabled transition~$t_1$ has environment place~$e_1$ in its precondition.
Analog arguments lead to $\{e_2,s_2,q_2\}$ not being an mcut and $C_2 = \{e_2,s_3,q_2\}$ being an mcut.
The final cut $C_3 = \{e_3,s_3,q_3\}$ is an mcut as there are no enabled transitions.

A decision marking~$\bbD$ in the states in the Büchi game \emph{corresponds to an mcut} when no~$\top$ and no positive \typeTwo status are part of $\bbD$ and every transition with only system places in its precondition is not enabled or not allowed by a participating system player in~$\bbD$. 
A state in the Büchi game can correspond to an mcut although the cut in the unfolding of the Petri game is not an mcut as the decisions of the system players in the Büchi game can disallow transitions.
\emph{States of Player~1} are $F_B$, $F_N$, and states corresponding to an mcut.
\emph{States of Player~0} are all other states.
\emph{Accepting states} are $F_B$ and states corresponding to an mcut.

\begin{example}
	The Petri game from \refFig{BMfigMotivation} has $\{\mathit{forecast}\}$, $\{ \{e, k:i\} \mid  e \in \{s, c, r\} \land 2 \leq i \leq 4 \}$, and $\{ \{s', w:w_{s'}, k:i\}, \{c', w:w_{c'}, k:i\}, \{r', w:w_{r'}, k:i\} \mid 2 \leq w_{s'}\leq 3 \land 1 \leq w_{c'}\leq 2 \land 0 \leq w_{r'}\leq 1 \land 2 \leq i \leq 4  \}$ as mcuts,
	i.e., the initial cut, cuts where the power plants produced energy while the energy production by renewable sources was not selected, and all final, reachable cuts.
	In the Büchi game in \refFig{graphgame}, the eight states $v_0$, $v_3$, and $v_{13}$ to $v_{18}$ of Player~1 have decision markings that correspond to an mcut.
	For states $v_0$, $v_{13}$, and $v_{14}$, all enabled transitions have an environment place in their precondition.
	For states~$v_{15}$~to~$v_{18}$, each decision marking corresponds to a final cut.
	The decision marking of state~$v_3$ corresponds to an mcut as the second system player in $p$ decided to not allow any of its outgoing transitions.
\end{example}

\subsection{Edges in the Büchi Game}

Edges in the Büchi game mostly connect states $\TPoneState =(\bbD, M_{T2}, \BR_1, \ldots, \BR_{\maxSys})$ and $\TPoneState'=(\bbD', M_{T2}', \BR_1', \ldots, \BR_{\maxSys}')$ where $\bbD$ is a decision marking, $M_{T2}$ is a marking, and $\BR_1, \ldots, \BR_{\maxSys}$ are as many sequences of backward moves as the maximum number $\maxSys$ of system players in the Petri game.
There are five sets of edges $\mathit{TOP}$, $\mathit{SYS}$, $\mathit{NES}$, $\mathit{MCUT}$, and $\mathit{STOP}$.
In the following description of the five sets of edges, not mentioned elements of the connected states stay the same.
The formal definitions can be found in \refAppendix{edgesBG}.
\begin{enumerate}[(1)]
	\item Edges from $\mathit{TOP}$ occur from states where at least one decision tuple in $\bbD$ has $\top$ as decision.
			To obtain $\bbD'$, Player~0 replaces each $\top$ in the decision tuples of system players with a set of allowed transitions and can change the \typeTwo status of decision tuples for system players from $\False$ to $\True$.
			The underlying marking of decision tuples with positive \typeTwo status $\True$ is stored in $M_{T2}'$ when a \typeTwo status changes.
	\item Edges from $\mathit{SYS}$ occur from states where all decision tuples in $\bbD$ have negative \typeTwo status and at least one transition with only system places in its precondition is enabled and allowed by the decision tuples in $\bbD$.
			To get~$\bbD'$, Player~0 simulates one such transition~$t$ firing by removing decision tuples~$\bbD_\mathit{pre}$ for the precondition of $t$ and adding decision tuples~$\bbD_\mathit{post}$ for the postcondition of $t$.
			For $\bbD_\mathit{post}$, the last mcut of all participating players is the maximum of their previous values and Player~0 picks the decisions and can change the \typeTwo status as in~(1).
			Marking~$M_{T2}'$ is obtained as in (1).
			Backward move $(\bbD_\mathit{pre}, \bbD_\mathit{post})$ is added to $\BR_\id$ of all participating players with identifier $\id$ to get $\BR_\id'$.
	\item Edges from $\mathit{NES}$ are the \typeTwo case and occur from states where a decision tuple in~$\bbD$ has positive \typeTwo status.
			To obtain $\bbD'$, Player~0 fires a transition as in~(2) but only from decision tuples with positive \typeTwo status resulting in new decision tuples with positive \typeTwo status. 
			This includes the storage of backward moves.
			The \typeTwo case is successful if the marking $M_{T2}$ is reached again and all players in it moved.
			Then, decision tuples with \typeTwo status $\True$ are set to \typeTwo status $\End$ and $M_{T2}'$ becomes the empty marking.
	\item Edges from $\mathit{MCUT}$ occur from states where all enabled and allowed transitions have an environment place in their precondition.
			To get $\bbD'$, Player~1 fires one such transition.
			Decision tuples for the precondition of the transition are removed, decision tuples for the postcondition are added.
			Added decision tuples for system players have \typeTwo status~$\False$, $\top$ as decision, an empty sequence of backward moves, and the highest last mcut.
			As backward moves store the past of system players until their last mcut, backward moves for system players that are part of the transition are removed.
			If backward moves become never applicable by firing the transition, they are removed from the successor state.
	\item Edges from $\mathit{STOP}$ occur from states with no transition enabled or corresponding to losing behavior.
			They replace other outgoing edges for losing behavior.
			States corresponding to termination lead to the winning state~$F_B$.
			States corresponding to a deadlock but not termination lead to the losing state~$F_N$.
			If backward moves detect a bad marking or a nondeterministic decision, the state leads to~$F_N$.
			In the \typeTwo case, a synchronization of decision tuples with positive and negative \typeTwo status or a deadlock or vanishing of decision tuples with positive \typeTwo status leads to~$F_N$.
			Decision tuples with positive \typeTwo status can vanish when transitions with empty postcondition fire.
\end{enumerate}

\begin{example}
	In \refFig{graphgame}, outgoing edges of state $v_1$ are in $\mathit{TOP}$.
	Outgoing edges of states~$v_2$,~$v_6$, and $v_8$ to $v_{12}$ are in $\mathit{SYS}$.
	Outgoing edges of states $v_0$, $v_3$, $v_{13}$, and $v_{14}$ are in $\mathit{MCUT}$.
	Other edges are in $\mathit{STOP}$.
	No edges in $\mathit{NES}$ exist in the depicted part.
	Outgoing edges of states~$v_4$ and~$v_5$ represent the deadlock of the second system player in~$p$ disallowing both outgoing transitions while only they are enabled.
	The outgoing edge of state~$v_7$ encodes a nondeterministic decision of the second system player, which allows two enabled transitions.
	Such a decision is only useful if another player ensures that at most one of the transitions becomes enabled.
	Outgoing edges of states $v_{15}$ to $v_{17}$ represent termination.
	The outgoing edge of state~$v_{18}$ represents a bad marking for six produced units of energy.
\end{example}

When, as in our construction,
(I) Player~0 immediately resolves $\top$ to the decisions of system players,
(II) Player~0 decides which transitions with only system places in their precondition fire following the decisions of system players, and
(III) Player~1 decides as late as possible at mcuts which transitions with an environment place in their precondition fire following the decisions of system players, then the corresponding Büchi games encode causal memory~\cite{DBLP:journals/iandc/FinkbeinerO17}.
Allowed transitions with only system places in their precondition fire in an order determined by Player~0 until an mcut is reached.
There, Player~1 decides for the environment player which allowed transition to fire.
Afterward, this process repeats itself.

\subsection{Backward Moves in the Büchi Game}
\label{sec:backwardrules}

In the Büchi game, Player~0 can avoid markings by picking the firing order for transitions with only system places in their precondition.
In \refFig{type2}, the two system players in~$s_2$ and~$s_3$ are reached after~$t_1$ fires.
One can fire~$t_3$, the other~$t_4$.
This results in the firing sequences $t_1 t_3 t_4$ and $t_1 t_4 t_3$.
If $s_2$ and~$s_7$ are in a bad marking, then Player~0 can decide for edges corresponding to the first firing sequence and the bad marking is missed.
We introduce backward moves to avoid such problems.
A \emph{backward move} is a pair of decision markings.
It stores the change to the decision tuples by edges from $\mathit{SYS}$ and $\mathit{NES}$.
For every such edge from $\TPoneState =(\bbD, M_{T2}, \BR_1, \ldots, \BR_{\maxSys})$ to $\TPoneState' =(\bbD', M_{T2}', \BR_1', \ldots, \BR_{\maxSys}')$, we obtain $\bbD_\mathit{pre}$ and $\bbD_\mathit{post}$ with $\bbD' = (\bbD \setminus \bbD_\mathit{pre}) \cup \bbD_\mathit{post}$ and add backward move $(\bbD_\mathit{pre}, \bbD_\mathit{post})$ to the end of $\BR_\id$ of all participating players with identifier $\id$.

For every state $\TPoneState'$ in the Büchi game, it is checked with backward moves if $\TPoneState'$ is losing due to a bad marking or a nondeterministic decision.
The decision marking~$\bbD'$ and all decision markings that are reachable via backward moves are checked.
Therefore, it is checked whether backward moves $(\bbD_\mathit{pre}, \bbD_\mathit{post})$ are \emph{applicable} to~$\bbD'$, i.e., whether $\bbD_\mathit{post} \subseteq \bbD'$ and $(\bbD_\mathit{pre}, \bbD_\mathit{post})$ is the last backward move of all participating players.
In this case, the backward move is removed from the end of the sequences of backward moves of all participating players and $\bbD = (\bbD' \setminus \bbD_\mathit{post}) \cup \bbD_\mathit{pre}$ results from the application of the backward move.
The underlying marking of~$\bbD$ is checked to not be a bad marking and $\bbD$ is checked to have only deterministic decisions.
This is repeated recursively from $\bbD$ for all applicable backward moves until no backward move is applicable.
If a decision marking corresponding to a bad marking or a nondeterministic decision is detected, the current state $\TPoneState'$ only has an edge to $F_N$.

The identifier of players in decision tuples is used to map the decision tuple to the corresponding sequence of backward moves, i.e., for each system player in the Petri game, the Büchi game collects a sequence of backward moves.
Edges from $\mathit{MCUT}$ empty the sequence of backward moves of decision tuples when their system place is in the precondition of the fired transition.
This removal can make backward moves not applicable because some participating players do not have the backward move as their last one anymore.

The sequence of backward moves can grow infinitely long when system players play infinitely without the environment player and without the \typeTwo case.
This would result in a Büchi game with infinitely many states.
To avoid this, the Büchi game becomes losing for Player~0 when it plays in a way that corresponds to a strategy with a variant of useless repetitions~\cite{DBLP:conf/fsttcs/Gimbert17} for the system players in the Petri game.
Our variant of useless repetitions identifies the repetition of a loop consisting only of transitions without the environment player in their precondition such that the last mcut of the system players does not change, i.e., the system players repeat a loop in which they do not exchange any new information about the environment player.
Thus, winning strategies have to avoid playing a useless repetition more than once between the successor of an mcut and the next mcut.
This can be achieved either by continuing to the next mcut or by setting some players to \typeTwo status $\True$ and completing the \typeTwo case, i.e., playing infinitely without the environment player.

\begin{example}
In \refFig{graphgame}, we include the collection of backward moves.
State~$v_{13}$ represents each power plant producing one unit of energy after a sunny weather forecast.
It is reached from state $v_6$ either via state $v_9$ or $v_{10}$ depending on which power plant produces energy first.
State~$v_{13}$ has a backward move for each power plant: $( \{ (1, p, \False, \{p_l\}, 1) \}, \{ (1, k, \False, \emptyset, 1) \} )$ and $( \{ (2, p, \False, \{p_l\}, 1) \}, \{ (2, k, \False, \emptyset, 1) \} )$.
Because the three markings $\{s, k:2\}$ (underlying marking of~$v_{13}$), $\{s, p, k\}$ (applying one backward move), and $\{s, p:2\}$ (applying both backward moves) are no bad markings and all decisions are deterministic, state $v_{13}$ continues with edges for the transitions of the environment place $s$ instead of having an edge to $F_N$.
\end{example}

\subsection{Encoding the \TypeTwo Case Directly in the Büchi Game}
\label{sec:typetwo}

We handle the \emph{\typeTwo case} where system players play infinitely without firing a transition with an environment place in its precondition directly in the Büchi game as players in the \typeTwo case might be in a bad marking.
This is in contrast to the reduction for bad places~\cite{DBLP:journals/iandc/FinkbeinerO17}.

In the Büchi game, Player~0 has to reach an accepting state infinitely often in order to win the game.
Only $F_B$ and states corresponding to an mcut are accepting states.
Transitions with only system places in their precondition are fired between successors of mcuts and the following mcut.
Thus, if the system players can fire transitions with only system places in their precondition infinitely often, eventually a useless repetition is reached which is losing.
To overcome this, we give Player~0 the possibility to change the \typeTwo status for decision tuples of system players from negative to positive.
The underlying marking of this change is stored and afterward only transitions from decision tuples with positive \typeTwo status can be fired.
Firing these transitions maintains the positive \typeTwo status for new decision tuples.
Instead of firing infinitely many transitions, the \typeTwo case is ended if the stored marking is reached again and all players in the marking have moved.
In this case, the \typeTwo status of all decision tuples with positive \typeTwo status is changed to ended \typeTwo status and the Büchi game continues with the remaining decision tuples with negative \typeTwo status.
The requirement to move is necessary as otherwise too many players could get ended \typeTwo status.
Decision tuples with ended \typeTwo status are maintained as backward moves can be applicable to them, i.e., backward moves store the \typeTwo status and allow us to reverse it in search for a bad marking.
We can thus ensure that continuing with the case where all decision tuples have negative \typeTwo status avoids bad markings that span the \typeTwo case.

Player~0 has to disclose decision tuples with positive \typeTwo status if system players fire infinitely many transitions with only system places in their precondition.
Otherwise, they lose the game as no accepting state is reached infinitely often.
It is losing if system players with positive and negative \typeTwo status synchronize, if players with positive \typeTwo status deadlock, if one such player is not moved and the marking from the start of the \typeTwo case is repeated, if all such players vanish, or if another marking is repeated.
Notice that at most one \typeTwo case is necessary per branch in the strategy tree of the Büchi game.
For a safety winning condition, possible \typeTwo cases after the first successful one can simply terminate.

\begin{example}
	A Petri game with necessary \typeTwo case in the encoding Büchi game is shown in \refFig{type2}.
	After Player~0 allows transition $t_2$ and Player~1 fires it, a state is reached where the decision tuples for $s_4$ and $s_5$ can be set to positive \typeTwo status by Player~0.
	After transitions $t_5$, $t_6$, and $t_7$ fire, the marking $\{s_4, s_5\}$ is repeated and the \typeTwo case is successful, proving that $t_5$, $t_6$, and $t_7$ can fire infinitely often.
\end{example}

\subsection{Decidability Result}

We analyze the properties of the constructed Büchi game.
Detailed proofs are in \refAppendix{BMcorrectness}.

\begin{lemma}[From Büchi game to Petri game strategies]\label{lem:BGtoPG}
	If Player~0 has a winning strategy in the Büchi game, then there is a winning strategy for the system players in the Petri game.
\end{lemma}
\begin{proof}[Proof Sketch]
	From the tree $T_f$ representing the winning strategy $f$ for Player~0 in the Büchi game, we inductively build a winning strategy $\sysstrat$ for the system players in the Petri game.
	Each cut in $\sysstrat$ is associated with a node in $T_f$, transitions are added following the edges in~$T_f$, and the associated cut is updated if needed.
	This strategy $\sysstrat$ for the system players in the Petri game is winning as it visits equivalent cuts to the reachable states in~$f$.
\end{proof}

\begin{lemma}[From Petri game to Büchi game strategies]\label{lem:PGtoBG}
	If the system players have a winning strategy in the Petri game, then there is a winning strategy for Player~0 in the Büchi game.
\end{lemma}
\begin{proof}[Proof Sketch]
	We skip unnecessary \typeTwo cases and useless repetitions in the winning strategy~$\sysstrat$ for the system players in the Petri game. 
	We replace $\top$ based on the postcondition of system places, disclose necessary \typeTwo cases, fire enabled transitions with only system places in their precondition in an arbitrary but fixed order between states after an mcut and the next mcut, and add all options at mcuts.
	This strategy for Player~0 in the Büchi game is winning as it visits equivalent states to the reachable cuts in~$\sysstrat$.
\end{proof}

\begin{theorem}[Game solving]\label{theo:gameSolving}
	For Petri games with a bounded number of system players, one environment player, and bad markings, the question of whether the system players have a winning strategy is decidable in \emph{2-EXPTIME}.
	If a winning strategy for the system players exists, it can be constructed in exponential time.
\end{theorem}
\begin{proof}[Proof Sketch]
	The complexity is based on the double exponential number of states in the Büchi game and polynomial solving of Büchi games.
	There are exponentially many states in the size of the Petri game to represent decision tuples and each of these states has to store sequences of backward moves of at most exponential length in the size of the Petri game.
	This transfers to the size of the winning strategy because it can be represented finitely.
\end{proof}

\begin{remark}
	In the presented construction, Player~0 in the Büchi game decides both the decisions of the system players in the Petri game and the order in which concurrent transitions with only system places in their precondition are fired.
	The question might arise whether it is possible that Player~1 representing the environment determines the order in which concurrent transitions with only system places in their precondition are fired. 
	This is not possible because the system players can make different decisions depending on the order of transitions decided by the environment player.
	We present a detailed counterexample where this change would result in a different winner of a Petri game in \refAppendix{SysVsEnvScheduling}. 
\end{remark}

\section{Undecidability in Petri Games with Good Markings}\label{sec:undecPG}

We prove that it is undecidable if a winning strategy exists for the system players in Petri games with at least two system and one environment player and good and bad markings by enforcing an undecidable synchronous setting in Petri games.
For this winning condition, no bad marking should be reached \emph{until} a good marking is reached, which can be expressed in LTL.
Notice also that, after a good marking has been reached, it is allowed to reach a bad marking.
Afterward, we prove that it is undecidable if a winning strategy exists for the system players in Petri games with only good markings and at least three players, out of which one is an environment player and each of the other two changes between system and environment player.
Bad markings from the previous result are encoded by system players repeatedly changing to environment players and back.
More details can be found in \refAppendix{undec}.
The underlying main idea of the first construction is also used in other settings~\cite{DBLP:conf/concur/MadhusudanT02, DBLP:conf/focs/PnueliR90, DBLP:journals/jcss/Reif84}.

\subsection{Petri Game for the Post Correspondence Problem}

We recall that a strategy $\sysstrat$ is winning for good and bad markings $\pSpecial = (\goodmarkings, \badmarkings) $ if, for all complete firing sequences $t_0 t_1 t_2 \ldots$ of all maximal plays $\pi$ of $\sysstrat$ with $M_0 = \init^\pi \land M_0 \firable{t_0} M_1 \firable{t_1} M_2 \firable{t_2} \ldots$, there exists $i \geq 0$ with $\lambda^\pi[M_i] \in \goodmarkings \land \forall 0 \leq j < i : \lambda^\pi[M_j]\notin \badmarkings$.
The undecidability proof uses the Post correspondence problem~\cite{post1946variant}.
The \emph{Post correspondence problem} (PCP) is to determine, for a finite alphabet $\Sigma$ and two finite lists $r_0, r_1, \ldots, r_n$ and $v_0, v_1, \ldots, v_n$ of non-empty words over $\Sigma$, if there exists a non-empty sequence $i_1, i_2, \ldots, i_l \in \{0,1,\ldots,n\}$ such that $r_{i_1}r_{i_2}\ldots r_{i_l} = v_{i_1}v_{i_2}\ldots v_{i_l}$.
This problem is undecidable. 

To simulate the PCP in a Petri game, we use one environment and two system players.
The three players are \emph{independent} as they cannot communicate with each other.
Each system player outputs a solution to the PCP.
By firing a transition, a player \emph{outputs the label} of the transition.
The output of the first system player is $i_1 r_{i_1} \tau i_2 r_{i_2} \tau \ldots i_l r_{i_l} \tau \oneend$ and the output of the second one is $j_1 v_{j_1} \tau j_2 v_{j_2} \tau \ldots j_m v_{j_m} \tau \twoend$ for $i_1,\ldots, i_l, j_1, \ldots, j_m \in \{0,1,\ldots n\}$.
Both system players output \emph{indices} followed by the word from the index position of the respective list and $\tau$, and end symbol $\oneend$ or $\twoend$ at the end of the sequence.
Words $r_i$ for $i\in\{i_1,\ldots, i_l\}$ and $v_j$ for $j\in\{j_1, \ldots, j_m\}$ are output letter-by-letter.
A correct solution to the PCP fulfills $l > 0$, $m > 0$, $l=m$, $i_1=j_1$, $i_2=j_2$, \ldots, $i_l=j_m$, and $r_{i_1}r_{i_2}\ldots r_{i_l} = v_{j_1}v_{j_2}\ldots v_{j_m}$.

We ensure that strategies for the two system players can only win by outputting the same sequence of indices at both players.
This permits to decide for these strategies if a good marking is reached where both system players have output a correct solution.
Using the independence of the three players and depending on a choice by the environment player, we either check the equality of the output sequences of indices or of the letter-by-letter output sequences of words.
Therefore, the strategy for the system players has to behave as if both is tested.
With good markings, we restrict the asynchronous setting of Petri games to turn-taking firing sequences on the output indices or letters.
Thus, we consider equivalent firing sequences to the synchronous setting and can check the conditions for a correct solution to the PCP after both system players have output the end symbol.
With bad markings, we identify when output indices or output letters do not match.
System players can only output the end symbol after outputting at least one index and word to ensure non-empty solutions.

\subsection{Linear Firing Sequences via Good Markings}

We use \MODT counters to restrict the asynchronous setting of Petri games to firing sequences equivalent to the synchronous setting of Pnueli and Rosner~\cite{DBLP:conf/focs/PnueliR90, DBLP:journals/ipl/Schewe14}.
The main idea is that we are just interested in runs where the first system player is only zero or one step ahead of the second system player.
When the second system player is ahead of the first one or the first system player is two or more steps ahead of the second one, then a good marking is reached and the possible reaching of bad markings afterwards does not matter.
Hence, bad markings are only checked for runs where the first system player is zero or one step ahead of the second system player until they reach a good marking for giving an answer to the PCP.

Formally, for each system player, we introduce two \emph{\MODT counters} to count the number of output indices and of output letters modulo three.
When a player outputs an index, the respective index counter is increased by one, and accordingly for output letters and the letter counter.
If a counter would reach value three, it is reset to zero.
We define good markings based on the two \MODT index counters and the two \MODT letter counters.
In a \emph{linear firing sequence for indices (letters)}, the two system players output the indices (letters) alternately with the first system player preceding the second one at each turn.
We ensure that the environment player first decides that either the output indices or letters are checked for equality.
Afterward, a good marking is reached when a firing sequence is not a linear firing sequence for indices or letters, depending on the decision by the environment player.

In \refFig{linear_scheduling}, we visualize the reachability graph for the two system players when only considering either the values of their \MODT counters for indices or letters.
Markings are differentiated in the reachability graph depending on if a good marking is reached before, e.g., position $(0\parallel 1)$ does not lead to position $(1\parallel 1)$ as the path to $(1\parallel 1)$ does not include a good marking.
With linear firing sequences, we only consider firing sequences where the first system player outputs the first index or letter before the second system player as the opposite cases are good markings.
For firing sequences not reaching a good marking, equality of output indices or letters is checked at positions $(0\parallel 0)$, $(1\parallel 1)$, and $(2\parallel 2)$.
Thereby, equality of output indices or letters at the same position can be checked without storing all outputs and it is ensured that solutions have the same length.

Notice that linear firing sequences for indices do not restrict the order in which the two system players output letters between two indices, and vice versa.
Also, we at least need a \MODT counter.
For a MOD\nobreakdash-2 counter, the good marking $(2\parallel 0)$ is replaced by $(0\parallel 0)$, implying that all firing sequences contain a good marking.
A \MODT counter prevents that one player overtakes the other.
Thus, indices or letters at different positions are not compared, i.e., output indices or letters at position $(0\parallel 3)$ (not modulo three) can be different.

\begin{figure}[t]
\centering
\scalebox{0.75}{	
\begin{tikzpicture}[
	terminal/.style={
		rectangle, minimum size=6mm,
		very thick, draw=black
	},
	good/.style={
		rectangle, minimum size=6mm,
		very thick, draw=black,
		top color=green!30, bottom color=green!30
	},
	node distance=5mm, every on chain/.style={join}, every join/.style={shorten <= 2pt, shorten >= 2pt, ->,->, >=stealth'}
]
{ [start chain]
	\node (start) [on chain, terminal] {$(0\parallel 0)$};
	{ [start branch=plus1]
		\node (plus1) [good, on chain=going below] {$(0\parallel 1)$};
	}
	\node [on chain=going right, terminal, xshift=10mm] {$(1\parallel 0)$};
	{ [start branch=minus1]
		\node (minus1) [good, on chain=going below] {$(2\parallel 0)$};
	}
	\node [on chain=going right, terminal, xshift=10mm] {$(1\parallel 1)$};
	{ [start branch=plus2]
		\node (plus2) [good, on chain=going below] {$(1\parallel 2)$};
	}
	\node [on chain=going right, terminal, xshift=10mm] {$(2\parallel 1)$};
	{ [start branch=minus2]
		\node (minus2) [good, on chain=going below] {$(0\parallel 1)$};
	}
	\node [on chain=going right, terminal, xshift=10mm] {$(2\parallel 2)$};
	{ [start branch=plus3]
		\node (plus3) [good, on chain=going below] {$(2\parallel 0)$};
	}
	\node (from) [on chain=going right, terminal, xshift=10mm] {$(0\parallel 2)$};
	{ [start branch=minus3]
		\node (minus3) [good, on chain=going below] {$(1\parallel 2)$};
	}
}
\draw[shorten <= 2pt, shorten >= 2pt, ->, >=stealth'] (from.east) -- ([xshift=1.5cm]from.east) -- ([xshift=1.5cm, yshift=-1.6cm]from.east) -- ([xshift=-1.5cm, yshift=-1.6cm]start.west) -- ([xshift=-1.5cm]start.west) -- (start.west);

\node (d0) [right of=plus1, right of=plus1, right of=plus1, yshift=0.25cm] {\ldots};
\draw[shorten <= 2pt, shorten >= 2pt, ->, >=stealth'] (plus1) to (d0);
\node (d1) [right of=plus1, right of=plus1, right of=plus1, yshift=-0.25cm] {\ldots};
\draw[shorten <= 2pt, shorten >= 2pt, ->, >=stealth'] (plus1) to (d1);
\node (d2) [right of=plus2, right of=plus2, right of=plus2, yshift=0.25cm] {\ldots};
\draw[shorten <= 2pt, shorten >= 2pt, ->, >=stealth'] (plus2) to (d2);
\node (d3) [right of=plus2, right of=plus2, right of=plus2, yshift=-0.25cm] {\ldots};
\draw[shorten <= 2pt, shorten >= 2pt, ->, >=stealth'] (plus2) to (d3);
\node (d4) [right of=plus3, right of=plus3, right of=plus3, yshift=0.25cm] {\ldots};
\draw[shorten <= 2pt, shorten >= 2pt, ->, >=stealth'] (plus3) to (d4);
\node (d5) [right of=plus3, right of=plus3, right of=plus3, yshift=-0.25cm] {\ldots};
\draw[shorten <= 2pt, shorten >= 2pt, ->, >=stealth'] (plus3) to (d5);
\node (d6) [right of=minus1, right of=minus1, right of=minus1, yshift=0.25cm] {\ldots};
\draw[shorten <= 2pt, shorten >= 2pt, ->, >=stealth'] (minus1) to (d6);
\node (d7) [right of=minus1, right of=minus1, right of=minus1, yshift=-0.25cm] {\ldots};
\draw[shorten <= 2pt, shorten >= 2pt, ->, >=stealth'] (minus1) to (d7);
\node (d8) [right of=minus2, right of=minus2, right of=minus2, yshift=0.25cm] {\ldots};
\draw[shorten <= 2pt, shorten >= 2pt, ->, >=stealth'] (minus2) to (d8);
\node (d9) [right of=minus2, right of=minus2, right of=minus2, yshift=-0.25cm] {\ldots};
\draw[shorten <= 2pt, shorten >= 2pt, ->, >=stealth'] (minus2) to (d9);
\node (d10) [right of=minus3, right of=minus3, right of=minus3, yshift=0.25cm] {\ldots};
\draw[shorten <= 2pt, shorten >= 2pt, ->, >=stealth'] (minus3) to (d10);
\node (d11) [right of=minus3, right of=minus3, right of=minus3, yshift=-0.25cm] {\ldots};
\draw[shorten <= 2pt, shorten >= 2pt, ->, >=stealth'] (minus3) to (d11);
\end{tikzpicture}
}
\caption{The reachability graph for the two system players is depicted when only considering either the values of their \MODT index counters or of their \MODT letter counters and differentiating markings depending on if a good marking is reached before. Good markings are colored green.
	All behavior after a good marking (including reaching a bad marking) is winning by definition.
	To compare output indices or letters, only the specific firing sequence in white has to be considered.
	}
\label{fig:linear_scheduling}
\end{figure}
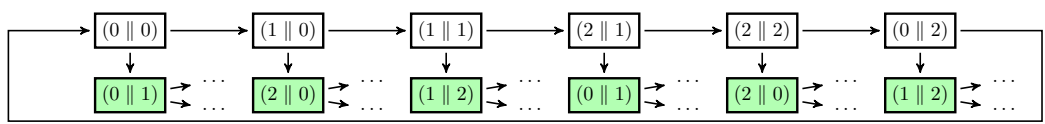

\subsection{Preventing Untruthful Termination}

The good markings to only consider linear firing sequences introduce new possibilities for the system players to be winning.
These possibilities arise when the system players can enforce all firing sequences to reach a good marking.
They occur when a system player terminates without the end symbol ($\oneend$ or $\twoend$) and are called \emph{untruthful termination}.
Untruthful termination is prevented by letting the environment player decide which system player it believes to not terminate with the end symbol or that everything is okay.
This decision happens together with the initial choice of the environment player between checking equality of indices or letters.
Due to the independence of the players, each system player has to behave as if the environment player is anticipating it to untruthfully terminate and has to output the end symbol to avoid this.
Therefore, no untruthful termination can occur.

\subsection{Undecidability Results}

A winning strategy exists in the Petri game iff there exists a solution to the instance of the PCP.
The only strategy with a chance to be winning for the two system players is to output the same solution to the PCP and we can translate solutions between both cases.
We obtain:

\begin{theorem}\label{theo:undecGoodAndBad}
	For Petri games with good and bad markings and at least two system and one environment player, the question if the system players have a winning strategy is undecidable.
\end{theorem}
Bad markings can be encoded by system players repeatedly changing to environment players and back.
Players commit to transitions and then system players become environment players.
Environment players either follow the committed transition and fire a transition returning to the respective system and environment players or fire a transition with all other environment players after which no good markings are reachable to encode a bad marking.
We obtain:
\begin{theorem}\label{theo:undecGoodAndThree}
	For Petri games with good markings and at least three players, out of which one is an environment player and each of the other two changes between system and environment player, the question if the system players have a winning strategy is undecidable.
\end{theorem}

\section{Conclusion}\label{sec:conclusion}

We have investigated global winning conditions for the synthesis of asynchronous distributed reactive systems with causal memory. 
The general decidability or undecidability of the synthesis problem for these systems is a long-standing open question~\cite{DBLP:conf/icalp/Muscholl15, DBLP:journals/iandc/FinkbeinerO17}. 
We encode the synthesis problem for these systems by Petri games.
For global winning conditions, we achieve a clear picture regarding decidability and undecidability.

From our decidability result and previous work~\cite{DBLP:conf/fsttcs/FinkbeinerG17}, we obtain for bad markings as global winning condition that the question of whether the system players have a winning strategy is decidable for Petri games where the number of system players or the number of environment players is at most one and the number of players of the converse type can be bounded by some arbitrary number.
For bad markings as global winning condition, this leaves the case of Petri games with two or more system players \emph{and} two or more environment players open.

From our undecidability results, we obtain for good markings as global winning condition that the question of whether the system players have a winning strategy is undecidable for Petri games with two or more system players and three or more environment players.
For good markings as global winning condition, this only leaves the corner case of Petri games with at most one system player and at most two environment players open. 

Thus, for the synthesis of asynchronous distributed reactive systems with causal memory, global safety winning conditions are decidable for a large class of such systems, whereas global liveness winning conditions are undecidable for almost all classes of such systems.
In the future, we plan to combine the decidability results for bad markings as global safety winning condition with local liveness specifications per player as in Flow-LTL~\cite{DBLP:conf/atva/FinkbeinerGHO19, DBLP:conf/cav/FinkbeinerGHO20}.



\bibliography{lipics-v2019-sample-article}

\begin{thebibliography}{10}

\bibitem{DBLP:conf/concur/BeutnerFH19}
Raven Beutner, Bernd Finkbeiner, and Jesko Hecking{-}Harbusch.
\newblock Translating asynchronous games for distributed synthesis.
\newblock In {\em 30th International Conference on Concurrency Theory,
  {CONCUR}}, volume 140 of {\em LIPIcs}. Schloss Dagstuhl - Leibniz-Zentrum
  f{\"{u}}r Informatik, 2019.
\newblock \href {https://doi.org/10.4230/LIPIcs.CONCUR.2019.26}
  {\path{doi:10.4230/LIPIcs.CONCUR.2019.26}}.

\bibitem{DBLP:conf/cav/BohyBFJR12}
Aaron Bohy, V{\'{e}}ronique Bruy{\`{e}}re, Emmanuel Filiot, Naiyong Jin, and
  Jean{-}Fran{\c{c}}ois Raskin.
\newblock Acacia+, a tool for {LTL} synthesis.
\newblock In {\em Computer Aided Verification - 24th International Conference,
  {CAV}}, volume 7358 of {\em Lecture Notes in Computer Science}. Springer,
  2012.
\newblock \href {https://doi.org/10.1007/978-3-642-31424-7\_45}
  {\path{doi:10.1007/978-3-642-31424-7\_45}}.

\bibitem{DBLP:conf/soda/ChatterjeeH12}
Krishnendu Chatterjee and Monika Henzinger.
\newblock An \emph{O}(\emph{n}\({}^{\mbox{2}}\)) time algorithm for alternating
  {B{\"{u}}chi} games.
\newblock In {\em Twenty-Third Annual {ACM-SIAM} Symposium on Discrete
  Algorithms, {SODA}}. {SIAM}, 2012.
\newblock \href {https://doi.org/10.1137/1.9781611973099.109}
  {\path{doi:10.1137/1.9781611973099.109}}.

\bibitem{DBLP:conf/tacas/Ehlers11}
R{\"{u}}diger Ehlers.
\newblock Unbeast: Symbolic bounded synthesis.
\newblock In {\em Tools and Algorithms for the Construction and Analysis of
  Systems - 17th International Conference, {TACAS}}, volume 6605 of {\em
  Lecture Notes in Computer Science}. Springer, 2011.
\newblock \href {https://doi.org/10.1007/978-3-642-19835-9\_25}
  {\path{doi:10.1007/978-3-642-19835-9\_25}}.

\bibitem{DBLP:journals/acta/Engelfriet91}
Joost Engelfriet.
\newblock Branching processes of {Petri} nets.
\newblock {\em Acta Informatica}, 28(6), 1991.
\newblock \href {https://doi.org/10.1007/BF01463946}
  {\path{doi:10.1007/BF01463946}}.

\bibitem{DBLP:series/eatcs/EsparzaH08}
Javier Esparza and Keijo Heljanko.
\newblock {\em Unfoldings - {A} Partial-Order Approach to Model Checking}.
\newblock Monographs in Theoretical Computer Science. An {EATCS} Series.
  Springer, 2008.
\newblock \href {https://doi.org/10.1007/978-3-540-77426-6}
  {\path{doi:10.1007/978-3-540-77426-6}}.

\bibitem{DBLP:conf/tacas/FaymonvilleFRT17}
Peter Faymonville, Bernd Finkbeiner, Markus~N. Rabe, and Leander Tentrup.
\newblock Encodings of bounded synthesis.
\newblock In {\em Tools and Algorithms for the Construction and Analysis of
  Systems - 23rd International Conference, {TACAS}}, volume 10205 of {\em
  Lecture Notes in Computer Science}, 2017.
\newblock \href {https://doi.org/10.1007/978-3-662-54577-5\_20}
  {\path{doi:10.1007/978-3-662-54577-5\_20}}.

\bibitem{DBLP:conf/birthday/Finkbeiner15}
Bernd Finkbeiner.
\newblock Bounded synthesis for {Petri} games.
\newblock In {\em Correct System Design - Symposium in Honor of
  Ernst-R{\"{u}}diger Olderog on the Occasion of His 60th Birthday}, volume
  9360 of {\em Lecture Notes in Computer Science}. Springer, 2015.
\newblock \href {https://doi.org/10.1007/978-3-319-23506-6\_15}
  {\path{doi:10.1007/978-3-319-23506-6\_15}}.

\bibitem{DBLP:journals/corr/abs-1711-10637}
Bernd Finkbeiner, Manuel Gieseking, Jesko Hecking{-}Harbusch, and
  Ernst{-}R{\"{u}}diger Olderog.
\newblock Symbolic vs. bounded synthesis for {Petri} games.
\newblock In {\em Sixth Workshop on Synthesis, SYNT@CAV}, volume 260 of {\em
  {EPTCS}}, 2017.
\newblock \href {https://doi.org/10.4204/EPTCS.260.5}
  {\path{doi:10.4204/EPTCS.260.5}}.

\bibitem{DBLP:conf/atva/FinkbeinerGHO19}
Bernd Finkbeiner, Manuel Gieseking, Jesko Hecking{-}Harbusch, and
  Ernst{-}R{\"{u}}diger Olderog.
\newblock Model checking data flows in concurrent network updates.
\newblock In {\em Automated Technology for Verification and Analysis - 17th
  International Symposium, {ATVA}}, volume 11781 of {\em Lecture Notes in
  Computer Science}. Springer, 2019.
\newblock \href {https://doi.org/10.1007/978-3-030-31784-3\_30}
  {\path{doi:10.1007/978-3-030-31784-3\_30}}.

\bibitem{DBLP:conf/cav/FinkbeinerGHO20}
Bernd Finkbeiner, Manuel Gieseking, Jesko Hecking{-}Harbusch, and
  Ernst{-}R{\"{u}}diger Olderog.
\newblock {AdamMC}: A model checker for {Petri} nets with transits against
  {Flow-LTL}.
\newblock In {\em Computer Aided Verification - 32nd International Conference,
  {CAV}}, volume 12225 of {\em Lecture Notes in Computer Science}. Springer,
  2020.
\newblock \href {https://doi.org/10.1007/978-3-030-53291-8\_5}
  {\path{doi:10.1007/978-3-030-53291-8\_5}}.

\bibitem{DBLP:conf/csl/FinkbeinerGHO22}
Bernd Finkbeiner, Manuel Gieseking, Jesko Hecking{-}Harbusch, and
  Ernst{-}R{\"{u}}diger Olderog.
\newblock Global winning conditions in synthesis of distributed systems with
  causal memory.
\newblock In {\em 30th {EACSL} Annual Conference on Computer Science Logic,
  {CSL} 2022}, LIPIcs. Schloss Dagstuhl - Leibniz-Zentrum f{\"{u}}r Informatik,
  2022.

\bibitem{DBLP:conf/cav/FinkbeinerGO15}
Bernd Finkbeiner, Manuel Gieseking, and Ernst{-}R{\"{u}}diger Olderog.
\newblock Adam: Causality-based synthesis of distributed systems.
\newblock In {\em Computer Aided Verification - 27th International Conference,
  {CAV}}, volume 9206 of {\em Lecture Notes in Computer Science}. Springer,
  2015.
\newblock \href {https://doi.org/10.1007/978-3-319-21690-4\_25}
  {\path{doi:10.1007/978-3-319-21690-4\_25}}.

\bibitem{DBLP:conf/fsttcs/FinkbeinerG17}
Bernd Finkbeiner and Paul G{\"{o}}lz.
\newblock Synthesis in distributed environments.
\newblock In {\em 37th {IARCS} Annual Conference on Foundations of Software
  Technology and Theoretical Computer Science, {FSTTCS}}, volume~93 of {\em
  LIPIcs}. Schloss Dagstuhl - Leibniz-Zentrum f{\"{u}}r Informatik, 2017.
\newblock \href {https://doi.org/10.4230/LIPIcs.FSTTCS.2017.28}
  {\path{doi:10.4230/LIPIcs.FSTTCS.2017.28}}.

\bibitem{DBLP:journals/iandc/FinkbeinerO17}
Bernd Finkbeiner and Ernst{-}R{\"{u}}diger Olderog.
\newblock Petri games: Synthesis of distributed systems with causal memory.
\newblock {\em Inf. Comput.}, 253, 2017.
\newblock \href {https://doi.org/10.1016/j.ic.2016.07.006}
  {\path{doi:10.1016/j.ic.2016.07.006}}.

\bibitem{DBLP:conf/fsttcs/GastinLZ04}
Paul Gastin, Benjamin Lerman, and Marc Zeitoun.
\newblock Distributed games with causal memory are decidable for
  series-parallel systems.
\newblock In {\em {FSTTCS}: Foundations of Software Technology and Theoretical
  Computer Science}, volume 3328 of {\em Lecture Notes in Computer Science}.
  Springer, 2004.
\newblock \href {https://doi.org/10.1007/978-3-540-30538-5\_23}
  {\path{doi:10.1007/978-3-540-30538-5\_23}}.

\bibitem{DBLP:conf/icalp/GenestGMW13}
Blaise Genest, Hugo Gimbert, Anca Muscholl, and Igor Walukiewicz.
\newblock Asynchronous games over tree architectures.
\newblock In {\em Automata, Languages, and Programming - 40th International
  Colloquium, {ICALP}}, volume 7966 of {\em Lecture Notes in Computer Science}.
  Springer, 2013.
\newblock \href {https://doi.org/10.1007/978-3-642-39212-2\_26}
  {\path{doi:10.1007/978-3-642-39212-2\_26}}.

\bibitem{DBLP:conf/tacas/GiesekingHY21}
Manuel Gieseking, Jesko Hecking{-}Harbusch, and Ann Yanich.
\newblock A web interface for {Petri} nets with transits and {Petri} games.
\newblock In {\em Tools and Algorithms for the Construction and Analysis of
  Systems - 27th International Conference, {TACAS}}, volume 12652 of {\em
  Lecture Notes in Computer Science}. Springer, 2021.
\newblock \href {https://doi.org/10.1007/978-3-030-72013-1\_22}
  {\path{doi:10.1007/978-3-030-72013-1\_22}}.

\bibitem{DBLP:conf/fsttcs/Gimbert17}
Hugo Gimbert.
\newblock On the control of asynchronous automata.
\newblock In {\em 37th {IARCS} Annual Conference on Foundations of Software
  Technology and Theoretical Computer Science, {FSTTCS}}, volume~93 of {\em
  LIPIcs}. Schloss Dagstuhl - Leibniz-Zentrum f{\"{u}}r Informatik, 2017.
\newblock \href {https://doi.org/10.4230/LIPIcs.FSTTCS.2017.30}
  {\path{doi:10.4230/LIPIcs.FSTTCS.2017.30}}.

\bibitem{DBLP:journals/iandc/GoltzR83}
Ursula Goltz and Wolfgang Reisig.
\newblock The non-sequential behavior of {Petri} nets.
\newblock {\em Inf. Control.}, 57(2/3):125--147, 1983.
\newblock \href {https://doi.org/10.1016/S0019-9958(83)80040-0}
  {\path{doi:10.1016/S0019-9958(83)80040-0}}.

\bibitem{DBLP:conf/atva/Hecking-Harbusch19}
Jesko Hecking{-}Harbusch and Niklas~O. Metzger.
\newblock Efficient trace encodings of bounded synthesis for asynchronous
  distributed systems.
\newblock In {\em Automated Technology for Verification and Analysis - 17th
  International Symposium, {ATVA}}, volume 11781 of {\em Lecture Notes in
  Computer Science}. Springer, 2019.
\newblock \href {https://doi.org/10.1007/978-3-030-31784-3\_22}
  {\path{doi:10.1007/978-3-030-31784-3\_22}}.

\bibitem{DBLP:journals/corr/Tentrup16}
Jesko Hecking{-}Harbusch and Leander Tentrup.
\newblock Solving {QBF} by abstraction.
\newblock In {\em Ninth International Symposium on Games, Automata, Logics, and
  Formal Verification, {GandALF}}, volume 277 of {\em {EPTCS}}, 2018.
\newblock \href {https://doi.org/10.4204/EPTCS.277.7}
  {\path{doi:10.4204/EPTCS.277.7}}.

\bibitem{DBLP:journals/corr/abs-1904-07736}
Swen Jacobs, Roderick Bloem, Maximilien Colange, Peter Faymonville, Bernd
  Finkbeiner, Ayrat Khalimov, Felix Klein, Michael Luttenberger, Philipp~J.
  Meyer, Thibaud Michaud, Mouhammad Sakr, Salomon Sickert, Leander Tentrup, and
  Adam Walker.
\newblock The 5th reactive synthesis competition {(SYNTCOMP} 2018): Benchmarks,
  participants {\&} results.
\newblock {\em CoRR}, abs/1904.07736, 2019.
\newblock \href {http://arxiv.org/abs/1904.07736} {\path{arXiv:1904.07736}}.

\bibitem{DBLP:conf/cav/JobstmannGWB07}
Barbara Jobstmann, Stefan~J. Galler, Martin Weiglhofer, and Roderick Bloem.
\newblock Anzu: {A} tool for property synthesis.
\newblock In {\em Computer Aided Verification, 19th International Conference,
  {CAV}}, volume 4590 of {\em Lecture Notes in Computer Science}. Springer,
  2007.
\newblock \href {https://doi.org/10.1007/978-3-540-73368-3\_29}
  {\path{doi:10.1007/978-3-540-73368-3\_29}}.

\bibitem{DBLP:journals/acta/KhomenkoKV03}
Victor Khomenko, Maciej Koutny, and Walter Vogler.
\newblock Canonical prefixes of {Petri} net unfoldings.
\newblock {\em Acta Informatica}, 40(2), 2003.
\newblock \href {https://doi.org/10.1007/s00236-003-0122-y}
  {\path{doi:10.1007/s00236-003-0122-y}}.

\bibitem{DBLP:journals/acta/LuttenbergerMS20}
Michael Luttenberger, Philipp~J. Meyer, and Salomon Sickert.
\newblock Practical synthesis of reactive systems from {LTL} specifications via
  parity games.
\newblock {\em Acta Informatica}, 57(1-2), 2020.
\newblock \href {https://doi.org/10.1007/s00236-019-00349-3}
  {\path{doi:10.1007/s00236-019-00349-3}}.

\bibitem{DBLP:conf/concur/MadhusudanT02}
P.~Madhusudan and P.~S. Thiagarajan.
\newblock A decidable class of asynchronous distributed controllers.
\newblock In {\em {CONCUR} - Concurrency Theory}, volume 2421 of {\em Lecture
  Notes in Computer Science}. Springer, 2002.
\newblock \href {https://doi.org/10.1007/3-540-45694-5\_11}
  {\path{doi:10.1007/3-540-45694-5\_11}}.

\bibitem{DBLP:conf/fsttcs/MadhusudanTY05}
P.~Madhusudan, P.~S. Thiagarajan, and Shaofa Yang.
\newblock The {MSO} theory of connectedly communicating processes.
\newblock In {\em {FSTTCS}: Foundations of Software Technology and Theoretical
  Computer Science}, volume 3821 of {\em Lecture Notes in Computer Science}.
  Springer, 2005.
\newblock \href {https://doi.org/10.1007/11590156\_16}
  {\path{doi:10.1007/11590156\_16}}.

\bibitem{DBLP:journals/tcs/MeseguerMS96}
Jos{\'{e}} Meseguer, Ugo Montanari, and Vladimiro Sassone.
\newblock Process versus unfolding semantics for place/transition {Petri} nets.
\newblock {\em Theor. Comput. Sci.}, 153(1{\&}2), 1996.
\newblock \href {https://doi.org/10.1016/0304-3975(95)00121-2}
  {\path{doi:10.1016/0304-3975(95)00121-2}}.

\bibitem{LTLsynt}
Thibaud Michaud and Maximilien Colange.
\newblock Reactive synthesis from {LTL} specification with {Spot}.
\newblock In {\em 7th Workshop on Synthesis, SYNT@CAV}, 2018.

\bibitem{DBLP:conf/icalp/Muscholl15}
Anca Muscholl.
\newblock Automated synthesis of distributed controllers.
\newblock In {\em Automata, Languages, and Programming - 42nd International
  Colloquium, {ICALP}}, volume 9135 of {\em Lecture Notes in Computer Science}.
  Springer, 2015.
\newblock \href {https://doi.org/10.1007/978-3-662-47666-6\_2}
  {\path{doi:10.1007/978-3-662-47666-6\_2}}.

\bibitem{DBLP:conf/fsttcs/MuschollW14}
Anca Muscholl and Igor Walukiewicz.
\newblock Distributed synthesis for acyclic architectures.
\newblock In {\em 34th International Conference on Foundation of Software
  Technology and Theoretical Computer Science, {FSTTCS}}, volume~29 of {\em
  LIPIcs}. Schloss Dagstuhl - Leibniz-Zentrum f{\"{u}}r Informatik, 2014.
\newblock \href {https://doi.org/10.4230/LIPIcs.FSTTCS.2014.639}
  {\path{doi:10.4230/LIPIcs.FSTTCS.2014.639}}.

\bibitem{DBLP:journals/tcs/NielsenPW81}
Mogens Nielsen, Gordon~D. Plotkin, and Glynn Winskel.
\newblock Petri nets, event structures and domains, {Part} {I}.
\newblock {\em Theor. Comput. Sci.}, 13, 1981.
\newblock \href {https://doi.org/10.1016/0304-3975(81)90112-2}
  {\path{doi:10.1016/0304-3975(81)90112-2}}.

\bibitem{DBLP:conf/focs/Pnueli77}
Amir Pnueli.
\newblock The temporal logic of programs.
\newblock In {\em 18th Annual Symposium on Foundations of Computer Science
  {FOCS}}, pages 46--57. {IEEE} Computer Society, 1977.
\newblock \href {https://doi.org/10.1109/SFCS.1977.32}
  {\path{doi:10.1109/SFCS.1977.32}}.

\bibitem{DBLP:conf/focs/PnueliR90}
Amir Pnueli and Roni Rosner.
\newblock Distributed reactive systems are hard to synthesize.
\newblock In {\em 31st Annual Symposium on Foundations of Computer Science
  {FOCS}}. {IEEE} Computer Society, 1990.
\newblock \href {https://doi.org/10.1109/FSCS.1990.89597}
  {\path{doi:10.1109/FSCS.1990.89597}}.

\bibitem{post1946variant}
Emil~L. Post.
\newblock A variant of a recursively unsolvable problem.
\newblock {\em Bull. Am. Math. Soc.}, 52(4):264--268, 1946.

\bibitem{DBLP:journals/jcss/Reif84}
John~H. Reif.
\newblock The complexity of two-player games of incomplete information.
\newblock {\em J. Comput. Syst. Sci.}, 29(2):274--301, 1984.
\newblock \href {https://doi.org/10.1016/0022-0000(84)90034-5}
  {\path{doi:10.1016/0022-0000(84)90034-5}}.

\bibitem{DBLP:books/sp/Reisig85a}
Wolfgang Reisig.
\newblock {\em Petri Nets: An Introduction}, volume~4 of {\em {EATCS}
  Monographs on Theoretical Computer Science}.
\newblock Springer, 1985.
\newblock \href {https://doi.org/10.1007/978-3-642-69968-9}
  {\path{doi:10.1007/978-3-642-69968-9}}.

\bibitem{DBLP:journals/ipl/Schewe14}
Sven Schewe.
\newblock Distributed synthesis is simply undecidable.
\newblock {\em Inf. Process. Lett.}, 114(4), 2014.
\newblock \href {https://doi.org/10.1016/j.ipl.2013.11.012}
  {\path{doi:10.1016/j.ipl.2013.11.012}}.

\bibitem{DBLP:journals/ita/Zielonka87}
Wieslaw Zielonka.
\newblock Notes on finite asynchronous automata.
\newblock {\em {ITA}}, 21(2), 1987.
\newblock \href {https://doi.org/10.1051/ita/1987210200991}
  {\path{doi:10.1051/ita/1987210200991}}.

\end{thebibliography}

\appendix
\allowdisplaybreaks

\section{Multisets}\label{sec:multisets}

We recall multisets as an extension of sets to allow repetitions of objects.
In a set, each object is either in the set or not, whereas, in a multiset, each object can additionally occur repeatedly.
In both sets and multisets, the order of the objects in the set or multiset does not matter.
Formally, a \emph{multiset} $M$ over a set $S$ is a function $M : S \rightarrow \mathbb{N}$.
We write $s \in M$ if $M(s) > 0$.
A \emph{set} is a $\{0,1\}$-valued multiset, and vice versa.
We use the following notations:
The \emph{empty multiset} $\emptyset$ is defined as $\emptyset(s) = 0$ for all $s \in S$.
For two multisets~$M$ and~$N$ over~$S$, \emph{multiset inclusion} $M \subseteq N$ is defined as $\forall s \in S : M(s) \leq N(s)$, \emph{multiset addition} $M + N$ is defined as $(M + N)(s) = M(s) + N(s)$ for all $s \in S$, \emph{multiset difference} $M - N$ is defined as $(M - N)(s) = \max(0, M(s) - N(s))$ for all $s \in S$, \emph{multiset intersection} $M \cap N$ is defined as $(M \cap N)(s) = \min(M(s), N(s)) $ for all $s \in S$, and \emph{multiset union} $M \cup N$ is defined as $(M \cup N)(s) = \max(M(s), N(s)) $ for all $s \in S$.

\section{Occurrence Nets, Branching Processes, and Unfoldings}\label{sec:unfolding}

We repeat the definitions of occurrence nets, branching processes, and unfoldings.
As the definitions result in safe Petri nets, we recall $k$-bounded and safe Petri nets.
For some $k \in \mathbb{N}$, a Petri net $\pNet$ is \emph{$k$-bounded} if $M(p) \leq k$ holds for all reachable markings $M \in \reach(\pNet)$ and all places $p \in \pl$.
The Petri net is \emph{bounded} if it is $k$-bounded for some given $k$; otherwise it is \emph{unbounded}.
The Petri net is \emph{safe} if it is $1$-bounded, i.e., places contain at most one token for all reachable markings in the Petri net.

An \emph{occurrence net} is a Petri net~$\pNet$, where the precondition and postcondition of all transitions are sets, the initial marking coincides with places without ingoing transitions ($\forall p \in \pl : p\in\init \Leftrightarrow |\pre{}{p}| = 0$), all other places have exactly one ingoing transition (${\forall p \in \pl \setminus \init : |\pre{}{p}| = 1}$), the inverse flow relation $\fl^{-1}$ is \emph{well-founded} (starting from any given node, no infinite path following the inverse flow relation exists), and no transition is in self-conflict.

An \emph{initial homomorphism} from $\pNet^1$ to $\pNet^2$ is a function $\lambda : \pl^1\cup \tr^1 \rightarrow \pl^2\cup \tr^2$ that respects node types ($\lambda(\pl^1) \subseteq \pl^2 \land \lambda(\tr^1) \subseteq \tr^2$),
is structure-preserving on transitions ($\forall t \in \tr^1 : \lambda[\pre{1}{t}] = \pre{2}{\lambda(t)} \land \lambda[\post{1}{t}] = \post{2}{\lambda(t)}$),
and agrees on the initial markings ($\lambda[\init^1] = \init^2$).

An occurrence net is a safe net.
If nodes $x \neq y$ of an occurrence nets are in conflict, they are mutually exclusive, i.e., there is a nondeterministic decision between $x$ and $y$.

A branching process \cite{DBLP:journals/acta/Engelfriet91, DBLP:journals/tcs/MeseguerMS96, DBLP:series/eatcs/EsparzaH08} describes parts of possible behaviors of a Petri net.
Formally, an \emph{(initial) branching process} of a Petri net $\pNet$ is a pair $\iota = (\pNet^\iota,\lambda^\iota)$ where $\pNet^\iota$ is an occurrence net and $\lambda^\iota :\pl^\iota\cup \tr^\iota\rightarrow\pl\cup \tr$ is an initial homomorphism  from $\pNet^\iota$ to $\pNet$ that is injective on transitions with the same precondition ($\forall t, t' \in \tr^\iota : (\pre{\iota}{t} = \pre{\iota}{t'} \land \lambda^\iota(t) = \lambda^\iota(t')) \Rightarrow t = t'$).
Intuitively, whenever a node can be reached on two distinct paths in a Petri net, it is split up in the branching process of the Petri net.
The initial homomorphism $\lambda^\iota$~can be thought of as label of the copies into nodes of $\pNet$.
The injectivity condition avoids additional unnecessary splits: Each transition must either be labeled differently or occur from different preconditions.

The \emph{unfolding} $\unf  =(\pNet^U, \lambda^U)$ of $\pNet$ is a maximal branching process: Whenever there is a set of pairwise concurrent places $C$ such that $\lambda^U[C] = \pre{\pNet}{t}$ for some transition $t \in \tr$, then there exists $t' \in \pNet^{U}$ with $\lambda^U(t')=t$ and $\pre{U}{t'} = C$.
It represents every possible behavior of $\pNet$.
Note that even finite Petri nets may have infinite unfoldings due to cycles.

Let $\iota_1 =(\pNet^1,\lambda^1)$ and $\iota_2 =(\pNet^2,\lambda^2)$ be two branching processes of $\pNet$.
A homomorphism from $\iota_1$ to $\iota_2$ is a homomorphism $h$ from $\pNet^1$ to $\pNet^2$ such that $\lambda^1=\lambda^2 \circ h$.
It is called \emph{initial} if $h$ is initial; it is an \emph{isomorphism} if $h$ is an isomorphism.
Two branching processes $\iota_1$ and $\iota_2$ are \emph{isomorphic} if there exists an initial isomorphism from $\iota_1$ to $\iota_2$.
A branching process $\iota_1$ \emph{approximates} a branching process $\iota_2$ if there exists an initial injective homomorphism from $\iota_1$ to $\iota_2$.
A branching process $\iota_1$ is a \emph{subprocess} of a branching process $\iota_2$ if $\iota_1$ approximates $\iota_2$ with the identity on $\pl^1 \cup \tr^1$ as the homomorphism.
Thus, $\pNet^1 \sqsubseteq \pNet^2$ and $\lambda^1 = \lambda^2 \upharpoonright (\pl^1\cup\tr^1)$.
If $\iota_1$ approximates $\iota_2$, then $\iota_1$ is isomorphic to a subprocess of $\iota_2$.
In \cite{DBLP:journals/acta/Engelfriet91}, it is shown that the unfolding $\iota_U = (\pNet^U, \lambda^U)$ of a net is unique up to isomorphism and that every initial branching process $\iota$ of~$\pNet$ approximates $\iota_U$.
Thus, $\iota$ is a subprocess of $\unf$ up to isomorphism.

\section{Decidability in Petri Games with Bad Markings}\label{sec:decidability}

We give the formal reduction from Petri games with a bounded number of system players, one environment player, and bad markings to Büchi games.
Let $\maxSys$ be the \emph{maximal number of system players} visible at the same time in the Petri game.

\subsection{Decision Tuples}

A \emph{decision tuple} for a player consists of an \emph{identifier}, a \emph{position}, a \emph{\typeTwo status}, a \emph{decision}, and a \emph{representation of the last mcut}.
It has type $\{1, \ldots, \maxSys\} \times \plS \times \{\False, \True, \End\} \times (\pom{\tr} \cup \{\top\}) \times \{1, \ldots, \maxSys\}$ for system places and type $\{0\} \times \plE \times \{\False\} \times \pom{\tr} \times \{0\}$ for environment places.

\begin{definition}[Decision tuples]%
	The set $\decisionsets_S$ of \emph{system decision tuples} is defined as  $\decisionsets_S = \{ (\id, p, b, T, K) \mid \id \in \{1, \ldots, \maxSys\} \land p \in \plS \land b \in \{\False, \True, \End\} \land (T \subseteq \post{}{p} \lor T = \top ) \land K \in \{1, \ldots, \maxSys\} \}$.
	The set $\decisionsets_E$ of \emph{environment decision tuples} is defined as $\decisionsets_E =\{ (0, p, \False, \post{}{p}, 0) \mid p \in \plE \}$.
	The set $\decisionsets_S$ of \emph{decision tuples} is defined as $\decisionsets = \decisionsets_S \cup \decisionsets_E$.
\end{definition}

We define the following functions to retrieve the respective elements of a decision tuple~$D$:
For $D=(\id, p, b, T, K)$, $\DSid$ obtains the first element representing the identifier, i.e., $\DSid(D)=\id$, $\DSpl$ obtains the second element representing the place, i.e., $\DSpl(D)=p$, $\DStypetwo$ the third element representing the \typeTwo status, i.e., $\DStypetwo(D)=b$, $\DSdec$ the fourth element representing the decision, i.e., $\DSdec(D)=T$, and $\DSlastmcut$ the fifth element representing the last mcut, i.e., $\DSlastmcut(D)=K$.

\subsection{Enabledness of Transitions from Decision Markings}

We use decision markings as subset of the set of decision tuples as representation of the marking of the Petri game.
Thus, we are only interested in decision markings that correspond to a reachable marking in the Petri game and where each identifier of players occurs at most once.
We define the underlying \emph{marking}~$\DSmarking(\bbD)$ of a decision marking~$\bbD$ as $\DSmarking(\bbD)(p) = |\{D \in \bbD \mid \DSpl(D) = p\}| $ for all places $p \in \pl$.

\begin{definition}[Decision markings]%
	The set $\decsetgame$ of \emph{decision markings corresponding to reachable markings in $\pGame$ and with unique identifiers of players} is defined as \[\decsetgame = \{ \bbD \subseteq \decisionsets \mid \exists M \in \reach(\pNet) : \DSmarking(\bbD) = M \land \forall i \in \{1, \ldots, \maxSys\} : |\{D \in \bbD \mid \DSid(D) = i\}| \leq 1 \}.\]
\end{definition}

We define the \emph{enabledness} of a transition $t$ from a decision marking~$\bbD$.
We first remove decision tuples from $\bbD$ that are not in the precondition of~$t$, that have $\top$ as their decision, that disallow $t$, or that have ended \typeTwo status.
Formally, we introduce the decision marking~$\bbD_{\pre{}{t}}$ to retain only decision tuples that represent a place in the precondition of $t$ and that allow $t$ as $\bbD_{\pre{}{t}} = \{ (\id, p, b, T, K) \in \decisionsets \mid (\id, p, b, T, K) \in \bbD \land t \in T \land p \in \pre{}{t} \land b \neq \End \}$.
Now, we can check the enabledness of a transition $t$ from a decision marking~$\bbD$ by $\pre{}{t} \subseteq \DSmarking(\bbD_{\pre{}{t}}) $.
We also define the decision marking~$\bbD_{t2}$ to retain only decision tuples with positive \typeTwo status as $\bbD_{t2}(\id, p, b, T, K) = \{ (\id, p, b, T, K) \in \decisionsets \mid (\id, p, b, T, K) \in \bbD \land b = \True \}$.
For the enabledness of a transition in the \typeTwo case, we introduce the decision marking~$\bbD_{\pre{}{t} \land t2}$ to retain only decision tuples with positive \typeTwo status that represent a place in the precondition of $t$ and that allow $t$ as $\bbD_{\pre{}{t} \land t2} = \bbD_{\pre{}{t}} \cap \bbD_{t2} $.

\subsection{Decision Markings Corresponding to Mcuts}

We call a decision marking~$\bbD$ \emph{corresponding to an mcut} when no $\top$ exists in $\bbD$, every transition with only system places in its precondition is either not enabled or not allowed by a participating system player in $\bbD$, and no positive \typeTwo status exists in $\bbD$.

\begin{definition}[Decision markings corresponding to an mcut]%
	A decision marking~$\bbD$ \emph{corresponds to an mcut} iff $ ( \forall D \in \bbD : \DSdec(D) \neq \top \land \DStypetwo(D) \neq \True ) \land  \forall t \in \tr : \pre{}{t} \not\subseteq \DSmarking(\bbD_{\pre{}{t}}) \lor \pre{}{t} \cap \plE \neq \emptyset $.
\end{definition}

This definition can be expressed simpler than in the original paper on Petri games~\cite{DBLP:journals/iandc/FinkbeinerO17} as the \typeTwo case is handled directly in the Büchi game.

\subsection{Backward Moves}\label{sec:backwardmoves}

We define \emph{backward moves} $\backwardrules$ as pairs of decision markings.
Backward moves can be based on transitions firing including at the start and at the end of the \typeTwo case.

\begin{definition}[Backward moves for transitions firing]%
	The set $\backwardrules_\tr$ of \emph{backward moves for transitions firing} is defined as
	\begin{align*}
	\backwardrules&_\tr= \{(\bbD, \bbD') \in \pom{\decisionsets_S} \times \pom{\decisionsets_S} \mid
	( \exists t \in \tr :
	\DSmarking(\bbD) = \pre{}{t} \land \DSmarking(\bbD') = \post{}{t} )~\land\\&
	( \forall D \in \bbD : t \in \DSdec(D)) \land
	( \forall D \in \bbD \cup \bbD' : \DSdec(D) \neq \top )~\land\\&
	 (( (\forall D \in \bbD  : \DStypetwo(D) = \False) \land (\forall D \in \bbD' : \DStypetwo(D) = \False \lor \DStypetwo(D) = \True) )~\lor\\&
	  ( (\forall D \in \bbD  : \DStypetwo(D) = \True) \land ((\forall D \in \bbD' : \DStypetwo(D) = \True) \lor (\forall D \in \bbD'  : \DStypetwo(D) = \End)) )) \}.
	\end{align*}
\end{definition}
Each backward move has to be based on a transition $t$.
All decision tuples of a backward move cannot have $\top$ as their decision.
The decision tuples before the transition further have to allow the transition.
All decision tuples on the left side of a backward move can have \typeTwo status $\False$ and all decision tuples on the right side have \typeTwo status $\False$ or $\True$.
Alternatively, all decision tuples on the left side of a backward move can have \typeTwo status $\True$ and all decision tuples on the right side have \typeTwo status $\True$ or $\End$.
Notice that \typeTwo status $\End$ cannot occur on the left side of a backward move.

\begin{definition}[Backward moves for the start and end of the \TypeTwo case]%
	The set $\backwardrules_{T2}$ of \emph{backward moves for the start and end of the \typeTwo case} is defined as
	\begin{align*}
	\backwardrules&_{T2}= \{(\bbD, \bbD') \in \pom{\decisionsets_S} \times \pom{\decisionsets_S} \mid\\&
	(
	(\forall D \in \bbD : \DStypetwo(D) = \False)\land
	(\forall D \in \bbD' : \DStypetwo(D) = \True)~\land\\&
	(\forall \id \in \{1, \ldots, \maxSys\}, p \in \plS, d \subseteq \post{}{t}, K \in \{1, \ldots, \maxSys\} :\\&
	(\id, p, d, \False, K) \in \bbD \Leftrightarrow (\id, p, d, \True, K) \in \bbD' )
	)~\lor\\&
	(
	(\forall D \in \bbD : \DStypetwo(D) = \True) \land
	(\forall D \in \bbD' : \DStypetwo(D) = \End)~\land\\&
	(\forall \id \in \{1, \ldots, \maxSys\}, p \in \plS, d \subseteq \post{}{t}, K \in \{1, \ldots, \maxSys\} :\\&
	\bbD((\id, p, d, \True, K)) = \bbD'((\id, p, d, \End, K)) )
	)
	 \}.
	\end{align*}
\end{definition}
Backward moves for the start and end of the \typeTwo case reverse the change of all decision tuples in a decision marking from negative to positive or from positive to ended \typeTwo status without changing the identifier, the positions, decisions of the decision tuples, or the last mcut.
This case becomes necessary because not all players of the \typeTwo case are in the precondition of the transition starting the \typeTwo case or the postcondition of the transitions ending the \typeTwo case.

\begin{definition}[Backward moves]%
	The set $\backwardrules$ of \emph{backward moves} is defined as $\backwardrules = \backwardrules_\tr \cup \backwardrules_{T2}$.
\end{definition}

\subsection{States in the Büchi Game}\label{sec:states}

\begin{definition}[States in the Büchi game]\label{def:BGstates}%
	The \emph{states $\TPstates$ in the Büchi game} are defined as \[\TPstates = \TPstates_\mathit{BN} \cup \decsetgame \times (\plS \rightarrow \{0, \ldots, k\}) \times (\backwardrules^*)^{\maxSys} \times \{1,\ldots,\maxSys \}\]
	with $\TPstates_\mathit{BN} = \{F_B, F_N\}$ ($B$ and $N$ stand for \emph{Büchi} and \emph{non-Büchi}).
\end{definition}

The sequences of backward moves $\backwardrules^*$ are of finite size because we do not allow useless repetition.
We will later see that they are limited to be of at most single exponential length in terms of the size of the Petri game.

We define the following functions to access the elements of a state $\TPoneState \in \TPstates \setminus \TPstates_\mathit{BN}$ in a Büchi game with $\TPoneState = (\bbD, M_{T2}, \BR_1, \ldots, \BR_{\maxSys}) $:
The function $\TPdecset$ obtains the first element representing the decision tuple, i.e., $\TPdecset(\TPoneState) = \bbD $.
The function $\TPtypeTwoMarking$ obtains the second element representing the marking that Player~0 claims will repeat itself in the \typeTwo case, i.e., $\TPtypeTwoMarking(\TPoneState) = M_{T2}$.
The function $\TPbackRulesGen_i$ obtains the next $\maxSys$ elements representing the sequences of backward moves, i.e., $\TPbackRulesGen_i(\TPoneState) = \BR_i $ for $i \in \{1, \ldots, \maxSys \} $.

\begin{definition}[States of Player~0 and states of Player~1]%
The \emph{states $\TPenvstates$ of Player~1} are defined as all states in $\TPstates$ where the decision marking corresponds to an mcut:
\[\TPenvstates = \{ \TPoneState \mid \TPdecset(\TPoneState) \text{ corresponds to an mcut} \}\]
The \emph{states $\TPsysstates$ of Player~0} are defined as all other states:
\[\TPsysstates = \TPstates \setminus \TPenvstates\]
\end{definition}

\begin{definition}[Initial state in the Büchi game]%
	The \emph{initial state $\TPinit$} in the Büchi game is defined as
	\[\TPinit = ( \bbD^S_\init \cup \{ (0, p_E, \False, \post{}{p_E}, 0) \mid p_E \in \init \cap \plE \}, \emptyset, []^{\maxSys})\]
	with $\bbD^S_\init \subseteq \pom{\decisionsets_S}$ such that
	\begin{align*}
		&(\DSmarking(\bbD^S_\init) =\init \setminus \plE) \land (\forall D \in \bbD^S_\init: \DStypetwo(D) = \False \land \DSdec = \top \land \DSlastmcut(D) = 1)~\land \\&
		(\forall i \in \{1, \ldots, \maxSys\} : |\bbD^S_\init \cap \{D \in \decisionsets_S  \mid \id(D) = i\}| \leq 1)
	\end{align*}
	and the system players in $\bbD^S_\init$ are ordered arbitrarily but fixed (e.g., lexicographically) according to their places and identifiers with the lowest possible sum are used.
\end{definition}
The \emph{initial state in the Büchi game} contains decision tuples for all places of the initial marking, an empty marking to repeat, and $\maxSys$ many empty sequences of backward moves.
The constraints on the the selection of unique identifiers of the players are only used to obtain a single initial state in the Büchi game.

\begin{definition}[Accepting states in the Büchi game]%
The \emph{accepting states $\TPfinal$ in the Büchi game} are defined as $F_B$ and all states corresponding to an mcut:
\[\TPfinal = \{F_B\} \cup \{ \TPoneState \in \TPstates \mid \TPdecset(\TPoneState) \text{ corresponds to an mcut} \} \]
\end{definition}

\subsection{Finite Winning and Losing Behavior in the Büchi Game}

We define finite winning and losing behavior in the Büchi game.
Finite losing behavior includes bad markings and nondeterministic decisions.
These two cases require to check the via backward moves reachable decision markings.
Therefore, we define a Petri net that has the via backward moves reachable decision markings as reachable markings.

\begin{definition}[Via backward moves reachable decision markings]%
Given a set of sequences of backward moves $\BR_1, \ldots, \BR_n$ and a decision marking~$\bbD_\init$, the $k$-bounded Petri net $\pNet_{\BR_1, \ldots, \BR_n}[\bbD_\init] = (\pl, \tr, \fl, \init)$ to calculate the \emph{via backward moves reachable decision markings} starting from $\bbD_\init$ is defined as follows.
\begin{itemize}
	\item We initialize $\pNet_{\BR_1, \ldots, \BR_n}[\bbD_\init]$ with $\pl = \bbD_\init$, $\tr = \emptyset$, $\fl = \emptyset$, and $\init = \bbD_\init$ and define the current marking $M = \bbD_\init$.
	\item We apply the following step recursively:
We check for the backward moves $(\bbD, \bbD')$ from the end of the sequences of backward moves $\BR_1, \ldots, \BR_n$ whether $\bbD'$ is contained in the current marking $M$ and whether all participating players in $\bbD'$ with identifier $\id$ have backward move $(\bbD, \bbD')$ at the end of $\BR_\id$.
If this is the case, then we add each decision tuple of $\bbD$ to $\pl$, $t_{(\bbD', \bbD)}$ to $\tr$, and $\{ (D', t_{(\bbD', \bbD)} ) \mid D' \in \bbD' \}$ and $\{ (t_{(\bbD', \bbD)},  D) \mid D \in \bbD \}$ to~$\fl$, we remove $(\bbD, \bbD')$ from the end of $\BR_\id$ for all participating players in $\bbD'$ with identifier $\id$, and we update the current marking $M$ by setting it to $(M - \bbD') + \bbD$.
We continue with the next backward move.
\end{itemize}
\end{definition}

Notice that we reverse the backward moves to apply standard forward algorithms to calculate the set of reachable markings.
Notice further that when more than one backward move is applicable from the current marking, then the second backward move is also applicable from the current marking obtained after adding the first backward move to the Petri net, because the backward moves are disjoint.
The algorithm terminates as the sequences of backward moves are finite and each iteration of the algorithm makes the sequences smaller.

\begin{definition}[Finite winning and losing behavior (non-\TypeTwo case)]%
\emph{Terminated states $\mathit{TERM}$}, \emph{deadlocked states $\mathit{DL}$}, \emph{states $\mathit{NDET}$ corresponding to nondeterministic decisions}, and \emph{states $\mathit{BAD}$ corresponding to a bad marking} are defined as follows:
\begin{align*}
	\TPstates' =& \{ \TPoneState \in \TPstates \mid \forall D \in \TPdecset(\TPoneState) : \DSdec(D) \neq \top \} \\
	\mathit{TERM} =& \{ \TPoneState \in \TPstates' \mid \forall t \in \tr : \pre{}{t} \not\subseteq \DSmarking(\TPdecset(\TPoneState)) \}\\
	\mathit{DL} =&  \{ \TPoneState \in \TPstates' \mid \forall t \in \tr : \pre{}{t} \not\subseteq \DSmarking(\TPdecset(\TPoneState)_{\pre{}{t}}) \}\\
	\mathit{NDET} =& \{ \TPoneState \in \TPstates' \mid \exists \bbD \in \reach(\pNet_{\TPbackRulesGen_1(\TPoneState), \ldots, \TPbackRulesGen_{\maxSys}(\TPoneState)}[\TPdecset(\TPoneState)]), D \in \bbD : \DSpl(D) \in \plS~\land\\& ((\exists t_1, t_2 \in \post{}{\DSpl(D)} : t_1 \neq t_2 \land \pre{}{t_1} \subseteq \DSmarking(\bbD_{\pre{}{t_1}})~\land\\& \pre{}{t_2} \subseteq \DSmarking(\bbD_{\pre{}{t_2}}) ) \lor (\exists t \in \post{}{\DSpl(D)} : \pre{}{t} \subset \DSmarking(\bbD_{\pre{}{t}}) )) \}\\
	\mathit{BAD} =& \{ \TPoneState \in \TPstates' \mid \exists \bbD \in \reach(\pNet_{\TPbackRulesGen_1(\TPoneState), \ldots, \TPbackRulesGen_{\maxSys}(\TPoneState)}[\TPdecset(\TPoneState)]) : \DSmarking(\bbD) \in \badmarkings \}
\end{align*}
\end{definition}
Only states where all $\top$ have been resolved are included.
A state is \emph{terminated} when no transition is enabled.
A state is \emph{deadlocked} when all transitions are not enabled or not allowed by enough decision tuples.
Terminated states correspond to winning behavior whereas deadlocked states that are not terminated correspond to losing behavior.

A state corresponds to a \emph{nondeterministic decision} if there is a via backward moves reachable decision marking from the state such that either two different transitions are enabled having the same system place in their precondition or if a single transition with a system place in its precondition is enabled with more than enough players allowing the transition.
In the second case, two different instances of the transition are enabled.

A state corresponds to a \emph{bad marking} if there is a via backward moves reachable decision marking from the state such that the underlying marking is in the set of bad markings.

\begin{definition}[Finite losing behavior (\TypeTwo case)]%
\emph{Deadlocked states $\mathit{DL}_{t2}$ in the \typeTwo case}, \emph{states $\mathit{SYNC}_{t2}$ corresponding to synchronization between decision tuples with positive and negative \typeTwo status}, and \emph{states corresponding to vanished decision tuples with positive \typeTwo status} are defined as follows:
	\begin{align*}
	\TPstates' =& \{ \TPoneState \in \TPstates \mid \forall D \in \TPdecset(\TPoneState) : \DSdec(D) \neq \top \} \\
	\mathit{DL}_{t2} =&  \{ \TPoneState \in \TPstates' \mid (\exists D \in \TPdecset(\TPoneState) : \DStypetwo(D) = \True) \land \forall t \in \tr :\\& \pre{}{t} \not\subseteq \DSmarking(\TPdecset(\TPoneState)_{\pre{}{t} \land t2}) \} \\
	\mathit{SYNC}_{t2} =& \{ \TPoneState \in \TPstates' \mid \exists \bbD \in \reach(\pNet_{\TPbackRulesGen_1(\TPoneState), \ldots, \TPbackRulesGen_{\maxSys}(\TPoneState)}[\TPdecset(\TPoneState)]), D \in \bbD :\\& \DStypetwo(D) = \True \land \exists t \in \post{}{\DSpl(D)} : \\& \pre{}{t} \not\subseteq \DSmarking(\bbD_{\pre{}{t} \land t2}) \land  \pre{}{t} \subseteq \DSmarking(\bbD_{\pre{}{t}})  \}\\
	\mathit{VAN}_{t2} =& \{ \TPoneState \in \TPstates' \mid \TPtypeTwoMarking(\TPoneState) \neq \emptyset \land \forall D \in \TPdecset(\TPoneState) : \DStypetwo(D) \neq \True \}
\end{align*}
\end{definition}
Only states where all $\top$ have been resolved are included.
A state is \emph{deadlocked in the \typeTwo case} when at least one decision tuple has positive \typeTwo status and all transitions are not enabled or not allowed from the decision marking with positive \typeTwo status.
Notice that this includes all \emph{in the \typeTwo case terminated} states.
We do not need to differentiate between terminated and deadlocked in the \typeTwo case because both cases are losing behavior.

A state corresponds to \emph{synchronization between decision tuples with negative and positive \typeTwo status} if there exists a via backward moves reachable decision marking from the state such that at least one decision tuple has positive \typeTwo status and it exists a transition from the postcondition of that place that is not enabled or not allowed by only decision tuples with positive \typeTwo status but is enabled and allowed when not regarding the distinction between decision tuples with negative and positive \typeTwo status.

A state corresponds to \emph{vanished decision tuples with positive \typeTwo status} when a nonempty marking to repeat exists but no more decision tuples with positive \typeTwo status exist.
This can occur when at least one transition with an empty postcondition exists that removes all decision tuples with positive \typeTwo status.

The definitions of states corresponding to nondeterministic decisions or to bad markings also apply to decision markings with positive \typeTwo status.

\subsection{Useless Repetitions in the Büchi Game}

We define useless repetitions in the Büchi game.
Therefore, we search for the repetition of loops in the graph of via backward moves reachable decision markings which did not increase the information of any participating player.
Such a useless repetition represents losing behavior to prevent the collection of sequences of backward moves of infinite length.

\begin{definition}[Useless repetition in the Büchi game]%
	For $j \in \{1,2,3,4\}$, we define $\BR^j_1, \ldots, \BR^j_{\maxSys}$ to be the corresponding backward moves to a decision marking $\bbD^j \in \reach(\pNet_{\TPbackRulesGen_1(\TPoneState), \ldots, \TPbackRulesGen_{\maxSys}(\TPoneState)}[\TPdecset(\TPoneState)])$ from the construction of $\pNet_{\TPbackRulesGen_1(\TPoneState), \ldots, \TPbackRulesGen_{\maxSys}(\TPoneState)}[\TPdecset(\TPoneState)]$.
	States $\mathit{UR}$ corresponding to a useless repetition in the Büchi game are defined as follows:
	\begin{align*}
		\textit{U}&\textit{R} = \{ \TPoneState \in \TPstates' \mid \exists \bbD^1, \bbD^2, \bbD^3, \bbD^4 \in \reach(\pNet_{\TPbackRulesGen_1(\TPoneState), \ldots, \TPbackRulesGen_{\maxSys}(\TPoneState)}[\TPdecset(\TPoneState)]) :\\& \bbD^1 = \bbD^2 = \bbD^3 = \bbD^4 \land \forall i \in \{1, \ldots, \maxSys\} : \BR^1_i \subset \BR^2_i \subseteq \BR^3_i \subset \BR^4_i~\land\\& \BR^2_i - \BR^1_i = \BR^4_i - \BR^3_i \}
	\end{align*}
\end{definition}
A useless repetition in the Büchi game occurs when the same loop happened twice without increasing the knowledge of the participating players.
The first loops occurs from $\bbD_1$ to $\bbD_2$ and the second one from $\bbD_3$ to $\bbD_4$ where $\bbD_2$ and $\bbD_3$ can be the same position in $ \reach(\pNet_{\TPbackRulesGen_1(\TPoneState), \ldots, \TPbackRulesGen_{\maxSys}(\TPoneState)}[\TPdecset(\TPoneState)])$.
In the definition, the knowledge of the players does not increase because the decision markings have to be the same which includes the last mcut.

\subsection{Edges in the Büchi Game}
\label{sec:edgesBG}

Winning behavior without a successor state (\textit{TERM}) leads to the unique accepting state $F_B$ with a self-loop.
Losing behavior leads to the unique non-accepting state $F_N$ with a self-loop.

\begin{definition}[Edges \textit{STOP} for finite winning and losing behavior]%
	Edges $\mathit{STOP}_B$ for \emph{finite winning behavior} without a successor state are defined as \[\mathit{STOP}_B = \{ (\TPoneState, F_B) \in \TPstates \times \{F_B\} \mid \TPoneState \in \mathit{TERM} \setminus (\mathit{BAD} \cup \mathit{UR} \cup \mathit{DL}_{t2} \cup \mathit{VAN}_{t2} )\}.\]
	Edges $\mathit{STOP}_N$ for \emph{finite losing behavior} are defined as
	\begin{align*}
	\mathit{STOP}_N = \{ (\TPoneState, F_N) \in \TPstates \times \{F_N\} \mid \TPoneState \in (&\mathit{DL} \setminus \mathit{TERM}) \cup \mathit{NDET} \cup \mathit{BAD} \cup \mathit{UR}~\cup\\& \mathit{DL}_{t2} \cup \mathit{SYNC}_{t2} \cup \mathit{VAN}_{t2} \} .
	\end{align*}
	Edges $\mathit{STOP}$ for \emph{finite winning and losing behavior} are defined as $\mathit{STOP} = \mathit{STOP}_B \cup \mathit{STOP}_N$.
\end{definition}

We introduce notation to calculate the successors of a decision marking:

\begin{definition}[\TypeTwo status change and $\top$ removal]%
	Given a decision marking~$\bbD$ without decision tuples with positive \typeTwo status, the function $\sucDEC(\bbD)$ identifies the set of all possible decision markings corresponding to $\bbD$ where the \typeTwo status can change and $\top$ is replaced by a set of allowed transitions for system places. It is defined as
	\begin{align*}
		\textit{s}&\textit{uc}_{\mathit{dec}}(\bbD) = \{ \bbD' \in \decsetgame \mid \bbD' = \{ (\id, p, b', T', K) \in \decisionsets \mid (\id, p, b, T, K) \in \bbD~\land\\&
		 (T \neq \top \Rightarrow T' = T) \land (T = \top \Rightarrow T' \subseteq \post{}{p})\land (p \in \plE \Rightarrow b' = b)~\land\\&
		 ( (\exists \id'', K'' \in \{1, \ldots, \maxSys\}, p'' \in \plS, T'' \subseteq \post{}{p} :\\& (\id'', p'', \End, T'', K'') \in \bbD) \Rightarrow b' = b )~\land\\&
		 ( (\forall \id'', K'' \in \{1, \ldots, \maxSys\}, p'' \in \plS, T'' \subseteq \post{}{p} :\\& (\id'', p'', \End, T'', K'') \notin \bbD) \Rightarrow b' \in \{\False, \True\} )\}~\land\\&
		 \forall \id \in \{1, \ldots, \maxSys\} : |\bbD \cap \{ D \in \decisionsets_S \mid \id(D) = \id \}| =
		 |\bbD' \cap \{ D \in \decisionsets_S \mid \id(D) = \id \}|
		 \}.
	\end{align*}
\end{definition}
We calculate all $\bbD'$ such that, for each element $(\id, p, b, T, K)$ in $\bbD$, there is exactly one element $(\id, p, b', T', K)$ in $\bbD$.
This is ensured by the last requirement because each identifier occurs at most once in $\bbD$.
For $T'$, $\top$ is replaced by a set of allowed transitions and remains the same otherwise.
For $p$ being an environment place, the \typeTwo status always remains $\False$.
For $p$ being a system place, the \typeTwo status can change from $\False$ to $\True$ unless there is a decision tuple in $\bbD$ with \typeTwo status $\End$.
In this case, the \typeTwo status remains the same.

All edges in the Büchi game go from a state of the form $\TPoneState = (\bbD, M_{T2}, \BR_1, \ldots, \BR_{\maxSys}) $ to a state $\TPoneState' = (\bbD', M_{T2}', \BR_1', \ldots, \BR_{\maxSys}') $ in the following.
We prevent states in $(\mathit{DL} \setminus \mathit{TERM}) \cup \mathit{NDET} \cup \mathit{BAD} \cup \mathit{UR} \cup \mathit{DL}_{t2} \cup \mathit{SYNC}_{t2} \cup \mathit{VAN}_{t2}$ to have further outgoing edges because they correspond to losing behavior.

\begin{definition}[Edges \textit{TOP}]%
Edges $\mathit{TOP}$ to let Player~0 make decisions and change \typeTwo status are defined as
\begin{align*}
\textit{T}&\textit{OP} = \{ (\TPoneState, \TPoneState')  \in \TPsysstates \times \TPstates \mid
\TPoneState = (\bbD, M_{T2}, \BR_1, \ldots, \BR_{\maxSys}) \land
\exists \bbD' \in \sucDEC(\bbD) :\\&
\TPoneState' = (\bbD', M_{T2}', \BR_1', \ldots, \BR'_{\maxSys}) \land
(\forall D \in \bbD : \DStypetwo(D) \neq \True) \land
M_{T2} = \emptyset~\land\\&
(\exists  D \in \bbD : \DSdec(D) =  \top) \land
M_{T2}' = \DSmarking(\bbD'_{t2}) \land
(\forall \id \in \{1, \ldots, \maxSys\} : \BR'_\id = \BR_\id)
\}.
\end{align*}
\end{definition}
In $\TPoneState$, no decision tuple with positive \typeTwo status exists in $\bbD$, the marking $M_{T2}$ to repeat in the \typeTwo case is empty, and a decision tuple in $\bbD$ has $\top$ as decision.
To obtain $\bbD'$ from $\bbD$, $\top$ is removed and system decision tuples can be designated as having positive \typeTwo status by choosing $\bbD'$ from $ \sucDEC(\bbD)$.
The underlying marking of $\bbD'$ restricted to decision tuples with positive \typeTwo status is stored in $M_{T2}'$.
It is the empty set when no decision tuple in $\bbD'$ has positive \typeTwo status.
All sequences of backward moves stay the same.

A transit relation $\tfl(t) \subseteq ( ((\pre{}{t} \cup \{\FLstart\}) \times (\post{}{t} \cup \{\FLend\})) \rightarrow  \mathbb{N} ) $~\cite{DBLP:conf/atva/FinkbeinerGHO19} (lifted to bounded Petri nets) relates places in the precondition with places in the postcondition of a transition.
It also indicates that a new token is created with $\FLstart$ and that a token is removed with $\FLend$.
In the following, we assume an arbitrary but fixed transit relation that preserves the type of players in the Petri game and represents the movement of players including their creation and removal.
This is needed to maintain the correct identifier for players.
In particular, the transit relation returns, for a decision marking $\bbD$, a player with identifier $\id$ at place $p$, and a transition $t$ to fire, the unique next place $p'$ for the player, written as $p \tfl(\bbD, \id, t) p'$.
We assume the existence of a function $\mathit{nextID}$ that, given a decision marking $\bbD$ and a transition $t$ fired from $\bbD$, returns a unique identifier for each new player $p'$ with $\FLstart \tfl(t) p'$.

\begin{definition}[Making decisions for successor decision markings]%
	For a decision marking $\bbD$ and a from there enabled and allowed transition $t$ with only system places in its precondition, the function $\sucS(\bbD, t)$ identifies the set of all possible decision markings corresponding to the postcondition of $t$ where no $\top$ exists.
	It is defined as
	\begin{align*}
		\textit{s}&\textit{uc}_S(\bbD, t) =  \{ \bbD' \in \decsetgame \mid \bbD' = \{(\id, p', \False, T', K') \in \decisionsets_S \mid
		(\id, p, \False, T, K) \in \bbD~\land\\&
		p \tfl(\bbD, \id, t) p' \land T' \subseteq \post{}{p'} \land
		K' = \max(\{ \DSlastmcut(D) \mid D \in \bbD_{\pre{}{t}} \})  \}~\cup\\&
		\{ (\id',p',\False, T', K') \in \decisionsets_S \mid
		\FLstart \tfl(t) p' \land
		T' \subseteq \post{}{p'}~\land\\&
		K' = \max(\{ \DSlastmcut(D) \mid D \in \bbD_{\pre{}{t}} \})\land
		\id' = \mathit{nextID}(\bbD, t,  \FLstart \tfl(t) p') \} \land
		\DSmarking(\bbD') = \post{}{t}~\land\\&
		\forall \id \in \{1, \ldots, \maxSys\} : |\bbD' \cap \{ D \in \decisionsets_S \mid \id(D) = \id \}| \leq 1 \}.
	\end{align*}
\end{definition}
The \typeTwo status is fixed to $\False$ as $\sucDEC$ is used afterward to give the possibility of changing the \typeTwo status from negative to positive.

\begin{definition}[Edges \textit{SYS}]%
Edges $\mathit{SYS}$ to let Player~0 fire a transition with no environment places in its precondition and afterward make decisions and change \typeTwo status are defined as
\begin{align*}
\textit{S}&\textit{YS} = \{ (\TPoneState, \TPoneState')  \in (\TPsysstates \setminus \mathit{NDET} \cup \mathit{BAD} \cup \mathit{UR}) \times \TPstates \mid
\TPoneState = (\bbD, M_{T2}, \BR_1, \ldots, \BR_{\maxSys})~\land\\&
\exists t \in \tr, \mathit{pc} \in \sucS(\bbD, t), \bbD' \in \sucDEC((\bbD \setminus \bbD_{\pre{}{t}}) \cup \mathit{pc} ) :  \\&
\TPoneState' = (\bbD', M_{T2}', \BR_1', \ldots, \BR'_{\maxSys}) \land
(\forall D \in \bbD : \DStypetwo(D) \neq \True) \land
M_{T2} = \emptyset~\land\\&
\bbD \text{ does \emph{not} correspond to an mcut} \land
\pre{}{t} \cap \plE = \emptyset \land
\pre{}{t} \subseteq \DSmarking(\bbD_{\pre{}{t}}) ~\land\\&
\bbD'' = \bbD' \setminus \{ (\id, p, \False, T, K ), (\id, p, \True, T, K )\mid (\id, p, b, T, K ) \in \mathit{pc}\}~\land\\&
M_{T2}' = \DSmarking(\bbD'_{t2}) \land
\forall \id \in \{1, \ldots, \maxSys\} :\\&
(( \exists p \in \plS, b \in \{\False, \True, \End\}, T \subseteq \post{}{p}, K \in \{1, \ldots, \maxSys\} :\\&
(\id, p, b, T, K) \in \mathit{pc} \cup (\bbD'' \setminus \bbD) ) \Rightarrow ((\bbD \setminus \bbD_{\pre{}{t}} \neq \bbD'' \Rightarrow \BR'_\id = \BR_\id~\cup\\& [(\bbD_{\pre{}{t}}, \mathit{pc}), (\bbD \setminus (\bbD_{\pre{}{t}} \cup (\bbD \cap \bbD'')), \bbD'' \setminus \bbD) ])~\land\\& (\bbD \setminus \bbD_{\pre{}{t}} = \bbD'' \Rightarrow \BR'_\id = \BR_\id \cup [(\bbD_{\pre{}{t}}, \mathit{pc})] ) ))~\land\\&
((\forall p \in \plS, b \in \{\False, \True, \End\}, T \subseteq \post{}{p}, K \in \{1, \ldots, \maxSys\} :\\&
(\id, p, b, T, K) \notin \mathit{pc} \cup \bbD'' \setminus (\bbD \cap \bbD'')) \Rightarrow \BR'_\id = \BR_\id )
\}.
\end{align*}
\end{definition}
In $\TPoneState$, no decision tuple with positive \typeTwo status exists in $\bbD$, the marking~$M_{T2}$ to repeat in the \typeTwo case is empty, and $\bbD$ does \emph{not} correspond to an mcut.
To obtain $\bbD'$ from $\bbD$, a transition $t$ with no environment places in its preconditions, that is enabled and allowed from $\bbD$, is chosen.
To simulate transition $t$ firing, the decision tuples $\mathit{pc}$ for the postcondition of $t$ are obtained by first choosing decisions via $\sucS$.
Afterward, decision tuples in $\bbD \setminus \bbD_{\pre{}{t}}$ and $\mathit{pc}$ can change their \typeTwo status from negative to positive via $\sucDEC$ to obtain $\bbD'$.
A corresponding backward move $(\bbD_{\pre{}{t}}, \mathit{pc})$ is added to $\BR_\id$ of all participating players with identifier $\id$ to obtain $\BR'_\id$.
If decision tuples in $\bbD$ have changed their \typeTwo status, a further backward move for this change is added.
Here, $\bbD''$ is used to remove $\mathit{pc}$ because the \typeTwo status of decision tuples in $\mathit{pc}$ might have changed.
All other backward moves are retained.
As in $\mathit{TOP}$, the underlying marking of $\bbD'$ restricted to decision tuples with positive \typeTwo status is stored in $M_{T2}'$.

We define the function $\sucMC$ that returns, for a decision marking $\bbD$ and a transition $t$, the decision marking with \typeTwo status $\False$ and decision~$\top$ corresponding to the system players participating in the transition.

\begin{definition}[Successor decision markings from an mcut]%
	For a decision marking $\bbD$ and a from there enabled and allowed transition $t$ with an environment place in its precondition, the function $\sucMC(\bbD, t)$ returns the decision marking corresponding to the system players of the postcondition of $t$ where $\False$ is the \typeTwo status and $\top$ is the decision. It is defined as
	\begin{align*}
		\textit{s}&\textit{uc}_\mathit{mcut}(\bbD, t) = \{ (\id, p', \False, \top, K') \in \decisionsets_S \mid (\id, p, \False, T, K) \in \bbD \land p \in \plS~\land\\& p \tfl(\bbD, \id, t) p' \land K' = \max(\{ \DSlastmcut(D) \mid D \in \bbD \setminus \bbD_{\pre{}{t}} \}) + 1\} ~\cup\\&
			\{ (\id',p',\False, \top, K') \in \decisionsets_S \mid
		\FLstart \tfl(t) p' \land
		K' = \max(\{ \DSlastmcut(D) \mid D \in \bbD_{\pre{}{t}} \}) + 1~\land\\&
		\id' = \mathit{nextID}(\bbD, t,  \FLstart \tfl(t) p') \}.
	\end{align*}
\end{definition}

The transit relation is followed and a new highest number for the last mcut is given to all participating players of the transition.
To stay in the range $\{1, \ldots, \maxSys\}$ for the last mcut, we need to reduce higher values when all decision tuples with a specific value are removed when a transition with an environment player in its precondition fires.
Therefore, we assume a function $\mathit{reduce}_\mathit{mcut}(\bbD)$ that only alters the last mcut in the decision marking $\bbD$ such that the order between decision tuples and their last mcut is maintained but the minimal number of last mcut numbers is used.
For example, when $\bbD$ contains decision tuples with last mcut $1$, $3$, $6$, and $7$, because all decision tuples with last mcut $2$, $4$, $5$, and $8$ participate in the transition with an environment place in its precondition, then all decision tuples with last mcut $1$ maintain $1$, all with $3$ get $2$, all with $6$ get $3$, and all with $7$ get $4$.

\begin{definition}[Edges \textit{MCUT}]%
	Edges $\mathit{MCUT}$ to let Player~1 fire a transition with an environment place in its precondition from an mcut are defined as
\begin{align*}
	\text{M}&\textit{CUT} = \{ (\TPoneState, \TPoneState')\in (\TPenvstates \setminus \mathit{NDET} \cup \mathit{BAD}) \times \TPstates \mid \exists t \in \tr :\\&
\TPoneState = (\bbD, M_{T2}, \BR_1, \ldots, \BR_{\maxSys})\land
\TPoneState' = (\bbD', M_{T2}', \BR_1', \ldots, \BR'_{\maxSys})~\land\\&
(\forall D \in \bbD : \DStypetwo(D) \neq \True) \land
\bbD \text{ corresponds to an mcut} \land
M_{T2}' = M_{T2} = \emptyset~\land\\&
\pre{}{t} \subseteq \DSmarking(\bbD_{\pre{}{t}}) \land
\bbD'' = \mathit{reduce}(\bbD \setminus \bbD_{\pre{}{t}}) \land
\bbD' = \bbD'' \cup \sucMC(\bbD'', t)~\cup\\& \{ (0, p, \False, \post{}{p}, 0) \mid p \in \post{}{t} \cap \plE \} \land
\forall \id \in \{ 1, \ldots, \maxSys \} :\\& \BR'_\id = \left\{\begin{array}{ll} \BR_\id &\text{if }\exists p \in \plS, b \in \{\False, \True, \End\}, T \subseteq \post{}{p},\\ & K \in \{1, \ldots, \maxSys\} : (\id, p, b, T, K) \in \bbD' \setminus \bbD'' \\ ~[] & \text{otherwise}\end{array}\right.
\}.
\end{align*}
\end{definition}
In $\TPoneState$, no decision tuple with positive \typeTwo status exists in $\bbD$, $\bbD$ corresponds to an mcut, and the marking $M_{T2}$ to repeat in the \typeTwo case remains the empty set.
To obtain $\bbD'$, a transition $t$ that is enabled and allowed from $\bbD$ is chosen.
This transition has an environment place $p$ in its precondition because $\bbD$ corresponds to an mcut.
The last mcut of decision tuples is \textit{reduce}d accounting for the possible free numbers by the removal of decision tuples due to $t$ firing to obtain $\bbD''$.
To simulate transition $t$ firing, decision tuples corresponding to the postcondition of $t$ are added to $\bbD''$ to obtain $\bbD'$.
All added decision tuples have negative \typeTwo status and $\top$ as decision for decision tuples corresponding to system players (ensured by $\sucMC$) and $\post{}{p}$ as decision for decision tuples corresponding to the environment player.
Backward moves for identifiers in $\bbD' \setminus \bbD''$ become empty.

We define three kinds of edges for the \typeTwo case: all three fire a transition only from decision tuples with positive \typeTwo status and the last two re-reach an already reached marking of the \typeTwo case.
For the first re-reaching, the marking is the marking to repeat and, for the second re-reaching, the marking is not.

For the \typeTwo case, we introduce the function $\sucTT$ that works exactly like the function $\sucS$ but it requires and produces decision tuples with positive \typeTwo status.

\begin{definition}[Edges \textit{NES}$_\mathit{fire}$]%
Edges $\mathit{NES}_\mathit{fire}$ let Player~0 fire a transition with no environment place in its precondition only from decision tuples with positive \typeTwo status. They let Player~0 make decisions while maintaining positive \typeTwo status and reach a new marking in the \typeTwo case. Edges $\mathit{NES}_\mathit{fire}$ are defined as
\begin{align*}
	\textit{N}&\textit{ES}_\mathit{fire} = \{ (\TPoneState, \TPoneState') \in (\TPsysstates \setminus \mathit{NDET} \cup \mathit{BAD} \cup \mathit{SYNC}_{t2} \cup \mathit{VAN}_{t2} ) \times \TPsysstates \mid\\&
\TPoneState = (\bbD, M_{T2}, \BR_1, \ldots, \BR_{\maxSys}) \land
\exists t \in \tr, \mathit{pc} \in \sucTT(\bbD, t):\\&
\TPoneState' = (\bbD', M_{T2}', \BR_1', \ldots, \BR'_{\maxSys}) \land
(\exists D \in \bbD :  \DStypetwo(D) = \True)~\land\\&
\pre{}{t} \cap \plE = \emptyset \land \pre{}{t} \subseteq  \DSmarking(\bbD_{\pre{}{t} \land t2}) \land
M_{T2}' = M_{T2}~\land\\&
\bbD' = (\bbD \setminus \bbD_{\pre{}{t} \land t2}) \cup \mathit{pc} \land
(\forall \bbD^M \in \reach(\pNet_{\BR_1, \ldots, \BR_{\maxSys}}[\bbD_{t2}]) : \DSmarking(\bbD'_{t2}) \neq \DSmarking(\bbD^M_{t2})  )~\land\\&
\forall \id \in \{1, \ldots, \maxSys\} : ( ( (\exists p \in \plS, T \subseteq \post{}{p}, K \in \{1, \ldots, \maxSys\} :\\& (\id, p, \False, T, K) \in \mathit{pc}) \Rightarrow \BR'_\id = \BR_\id \cup [(\bbD_{\pre{}{t} \land t2}, \mathit{pc})] )~\land\\& ( (\forall p \in \plS, T \subseteq \post{}{p}, K \in \{1, \ldots, \maxSys\} :\\& (\id, p, \False, T, K) \notin \mathit{pc}) \Rightarrow \BR'_\id = \BR_\id) )
\}.
\end{align*}
\end{definition}
In $\TPoneState$, a decision tuple with positive \typeTwo status exists in $\bbD$ and a transition $t$ is chosen such that the firing of this transition has to reach a decision marking with a new underlying marking for the \typeTwo case.
The transition $t$ cannot have an environment place in its precondition and is enabled and allowed from $\bbD$ restricted to decision tuples with positive \typeTwo status.
To simulate $t$ firing, decision tuples $\bbD_{\pre{}{t}\land t2}$ corresponding to the precondition of $t$ are removed from $\bbD$ and then decision tuples corresponding to the postcondition of $t$ are added.
All added decision tuples have positive \typeTwo status and a made decision via $\sucTT$.
It is ensured via the reachability of decision markings~$\bbD^M_{t2}$ and their underlying marking $\DSmarking(\bbD^M_{t2})$ from the backward moves that the marking in the \typeTwo case is not repeated.
The corresponding backward move $( \bbD_{\pre{}{t} \land t2} , \mathit{pc} )$ is added to the backward move of all participating players.
The remaining elements of the pair of states stay the same.

\begin{definition}[Edges \textit{NES}$_\mathit{finish}$]%
Edges $\mathit{NES}_\mathit{finish}$ let Player~0 fire a transition with no environment place in its precondition only from decision tuples with positive \typeTwo status. They re-reach the marking in which the \typeTwo case started and thereby end the \typeTwo case. Edges $\mathit{NES}_\mathit{finish}$ are defined as
\begin{align*}
	\textit{N}&\textit{ES}_\mathit{finish} = \{ (\TPoneState, \TPoneState') \in (\TPsysstates \setminus\mathit{NDET} \cup \mathit{BAD} \cup \mathit{UR} \cup \mathit{SYNC}_{t2} \cup \mathit{VAN}_{t2} ) \times \TPstates \mid\\&
\TPoneState = (\bbD, M_{T2}, \BR_1, \ldots, \BR_{\maxSys})~\land\\&
\exists t \in \tr, \mathit{pc} \in \sucTT(\bbD, t) :
\TPoneState' = (\bbD', M_{T2}', \BR_1', \ldots, \BR'_{\maxSys})~\land\\&
(\exists D \in \bbD :  \DStypetwo(D) = \True) \land
\pre{}{t} \cap \plE = \emptyset \land \pre{}{t} \subseteq  \DSmarking(\bbD_{\pre{}{t} \land t2})~\land\\&
\bbD'' = (\bbD \setminus \bbD_{\pre{}{t} \land t2}) \cup \mathit{pc} \land
\DSmarking(\bbD''_{t2} ) = M_{T2}\land
M_{T2}' = \emptyset~\land\\&
(\forall p \in M_{T2} : \exists \bbD^M \in \reach(\pNet_{\BR_1, \ldots, \BR_{\maxSys}}[\bbD_{t2}]) : p \notin \DSmarking(\bbD^M_{t2}))~\land\\&
\bbD' = \{ (\id, p, \End, T, K) \in \decisionsets_S \mid (\id, p, \True, T, K) \in \bbD' \}~\cup\\&
\{ (\id, p, \False, T, K) \in \decisionsets_S \mid (\id, p, \False, T, K) \in \bbD'\}~\land\\&
\bbD_\BR' = \{ (\id, p, \End, T, K) \in \decisionsets_S \mid (\id, p, \True, T, K) \in \mathit{pc} \}~\land\\&
\forall \id \in \{1, \ldots, \maxSys\} : (( \exists p \in \plS, T \subseteq \post{}{p}, K \in \{1, \ldots, \maxSys\} :\\& (\id, p, \True, T, K) \in ((\bbD' \setminus \bbD'' ) \cup \bbD'_\BR)) \Rightarrow (( (\bbD'' \setminus \mathit{pc}) \neq (\bbD' \setminus \bbD'_\BR) \Rightarrow\\& \BR'_\id = \BR_\id \cup \{ ( \bbD_{\pre{}{t} \land t2} , \bbD_\BR' ), (\bbD'' \setminus (\bbD' \cup \mathit{pc}) , \bbD' \setminus ( \bbD' \cup \bbD_\BR')) \} )~\land\\&
( (\bbD'' \setminus \mathit{pc}) = (\bbD' \setminus \bbD'_\BR) \Rightarrow \BR'_\id = \BR_\id \cup \{ ( \bbD_{\pre{}{t} \land t2} , \bbD_\BR' ) \} )))~\land\\&
( \forall p \in \plS, T \subseteq \post{}{p}, K \in \{1, \ldots, \maxSys\} :\\& (\id, p, \True, T, K) \notin ((\bbD' \setminus (\bbD'' \cap \bbD') ) \cup \bbD'_\BR ) \Rightarrow \BR'_\id = \BR_\id)
\}.
\end{align*}
\end{definition}
In $\TPoneState$, a decision tuple with positive \typeTwo status exists in $\bbD$ and a transition $t$ is chosen such that the firing of this transition reaches a decision marking with the same underlying marking $M_{T2}$ as when the \typeTwo case started.
The transition $t$ cannot have an environment place in its precondition and is enabled and allowed from $\bbD$ restricted to decision tuples with positive \typeTwo status.
Simulating $t$ firing would reach the decision marking~$\bbD''$ by removing $\bbD_{\pre{}{t}\land t2}$ from $\bbD$ and adding $\mathit{pc}$.
The underlying marking of $\bbD''$ restricted to decision sets with positive \typeTwo status has to be reached for the second time, i.e., it has to be $M_{T2}$.
Every player $p$ in $M_{T2}$ has to be at least once not included in an underlying marking of a decision marking~$\bbD^M_{t2}$ reachable by the backward moves, i.e., every player has to move at least once before completing the \typeTwo case.
To obtain $\bbD'$, all decision tuples with positive \typeTwo status in $\bbD''$ are set to ended \typeTwo status.
The marking $M_{T2}'$ is set to the empty set.

By $\bbD'_\BR$, we identify $\mathit{pc}$ with ended \typeTwo status for decision tuples with positive \typeTwo status.
When the change from $\bbD''$ to $\bbD'$ changed decision tuples not corresponding to the postcondition of $t$, i.e., not in $\mathit{pc}$ and $\bbD'_\BR$, then the backward moves $( \bbD_{\pre{}{t} \land t2} , \bbD_\BR' )$ and $(\bbD'' \setminus (\bbD' \cup \mathit{pc}) , \bbD' \setminus (\bbD'' \cup \bbD_\BR'))$ are added to $\BR_\id$ to obtain $\BR'_\id$ for participating players of the transition with identifier $\id$.
Otherwise, the second backward move relates the empty decision marking with the empty decision marking and therefore only the backward move $( \bbD_{\pre{}{t} \land t2} , \bbD_\BR' )$ is added to $\BR_\id$ to obtain $\BR'_\id$.
The other sequences of backward moves stay the same.
No new \typeTwo cases are necessary as one successful disclosure for this branch suffices.

\begin{definition}[Edges \textit{NES}$_\mathit{bad}$]%
	Edges $\mathit{NES}_\mathit{bad}$ corresponding to a bad finish of the \typeTwo case because the wrong marking is repeated or not all players moved are defined as
\begin{align*}
	\textit{N}&\textit{ES}_\mathit{bad} = \{ (\TPoneState, F_N) \in (\TPsysstates \setminus \mathit{NDET} \cup \mathit{BAD} \cup \mathit{SYNC}_{t2} \cup \mathit{VAN}_{t2} ) \times \{F_N\} \mid\\&
\TPoneState = (\bbD, M_{T2}, \BR_1, \ldots, \BR_{\maxSys})\land
\exists t \in \tr, \mathit{pc} \in \sucTT(\bbD, t) :\\&
(\exists D \in \bbD :  \DStypetwo(D) = \True ) \land
 \pre{}{t} \cap \plE = \emptyset \land \pre{}{t} \subseteq  \DSmarking(\bbD_{\pre{}{t} \land t2})~\land\\&
\bbD'' = (\bbD \setminus \bbD_{\pre{}{t} \land t2}) \cup \mathit{pc} \land
( \exists \bbD^M \in \reach(\pNet_{\BR_1, \ldots, \BR_{\maxSys}}[\bbD_{t2}]) : \DSmarking(\bbD''_{t2}) = \DSmarking(\bbD^M_{t2})  )~\land\\&
( \DSmarking(\bbD''_{t2} ) \neq M_{T2} \lor \exists p \in M_{T2} : \forall \bbD^M \in \reach(\pNet_{\BR_1, \ldots, \BR_{\maxSys}}[\bbD_{t2}]) : p \in \DSmarking(\bbD^M_{t2}) )
\}.
\end{align*}
\end{definition}
In $\TPoneState$, a transition $t$ with only system places in its precondition has to exist that is enabled and allowed from the decision tuples with positive \typeTwo status in $\bbD$.
The firing of the marking would reach the decision marking~$\bbD''$ by removing $\bbD_{\pre{}{t}\land t2}$ corresponding to the decision tuples for the precondition of $t$ and adding $\mathit{pc}$ corresponding to the decision tuples for the postcondition of $t$.
The underlying marking of the decision tuples with positive \typeTwo status in $\bbD''$ is reached for the second time, i.e., it is equal to the underlying marking of $\bbD^M_{t2}$ reachable via the backward moves from $\bbD_{t2}$ ($\bbD$ restricted to decision tuples with positive \typeTwo status).
The negative behavior occurs because the underlying marking is not equal to $M_{T2}$ or some place $p$ in $M_{T2}$ is part of the entire loop, i.e., part of all $\bbD^M_{t2}$.
This makes $\mathit{NES}_\mathit{finish}$ and $\mathit{NES}_\mathit{bad}$ mutual exclusive when a marking is repeated in the \typeTwo case.

\begin{definition}[Edges \textit{NES}]%
	The edges $\mathit{NES}$ for the \typeTwo case in the Büchi game are defined as $\mathit{NES}_\mathit{fire} \cup \mathit{NES}_\mathit{finish} \cup \mathit{NES}_\mathit{bad}$.
\end{definition}

\begin{definition}[Edges in the Büchi game]%
	The edges $\TPedges$ in the Büchi game are defined as $\TPedges = \mathit{TOP} \cup \mathit{SYS} \cup \mathit{MCUT} \cup \mathit{NES} \cup \mathit{STOP}$.
\end{definition}

\subsection{Correctness}\label{sec:BMcorrectness}

We prove assumptions correct that are used in the translation from bounded Petri games with one environment player and bad markings to Büchi games.

\begin{lemma}[Sufficiently many identifiers]
	There is always a free identifier when the number of system players increases from a transition firing.
\end{lemma}
\begin{proof}
	There are at most $\maxSys$ system players in the Petri game.
	When the number of system players can increase by $x$, then there are at most $\maxSys - x$ system players in the Petri game.
	As there are $\maxSys$ unique identifiers for system players, there are $x$ free identifiers.
\end{proof}

\begin{lemma}[Sufficiently many last mcuts]
	There is always a free last mcut when a transition with an environment place and a system place in its precondition fires.
\end{lemma}
\begin{proof}
	System players remember their last mcut and there are at most $\maxSys$ of them in the Petri game.
	Therefore, at most $\maxSys$ last mcuts can exist.
	Whenever a system player participates in a transition with the environment player, its last mcut becomes free, other last mcuts are reordered to keep their order, and a new last mcut can be added.
\end{proof}

We prove that winning strategies for the Büchi game can be translated to winning strategies for the system players in the encoded Petri game.
First, we define the translation.
Second, we define an equivalence relation between decision markings in the strategy for Player~0 in the Büchi game and cuts in the strategy for the system players in the Petri game.
Third, we use this equivalence relation to prove the translation correct.

\begin{definition}[Translation from Büchi game strategies to Petri game strategies]\label{def:BGtoPG}%
	In the following, we translate winning strategies for Player~0 in the Büchi game into winning strategies for the system players in the Petri game.
	By $f$, we identify the winning strategy for Player~0 in the Büchi game $\BuchiGame$ which simulates the Petri game $\pGameBadMarkings$.
	From the initial state $\TPinit$, the strategy $f$ generates a tree~$T_f$ which is possibly infinite.
	In~$T_f$, the nodes are labeled with states from $\TPstates$ and the root is labeled with $\TPinit$.
	A node~$N$ labeled with a $\TPsysstates$-state $\TPoneState$ has a unique successor node labeled with the state $f(w)$ with $w$ being the sequence of labels from the root to $N$.
	A node~$N$ labeled with a $\TPenvstates$-state $\TPoneState$ has for each successor state $(\TPoneState, \TPoneState') \in E$ a successor node in $T_f$ labeled with $\TPoneState'$.

	We traverse the nodes of $T_f$ in a breadth-first order and inductively construct a strategy $\sysstrat = ( \pNet^\sysstrat , \lambda^\sysstrat )$ for $\pGame$.
	We associate to each node labeled with a state $\TPoneState$ a cut $C$ in the strategy~$\sysstrat$ under construction such that $\lambda[C] = \DSmarking(\TPdecset(\TPoneState))$ holds, i.e., the cut and the decision marking of the state represent the same marking.
	To~$\TPinit$, we associate a cut $C_0$ labeled with the initial marking of $\pNet$, i.e., $\lambda[C_0] = \init$.
	Then, $\lambda[C_0] = \DSmarking(\TPdecset(\TPinit))$ holds.

	Suppose now the breadth-first traversal has reached a node in $T_f$ labeled with a $\TPsysstates$-state~$\TPoneState$ to which the cut $C$ is associated.
	Then, there is a unique successor node in $T_f$ labeled with~$\TPoneState'$.
	If the edge from $\TPoneState$ to $\TPoneState'$ resolves $\top$, nothing is added to the strategy, and $C$ is associated with~$\TPoneState'$ as well.
	If the edge from $\TPoneState$ to $\TPoneState'$ corresponds to a transition~$t$ with $\lambda(C_\mathit{pre}) = \pre{}{t}$ for $C_\mathit{pre} \subseteq C$ and $\TPoneState$ has less or equally many decision tuples with positive \typeTwo status than $\TPoneState'$, i.e., $\sum_{D \in \decisionsets \land \DStypetwo(D) = \True} \TPdecset(\TPoneState)(D) \leq \sum_{D \in \decisionsets \land \DStypetwo(D) = \True} \TPdecset(\mathcal{V'})(D) $, then we extend the strategy with a new transition $t'$ and new places $\mathit{NP}$ such that
	$\lambda[\mathit{NP}] = \post{}{t}$,
	$\lambda(t') = t $,
	$\pre{}{t'} = C_\mathit{pre}$, and
	$\post{}{t'} = \mathit{NP}$.
	We associate to $\TPoneState'$ the cut $C' = \mathit{NP} + (C - C_\mathit{pre})$ containing the new places $\mathit{NP}$ and the places in $C$ that did not participate in $t$.

	If the edge from $\TPoneState$ to $\TPoneState'$ corresponds to a transition $t$ with $\lambda(C_\mathit{pre}) = \pre{}{t}$ for $C_\mathit{pre} \subseteq C$ and $\TPoneState$ has a lower number of decision tuples with ended \typeTwo status than~$\TPoneState'$ as the marking to repeat in the \typeTwo case is reached for the second time, i.e.,
	$ \sum_{D \in \decisionsets \land \DStypetwo(D) = \End} \TPdecset(\TPoneState)(D) < \sum_{D \in \decisionsets \land \DStypetwo(D) = \End} \TPdecset(\mathcal{V'})(D) $%
	, a unique predecessor state~$\TPoneState^{1\mathit{st}}$ exists such that $\DSmarking(\TPdecset(\TPoneState)) - \pre{}{t} + \post{}{t} = \DSmarking(\TPdecset(\TPoneState^{1\mathit{st}})) $.
	The state $\TPoneState^{1\mathit{st}}$ represents the first occurrence of the marking $\DSmarking(\TPdecset(\TPoneState^{1\mathit{st}}))$ that is reached for a second time when firing $t$ from $\TPoneState$.
	By construction, only transitions from decision tuples with positive \typeTwo status fired between $\TPoneState^{1\mathit{st}}$ and $\TPoneState$.
	By $C^{1\mathit{st}}$, we identify the cut associated with $\TPoneState^{1\mathit{st}}$.
	We extend the strategy with a new transition $t'$ and new places $\mathit{NP}$ such that
	$\lambda[\mathit{NP}] = \post{}{t} $,
	$\lambda(t') = t$,
	$\pre{}{t'} = C_\mathit{pre}$, and
	$\post{}{t'} = \mathit{NP}$.
	From $C' = \mathit{NP} + (C - C_\mathit{pre})$ containing the new places $\mathit{NP}$ and the places in $C$ that did not participate in $t$, we repeat the strategy from $C^{1\mathit{st}}$ to $C'$ infinitely often with new places and new transitions.
	The breadth-first search needs to continue with the situation before the \typeTwo case.
	Therefore, we associate to $\TPoneState'$ the cut $C^{1\mathit{st}}$.

	Suppose now the breadth-first traversal has reached a node in $T_f$ labeled with a $\TPenvstates$-state~$\TPoneState$ to which the cut $C$ is associated.
	Then, for each successor node in $T_f$ labeled with $\TPoneState'$, a corresponding environment transition $t$ exists with $\lambda(C_\mathit{pre}) = \pre{}{t}$ for $C_\mathit{pre} \subseteq C$ and we extend the strategy with a new transition~$t'$ and new places $\mathit{NP}$ such that
	$\lambda[\mathit{NP}] = \post{}{t}$,
	$\lambda(t') = t $,
	$\pre{}{t'} = C_\mathit{pre} $, and
	$\post{}{t'} = \mathit{NP}$.
	We associate to $\TPoneState'$ the cut $C' = \mathit{NP} + (C - C_\mathit{pre})$ containing the new places $\mathit{NP}$ and the places in $C$ that did not participate in $t$.
\end{definition}

To prove that the constructed strategy from \refDef{BGtoPG} for the system players in the Petri game  is winning, we prove that the considered decision markings of every winning strategy for Player~0 in the Büchi game are equivalent to the reachable cuts in the corresponding strategy for the system players in the Petri game.
A decision marking is \emph{considered} if it is directly reachable or indirectly reachable via backward moves in a state in the Büchi game.
By definition, reaching $F_B$ or $F_N$ does not consider a decision marking.
The \emph{equivalence relation} between a decision marking and a cut checks that a place $p$ with decision $d$ is in the decision marking iff a place $p'$ in the cut exists such that $p'$ represents $p$ in the original Petri game, $\post{}{p'}$ represents $d$, and the causal past of both $p$ and $p'$ is equivalent.

\begin{definition}[Equivalence relation between decision markings and cuts]\label{def:equivDTcuts}%
	The equivalence relation $\equiv$ between decision markings and cuts is defined as
	\begin{align*}
		\bbD \equiv C &\iff ( ( \forall (p, b, T) \in \bbD : \exists^{=1} c \in C : p = \lambda(c) \land T = \lambda(\post{}{c})~\land\\& \past^\tr((p, b, T)) = \past^\tr(c) ) \land ( \forall c \in C : \exists^{=1} (p, b, T) \in \bbD :\\& p = \lambda(c) \land T = \lambda(\post{}{c}) \land \past^\tr((p, b, T)) = \past^\tr(c) ) )
	\end{align*}
where $\past^\tr(c)$ is the causal past represented by a sequence of transitions and $\past^\tr((p, b, T))$ is the causal past represented by a sequence of transition subtracted by the used backward moves.
\end{definition}
The function $\past^\tr(c)$ returns the sequence of transitions that leads from the initial marking $\init$ in $\pNet^\sysstrat$ to~$C$ following the order of the breadth-first search in \refDef{BGtoPG}.
We apply $\lambda$ to each transition to obtain the original transitions.
The function $\past^\tr((p, b, T))$ returns the sequence of transitions minus the used backward moves that match the edges that lead to the state $\TPoneState$ from which $(p, b, T)$ was reached via backward moves.
The equivalence of sequences is up to reordering of transitions with disjoint preconditions which can be necessary for $\past^\tr((p, b, T))$ as the order of $\past^\tr(c)$ is fixed by the breadth-first search.

\begingroup
\def\thelemma{\ref{lem:BGtoPG}}
\begin{lemma}[From Büchi game to Petri game strategies]
	If Player~0 has a winning strategy in the Büchi game, then there is a winning strategy for the system players in the Petri game.
\end{lemma}
\addtocounter{lemma}{-1}
\endgroup
\begin{proof}
	We translate a winning strategy for Player~0 in the Büchi game into a winning strategy for the system players in the Petri game as defined in \refDef{BGtoPG}.
	To prove that the constructed strategy is winning, we prove that the considered decision markings of every winning strategy for Player~0 in the Büchi game are equivalent as defined in \refDef{equivDTcuts} to the reachable cuts in the corresponding strategy for the system players in the Petri game.
	Then, losing behavior in the strategy for Player~0 in the Büchi game occurs if and only if losing behavior in the strategy for the system players in the Petri game occurs.
	We consider both directions for a winning strategy $f$ for Player~0 in the Büchi game and the corresponding strategy~$\sysstrat$ for the system players in the Petri game:
	\begin{itemize}
		\item We show that if a decision marking $\bbD$ is considered in a by $f$ reachable state of the Büchi game, then a reachable cut $C$ exists in $\sysstrat$ such that $ \bbD \equiv C $ holds.
			If~$\bbD$ is not considered via backward moves, then the word $w$ exists to reach the state of $\bbD$.
			The word $w^-$ is defined as $w$ without symbols $\top$.
			Firing transitions in the order of $w^-$ in the strategy leads to a cut $C$.
			By the strategy translation from \refDef{BGtoPG}, $w^-$ is applicable to $\sysstrat$ and, by the construction of the Büchi game, the precondition and postcondition of the used transitions are the same in both strategies.
			Therefore, $\bbD \equiv C$ holds.
			If $\bbD$ is reached via backward moves, then the word $w$ to reach the state of $\bbD$ and the word $w^\mathit{back}$ of applied backward moves exist.
			We define the subtraction of $w^\mathit{back}$ from $w^-$ from the end of both words.
			Firing transitions in the order of $w^- - w^\mathit{back}$ in the strategy leads to a cut $C$.
			By the strategy translation from \refDef{BGtoPG}, the backward moves can be subtracted from the word, $w^- - w^\mathit{back}$ is applicable to $\sysstrat$, and, by the construction of the Büchi game, the precondition and postcondition of the used transitions are the same in both strategies.
			Therefore, $\bbD \equiv C$ holds.
		\item We show that if a reachable cut $C$ exists in $\sysstrat$, then a decision marking is considered in a by $f$ reachable state of the Büchi game.
			The cut $C$ is reached via the sequence of transitions $T$.
			We only retain transitions with an environment place in their precondition in $T$.
			The choices of system players are already represented by $\sysstrat$.
			We start from the initial marking of $\sysstrat$ and obtain a sequence of transitions $w$ that is applicable to~$f$ and reaches~$C$ by recursively applying the following steps:
			If a \typeTwo case starts from the current cut (i.e., a minimal set of system players in the cut can play infinitely together without synchronizing with the environment player) and the underlying marking of the cut is the first marking to repeat in this \typeTwo case, then we add the transitions until the repetition of the marking occurs to $w$ and update the current cut accordingly.
			By the strategy translation from \refDef{BGtoPG}, at most on \typeTwo case can occur.
			Otherwise, we try to fire a transition with only system places in its precondition.
			If this is not possible, we try to fire a transition with an environment places in its precondition.
			We pick the environment transition corresponding to the first element of $T$ and afterward remove it from $T$.
			In both cases, the current cut becomes the cut reached after firing the transition.

			When none of these cases apply and there is a transition left in $T$, we need to add transitions with only system places in their precondition which are rolled back by backward moves to reach an equivalent decision marking to the cut~$C$.
			In this case, we fire a maximal sequence of transitions with only system places in their precondition such that only transitions with an environment place in the precondition are enabled.
			The enabled transitions include the remaining transition in~$T$.
			These fired transitions and all following transitions with only system places in their precondition are collected to be used as backward moves.
			We apply the recursive steps from before again.
			The first iteration removes the last transition from $T$ and afterward a finite number of transitions with only system places in their precondition are added.
			This procedure obtains $w'$.
			We add steps to $w'$ to make explicit decisions removing $\top$ according to the postcondition of the respective players and indicate the beginning of the \typeTwo case by changing the corresponding decision sets from negative to positive \typeTwo status to obtain $w$.
			By the construction of the Büchi game, $w$ is applicable to the Büchi game.
			Either $w$ directly leads to a decision marking~$\bbD$ such that $\bbD \equiv C$ holds or some of the collected backward moves and of the backward moves in the \typeTwo case are applied to the state in the Büchi game leading to a decision marking~$\bbD'$ such that $\bbD' \equiv C$ holds.
	\end{itemize}

	It remains to show that it suffices to only apply backward moves from one state, i.e., at most one transition can be left in $T$.
	By way of contradiction, assume that there are two states when backward moves are applied.
	Either the first earlier application is also possible at the second later application because the corresponding system players have not moved or performing the first backward moves makes it impossible to continue to the second backward move because the environment player is part of a transition in between and the used backward move is negated.
	
	It also remains to show that the constructed strategy for the Petri game is, in fact, a strategy, i.e., satisfies justified refusal and is deterministic and deadlock-avoiding.
	The strategy is deterministic and deadlock-avoiding because \emph{all} reachable markings are tested to exactly fulfill these requirements by construction.
	The usage of backward moves ensures that all reachable markings are checked.
	Justified refusal is fulfilled because system players can only disallow transitions from the postcondition in the original Petri game.
	Meanwhile, transitions with only an environment place in their precondition are always added.
	Therefore, every possible transition not in the Petri game has at least one system place in its precondition that universally forbids this transition in the original Petri game.   
	Because every reachable marking in the strategy is checked to have only deterministic decisions, we know that it is impossible to overlook transitions that could violate justified refusal when building the strategy from the Büchi game.
	Remember that the Büchi game fires transitions with the environment player as late as possible at mcuts, restricting the order in which transitions are added.

	As we showed that a strategy for Player~0 in a Büchi game and the corresponding strategy for the system players in a Petri game visit equivalent decision markings and cuts, we can conclude that for a winning strategy for Player~0 in a Büchi game no losing situation in the strategy for the system players in the Petri game can occur.
	Therefore, the strategy for the system players in the Petri game is winning.
\end{proof}

Before we can prove the reverse direction, we have to show that it suffices to only consider \emph{minimal} winning strategies for the system players in Petri games when translating them into strategies for Player~0 in the Büchi game.
This is necessary because for not minimal but winning strategies for the system players in Petri games, the backward moves in the corresponding Büchi game might falsely classify the corresponding strategy for Player~0 in the Büchi game as losing.
We prove that this cannot occur for minimal winning strategies and that all winning strategies can be made minimal while maintaining their winningness.

We define \emph{minimal strategies for the system players in Petri games} based on \emph{useless repetitions}~\cite{DBLP:conf/fsttcs/Gimbert17} which are defined for control games~\cite{DBLP:conf/icalp/GenestGMW13} played on asynchronous automata~\cite{DBLP:journals/ita/Zielonka87}.

Before defining our useless repetitions formally, we recall the prefix~\cite{DBLP:journals/acta/KhomenkoKV03} of a Petri net or a Petri game and define the suffix of a prefix.
\begin{definition}[Prefix~\cite{DBLP:journals/acta/KhomenkoKV03}]\label{def:prefix}%
	A branching process $\iota_\mathit{pre} = (\pNet^\iota_\mathit{pre}, \lambda^\iota_\mathit{pre})$ is a \emph{prefix} of a branching process $\iota = (\pNet^\iota, \lambda^\iota)$, denoted $\iota_\mathit{pre} \sqsubseteq_\mathit{pre} \iota$, if $\pNet^\iota_\mathit{pre}$ is a subnet of $\pNet^\iota$ (i.e., $\pl^\iota_\mathit{pre} \subseteq \pl^\iota$, $\tr^\iota_\mathit{pre} \subseteq \tr^\iota$, $\init^\iota_\mathit{pre} = \init^\iota$, and $\fl^\iota_\mathit{pre} = \fl^\iota \upharpoonright \pl^\iota_\mathit{pre} \times \tr^\iota_\mathit{pre} \cup \tr^\iota_\mathit{pre} \times \pl^\iota_\mathit{pre}$) such that if $t \in \tr^\iota_\mathit{pre}$ and $(p, t) \in \fl^\iota$ or $(t, p) \in \fl^\iota$, then $p \in \pl^\iota_\mathit{pre}$, if $p \in \pl^\iota_\mathit{pre}$ and $(t, p) \in \fl^\iota$, then $t \in \tr^\iota_\mathit{pre}$, and $\lambda^\iota_\mathit{pre}$ is the restriction of $\lambda^\iota$ to $\pl^\iota_\mathit{pre} \cup \tr^\iota_\mathit{pre}$.
\end{definition}

\begin{definition}[Suffix]\label{def:suffix}%
	A branching process $\iota_\mathit{suf} = (\pNet^\iota_\mathit{suf}, \lambda^\iota_\mathit{suf})$ is a \emph{suffix} of a branching process $\iota = (\pNet^\iota, \lambda^\iota)$
	if there exists a subset $M$ of a reachable marking in $\pNet^\iota$, $\pNet^\iota_\mathit{suf} = \pNet^\iota[M]$ without unreachable places, transitions, and flows, and $\lambda^\iota_\mathit{suf}$ is the restriction of $\lambda^\iota$ to $\pl^\iota_\mathit{suf} \cup \tr^\iota_\mathit{suf}$.
\end{definition}

We require that $\lambda^\sysstrat$, $\lambda^\iota_\mathit{pre}$, and $\lambda^\iota_\mathit{suf}$ are restricted corresponding to the Petri net they belong to.
The \emph{future} in $\pNet$ of a node $x$ in $\pNet$ is defined as the set $\mathit{fut}(x) = \{ y \in \pl \cup \tr \mid x \leq y \}$.

\begin{definition}[Useless repetition]\label{def:uselessrepetition}%
	A finite prefix $\iota_\mathit{pre} = (\pNet^\iota_\mathit{pre}, \lambda^\iota_\mathit{pre})$ is a \emph{useless repetition} of a winning strategy $\sysstrat = (\pNet^\sysstrat, \lambda^\sysstrat)$ iff there exists a suffix $\iota_\mathit{suf} = (\pNet^\iota_\mathit{suf}, \lambda^\iota_\mathit{suf})$ of the prefix $\iota_\mathit{pre}$ such that
	\begin{enumerate}[(1)]
		\item\label{uselessRepONE} the suffix contains only system places and the definition of the prefix has either removed all or no outgoing transitions of the places in the suffix (i.e., $\forall p \in \pl^\iota_\mathit{suf} : \forall t \in \tr : (p, t) \in \fl \Rightarrow ((p, t) \in \fl^\iota_\mathit{pre} \lor \post{\iota_\mathit{pre}}{p} = \emptyset ) $),
		\item\label{uselessRepTWO} the initial marking and the final marking of the suffix represent the same marking in the original Petri game (i.e., $ \exists^{=1} M \in \reach(\pNet^\iota_\mathit{suf}) : (\forall t \in \tr^\iota_\mathit{suf} : \pre{\iota_\mathit{suf}}{t} \not\subseteq M) \land \lambda^\iota_\mathit{suf}[\init^\iota_\mathit{suf}] = \lambda^\iota_\mathit{suf}[M] $),
		\item\label{uselessRepTHREE} the decisions of the system players in the initial marking of the suffix differ from the decisions of the system players in the final marking of the suffix (i.e., $ \exists^{=1} M \in \reach(\pNet^\iota_\mathit{suf}) : (\forall t \in \tr^\iota_\mathit{suf} : \pre{\iota_\mathit{suf}}{t} \not\subseteq M) \land \exists p^\iota_\mathit{suf} \in \init^\iota_\mathit{suf}, p \in M : \lambda(p^\iota_\mathit{suf}) = \lambda(p) \land \exists t^\iota_\mathit{suf} \in \post{\sysstrat}{p^\iota_\mathit{suf}} : \pre{\iota_\mathit{suf}}{t^\iota_\mathit{suf}} \subseteq \init^\iota_\mathit{suf} \land \forall t \in \post{\sysstrat}{p} : \lambda(t^\iota_\mathit{suf}) = \lambda(t) \Rightarrow \pre{\sysstrat}{t} \not\subseteq M $), and
		\item\label{uselessRepFOUR} the suffix only contains synchronizations between the players when the synchronization also occurred before the suffix but after the last synchronization with the environment player (i.e., $\forall t \in \tr^\iota_\mathit{suf} : \pre{\iota_\mathit{suf}}{t} > 1 \Rightarrow \forall : p_1, p_2 \in \init^\iota_\mathit{suf} \cap \past^{\iota_\mathit{pre}}(t) : \exists t^\iota_\mathit{pre} \in \tr^\iota_\mathit{pre} : \mathit{fut}^{\pNet^\iota_\mathit{pre}}(t^\iota_\mathit{pre}) \cap \plE^{\iota_\mathit{pre}} = \emptyset \Rightarrow t^\iota_\mathit{pre} \in \past^{\iota_\mathit{pre}}(p_1) \land t^\iota_\mathit{pre} \in \past^{\iota_\mathit{pre}}(p_2) $).
	\end{enumerate}
\end{definition}
We utilize useless repetitions to simplify the behavior of system players in the \typeTwo case and between successors of mcuts and the next mcut (cf.\ Condition~(\ref{uselessRepONE})).
Therefore, our application of useless repetitions is independent of the environment.
Condition~(\ref{uselessRepTWO}) carries over from~\cite{DBLP:conf/fsttcs/Gimbert17} to our usage.
We add Condition~(\ref{uselessRepTHREE}) to handle infinite strategies in the case of Petri games.
It ensures that the suffix is not infinitely often repeated immediately after the suffix.
This would lead to infinitely many applications of useless repetitions at one position.
We prevent this by ensuring that the transition after the suffix is \emph{not} the same as at the beginning of the suffix.
We add Condition~(\ref{uselessRepFOUR}) to require that a useless repetition only contains exchange of causal past via synchronization between the system players if this exchange already occurred before the useless repetition but after the respective last successor of an mcut of the system players.
Notice that a suffix of a useless repetition has a unique final marking and the removal of outgoing transitions can only happen at the unique final marking of the suffix.

Our usage of useless repetitions is far more restricted than for the original definition~\cite{DBLP:conf/fsttcs/Gimbert17}.
There it is used to decide the existence of a winning strategy in \emph{decomposable games} by bounding the size of strategies without useless repetitions.
The original definition of useless repetitions further requires that the future is the same when skipping and not skipping the suffix.
For our approach, this is fulfilled by construction.

\begin{definition}[Minimal strategies in Petri games]%
	A \emph{minimal strategies for the system players in Petri games} has no useless repetitions.
\end{definition}

We skip useless repetitions in winning strategies by shortcuts.

\begin{definition}[Shortcut]\label{def:shortcut}%
	A strategy $\iota_\mathit{short} = (\pNet^\iota_\mathit{short}, \lambda^\iota_\mathit{short})$ is a \emph{shortcut} via useless repetition $\iota_\mathit{pre} = (\pNet^\iota_\mathit{pre}, \lambda^\iota_\mathit{pre})$ with suffix $\iota_\mathit{suf} = (\pNet^\iota_\mathit{suf}, \lambda^\iota_\mathit{suf})$ of a strategy $\iota = (\pNet^\iota, \lambda^\iota)$ iff
	$\pl^\iota_\mathit{short} = \pl^\iota \setminus (\pl^\iota_\mathit{suf} \setminus \init^\iota_\mathit{suf})$,
	$\tr^\iota_\mathit{short} = \tr^\iota \setminus \tr^\iota_\mathit{suf}$,
	$\fl^\iota_\mathit{short}$ satisfies $\fl^\iota_\mathit{suf} \cap \fl^\iota_\mathit{short} = \emptyset \land \exists^{=1} M \in \reach(\pNet^\iota_\mathit{suf}) : (\forall t \in \tr^\iota_\mathit{pre} : \pre{\iota_\mathit{pre}}{t} \not\subseteq M) \land  \forall p \in M : \forall (p, t) \in \fl^\iota : (p, t) \notin \fl^\iota_\mathit{short} \land \exists^{=1} p^\iota_\mathit{suf} \in \init^\iota_\mathit{suf} : \lambda(p^\iota_\mathit{suf}) =  \lambda(p) \land (p^\iota_\mathit{suf}, t) \in \fl^\iota_\mathit{short}$,
	$\init^\iota_\mathit{short} = \init^\iota$, and
	$\lambda^\iota_\mathit{short}$ is the restriction of $\lambda^\iota$ to $\pl^\iota_\mathit{short} \cup \tr^\iota_\mathit{short}$.
\end{definition}
A shortcut is obtained by removing the suffix of a useless repetition from a given strategy. 
We identify the reachable final marking of the suffix by~$M$.
The behavior following this marking in the original strategy is pulled forwarded to the initial marking of the suffix.
All places, transitions, and flows of the suffix between its initial marking and the marking~$M$ are removed.

In \cite{DBLP:conf/fsttcs/Gimbert17}, removing useless repetitions by shortcuts is applied only to finite strategies because the winning condition is termination in final states whereas we also apply it to infinite strategies.
This is no problem because it is only applied to the finite parts between successors of mcuts and the next mcut and to system players in the \typeTwo case.
Therefore, we might require infinitely many shortcuts but each of them is well-defined.
We prove that every winning strategy can be reduced to a minimal winning strategy by removing all useless repetitions.

\begin{lemma}[Minimal winning strategies]
	Removing a useless repetition via a shortcut from a winning strategy in a Petri game results again in a winning strategy in the Petri game.
\end{lemma}
\begin{proof}
	The construction of the shortcut only removes transitions with only system places in their precondition, that are not synchronizing with other players and have exchanged their history before the removed part, and then continues with the existing future.
	Therefore, the strategy is still deterministic and satisfies justified refusal.
	Fewer markings (no new markings) are reached and therefore the shortcut also avoids bad markings.
\end{proof}

Removing all useless repetitions and updating the winning strategy accordingly in each step gives a minimal winning strategy.

Before we can prove the reverse direction, we also have to show that unnecessary \typeTwo cases in strategies for the system players in Petri games can be removed.
The Büchi game allows at most one \typeTwo case per branch by preventing the change from negative to positive \typeTwo status for decision tuples from the point onwards that a decision tuple with ended \typeTwo status exists in the decision marking.
We use the following definition to identify unnecessary \typeTwo cases in strategies for the system players in Petri games.

\begin{definition}[System place for the \TypeTwo case]%
	Given a strategy $\sysstrat$ for the system players in a Petri game and a cut $C$ in the Petri game, the function $\mathit{nes}_\sysstrat(C)$ calculates all minimal sets of system places from $C$ that satisfy the \typeTwo case in the Büchi game and is defined as
	\begin{align*}
		\mathit{nes}_\sysstrat(C)&= \{ C_{t2} \subseteq C \mid
	(\forall p \in C_{t2} : (\forall p_E \in \plE^\sysstrat : p_E \not\leq p \Rightarrow \forall q \in \mathit{fut}^\sysstrat(p) :\\& p_E \not\leq q ) \land \mathit{fut}^\sysstrat(p) \text{ is infinite} ) \land
	(\exists t_1, \ldots, t_n \in \tr^\sysstrat : \exists C_1, \ldots, C_n \subseteq \pl^\sysstrat :\\& C_{t2} \firable{t_1} C_1 \firable{t_2} \ldots \firable{t_n} C_n \land \DSmarking(C_{t2}) = \DSmarking(C_n) \land \forall i \in {1, \ldots, n-1} :\\& \DSmarking(C_{t2}) \neq \DSmarking(C_i)) \land
	( \forall C_{t2}' \subseteq C_{t2} : C_{t2} \setminus C_{t2}' \notin \mathit{nes}_\sysstrat(C))~\land\\&
	( \forall p \in C_{t2} : \exists i \in \{1,\ldots,n\} : p \notin C_i )
	\} .
	\end{align*}
\end{definition}
We use $n$ to search for repetitions of markings in the underlying Petri net.
Therefore, $n$ can be upper bounded by the exponential size of the reachability graph of the underlying Petri net.
The cut~$C_{t2}$ corresponding to the \typeTwo case cannot synchronize with the environment player and has to have an infinite future.
Further, there has to be a sequence of transitions such that the marking of $C_{t2}$ is repeated after firing the sequence but not during firing it.
Additionally, $C_{t2}$ has to be minimal in the sense that no places can be removed from $C_{t2}$ while $C_{t2}$ remains a cut corresponding to the \typeTwo case.
At last, each place has to be left at least once between $C$ and $C_{t2}$ at a position $i$.

We remove unnecessary \typeTwo cases in strategies for the system players in Petri games:
\begin{lemma}[One \TypeTwo case per branch]\label{lem:oneTTperBranch}
	For each winning strategy $\sysstrat$ for the system players in a bounded Petri game with one environment player and bad markings, there exists a winning strategy $\sysstrat'$ for the system players with, for each reachable cut, at most one minimal set of system players playing infinitely together without synchronizing with the environment player.
	The strategy $\sysstrat'$ has \emph{at most one \typeTwo case per branch} in the reachability graph.
\end{lemma}
\begin{proof}
	We initialize $\sysstrat'$ with the initial marking of $\sysstrat$.
	We fire enabled transitions from the initial marking in $\sysstrat$ in breadth-first search manner to obtain all reachable cuts.
	For each reached cut $C$ (including the initial marking), we calculate $\mathit{nes}_\sysstrat(C)$.
	We copy fired transitions $t$ and their flow and postcondition to $\sysstrat'$ unless either the precondition of $t$ is not part of $\sysstrat'$ or $\mathit{nes}_\sysstrat(C)$ contains at least one set, the precondition of $t$ is a subset or equal to one set in $\mathit{nes}_\sysstrat(C)$, and one of the two following conditions holds:
	\begin{enumerate}
		\item The fired transition $t$ is not added if there exists a disjoint \typeTwo case in the past of $C$.
		\item The fired transition $t$ is not added if there exists no disjoint \typeTwo case in the past of $C$ but an outgoing transition from another set in $\mathit{nes}_\sysstrat(C)$ is already added to $\sysstrat'$.
	\end{enumerate}

	This construction allows for each reachable cut at most one minimal set of system players playing infinitely together without synchronizing with the environment player from the point onward that the terminated system players have the same position of the environment player in their causal past, i.e., at the latest after the first repetition of a marking for this set of system players playing infinitely together without synchronizing with the environment player.
	The winningness of $\sysstrat$ translates to $\sysstrat'$ because the termination of the system players is deadlock-avoiding as another set of system players plays infinitely.
	Furthermore, fewer markings are reached which is winning for the winning condition of bad markings.
\end{proof}

In the following, we assume that minimal strategies for the system players in a Petri game have at most one \typeTwo case per branch in the reachability graph.
We can now prove the reverse direction from strategies for the players in Petri games to strategies for Player~0 in the corresponding Büchi game.
We define the translation from strategies for the players in Petri games to strategies for Player~0 in the corresponding Büchi game and prove it correct.

\begin{definition}[Translation from Petri game strategies to Büchi game strategies]\label{def:PGtoBG}%
	In the following, we translate minimal winning strategies for the system players in the Petri game into winning strategies for Player~0 in the Büchi game.
	Given a winning strategy $\sysstrat = (\pNet^\sysstrat, \lambda^\sysstrat)$ of $\pGame$ and a prefix $w \in \TPstates^* \cdot \TPsysstates$ of a play of $\TPgame$ simulating $\pGame$, we compute the choice $f(w)$ of a strategy $f$ for $\TPgame$.
	Let~$w^-$ result from $w$ by deleting all states containing $\top$.
	Starting with the initial marking $\init^\sysstrat$ of $\pNet^\sysstrat$ and firing the sequence of transitions corresponding to $w^-$ in $\sysstrat$, we arrive at a cut $C$ of $\sysstrat$, which by definition of $\TPsysstates$ does not correspond to an mcut.
	We assume a function $\mathit{cut}(\id, p, C)$ that returns for an identifier $\id \in \{1, \ldots, \maxSys\}$ and a place~$p$ in the Petri game the corresponding place $p_\mathit{cut}$ in $C$ of $\sysstrat$.

	There are three cases: the first one resolves $\top$, the second one fires a transition with only system places in its precondition not in the \typeTwo case, and the third one does so in the \typeTwo case.
	For the three cases, we apply the definitions of $\mathit{TOP}$, $\mathit{SYS}$, and $\mathit{NES}$ for a given fired transition and a  by $\mathit{nes}_\sysstrat(C)$ given (possibly empty) set of multisets over system places for the \typeTwo case.
	Due to \refLemma{oneTTperBranch}, $\mathit{nes}_\sysstrat(C)$ contains at most one set and if it contains a set no \typeTwo case occurred in the past of the current cut $C$.

	\smallskip
	\noindent
	\emph{Case 1} ($\mathit{TOP}$). If $\mathit{last}(w) \in \TPsysstates$ contains $\top$ in the decision marking, we make decisions and start the \typeTwo case if $\mathit{nes}_\sysstrat(C)$ is not empty.
		The state $\mathit{last}(w)$ does not contain any decision tuples with positive \typeTwo status in the decision marking.
		Formally, we are at a state $ \mathit{last}(w) = (\bbD, M_{T2}, \BR_1, \ldots, \BR_{\maxSys}) $ with corresponding cut $C$.
		We build the successor state $ (\bbD', M_{T2}', \BR_1', \ldots, \BR_{\maxSys}') $ that our strategy picks.
		We pick the unique $C_{t2} \in \mathit{nes}_\sysstrat(C)$ if $\mathit{nes}_\sysstrat(C)$ is not empty.
		Otherwise, $C_{t2} = \emptyset$.
		We define $\bbD' = \{ (\id, p, b, \lambda(\post{\sysstrat}{p_\mathit{cut}}), K) \in \decisionsets \mid (\id, p, b, T, K) \in \bbD \land \mathit{cut}(p, \id, C) = p_\mathit{cut} \land p_\mathit{cut} \in C \setminus C_{t2} \} \cup
		\{ (\id, p, \True, \lambda(\post{\sysstrat}{p_\mathit{cut}}), K) \in \decisionsets \mid (\id, p, \False, T, K) \in \bbD \land \mathit{cut}(p, \id, C) = p_\mathit{cut} \land p_\mathit{cut} \in C_{t2} \} $.
		If $C_{t2}$ is not the empty set, then $M_{T2}' = \DSmarking(\bbD_{t2}')$ and the remaining elements stay the same.
		If $C_{t2}$ is the empty set, then all elements except for $\bbD'$ stay the same.

	\smallskip
	\noindent
	\emph{Case 2} ($\mathit{SYS}$). If $\mathit{last} \in \TPsysstates$ contains no $\top$ and no decision tuple in the decision marking has positive \typeTwo status, an enabled transition fires to the cut~$C'$ and the \typeTwo case is started if $\mathit{nes}_\sysstrat(C)$ is not empty.
		As~$C$ is not an mcut, some transition with only system places in its precondition is enabled in $C$.
		We choose one of the enabled transitions and call it $t$.
		Let $t$ lead to a successor cut $C'$.
		Formally, we are at a state $ \mathit{last}(w) = (\bbD, M_{T2}, \BR_1, \ldots, \BR_{\maxSys}) $ with corresponding cut $C$.
		We build the successor state $ (\bbD', M_{T2}', \BR_1', \ldots, \BR_{\maxSys}') $ that our strategy picks.
		We pick the unique $C_{t2} \in \mathit{nes}_\sysstrat(C')$ if $\mathit{nes}_\sysstrat(C')$ is not empty.
		Otherwise, $C_{t2} = \emptyset$.
		We define $\bbD' = \{ (\id, p, b', T, K) \in \decisionsets \mid (\id, p, b, T, K) \in \bbD \land (p \notin \pre{}{t} \lor t \notin T) \land \mathit{cut}(p', \id) = p_\mathit{cut} \land (p_\mathit{cut} \in (C \setminus C_{t2}) \Rightarrow b' = b) \land (p_\mathit{cut} \in C_{t2} \Rightarrow b' = \True ) \} \cup
		\{ (\id, p', b', \lambda(\post{\sysstrat}{p_\mathit{cut}}), K') \in \decisionsets \mid (\id, p, b, T, K) \in \bbD \land p\tfl(\bbD, \id, t)p' \land p \in \pre{}{t} \land t \in T \land  \mathit{cut}(p', \id) = p_\mathit{cut} \land (p_\mathit{cut} \in (C \setminus C_{t2}) \Rightarrow b' = \False) \land (p_\mathit{cut} \in C_{t2} \Rightarrow b' = \True ) \land K' = \max(\{\DSlastmcut(D) \mid D \in \bbD_{\pre{}{t}} \}) \} \cup
		\{ (\id', p', b', \lambda(\post{\sysstrat}{p_\mathit{cut}}), K') \in \decisionsets \mid \FLstart \tfl(t)p' \land \id' = \mathit{nextID}(\bbD, t,  \FLstart \tfl(t) p') \land  \mathit{cut}(p', \id) = p_\mathit{cut} \land (p_\mathit{cut} \in (C \setminus C_{t2}) \Rightarrow b' = \False) \land (p_\mathit{cut} \in C_{t2} \Rightarrow b' = \True ) \land K' = \max(\{\DSlastmcut(D) \mid D \in \bbD_{\pre{}{t}} \} $.
		The backward moves of all participating players are updated accordingly, i.e., we add $ (\bbD_{\pre{}{t}} , \bbD'_{\post{}{t}} ) $ and potentially the backward move for all changes of \typeTwo status of decision tuples that did not participate in~$t$.
		If $C_{t2}$ is not the empty set, then $M_{T2}' = \DSmarking(\bbD_{t2}')$ and the remaining elements stay the same.
		If $C_{t2}$ is the empty set, then all remaining elements stay the same.

	\smallskip
	\noindent
	\emph{Case 3} ($\mathit{NES}$). If $\mathit{last} \in \TPsysstates$ contains a decision tuple with positive \typeTwo status in the decision marking, an enabled transition fires and we finish the \typeTwo case if the indicated marking to repeat in the \typeTwo case is reached.
	The state $\mathit{last}(w)$ does not contain any $\top$ in the decision marking.
	As at least one decision tuple in the decision marking from $\mathit{last}$ corresponding to cut $C$ has positive \typeTwo status, some transition with only system places in its precondition is enabled in~$C$ and has a future of infinite length and without synchronization with the environment player.
	We choose one of the these transitions and call it~$t$.
	Let $t$ lead to a successor cut $C'$.
	Formally, we are at a state $ \mathit{last}(w) = (\bbD, M_{T2}, \BR_1, \ldots, \BR_{\maxSys}) $ with corresponding cut $C$.
	We build the successor state $ (\bbD', M_{T2}', \BR_1', \ldots, \BR_{\maxSys}') $ that our strategy picks.
	Two cases can occur: either $M_{T2}$ is reached for the second time or not.
	\begin{itemize}
		\item If $M_{T2}$ is \emph{not} reached for the second time, then we define $\bbD' = ( \bbD \setminus \bbD_{\pre{}{t} \land t2} ) \cup \{ (\id, p', \True, \lambda(\post{\sysstrat}{p_\mathit{cut}}), K') \in \decisionsets \mid (\id, p, b, T, K) \in \bbD \land p\tfl(\bbD, \id, t)p' \land p \in \pre{}{t} \land t \in T \land \mathit{cut}(p', \id) = p_\mathit{cut} \land K' = \max(\{\DSlastmcut(D) \mid D \in \bbD_{\pre{}{t}} \})\} \cup\\
			\{ (\id', p', \True, \lambda(\post{\sysstrat}{p_\mathit{cut}}), K') \in \decisionsets \mid \FLstart \tfl(t)p' \land K' = \max(\{\DSlastmcut(D) \mid D \in \bbD_{\pre{}{t}} \}) \land \mathit{cut}(p', \id) = p_\mathit{cut} \land \id' = \mathit{nextID}(\bbD, t,  \FLstart \tfl(t) p') \} $, i.e., transition $t$ fires, decision tuples for places in the postcondition of $t$ retain positive \typeTwo status, and these decision tuples make their decision according to the structure of the strategy.
			In this case, we update the backward moves, i.e., we add $ (\bbD_{\pre{}{t}} , \bbD'_{\post{}{t}} ) $ to $\BR_\id$ to obtain $\BR_\id'$ for participating player with identifier $\id$.
			The remaining elements stay the same.
		\item If $M_{T2}$ is reached for the second time, we perform the same step as for not reaching $M_{T2}$ for the second time and update all decision tuples with positive \typeTwo status to ended \typeTwo status, i.e., for $\bbD''$ obtained from the previous step,
			$\bbD' = \{ (\id, p, b', T, K) \in \decisionsets \mid (\id, p, b, T, K) \in \bbD'' \land (b = \False \Rightarrow b' = \False) \land (b = \True \Rightarrow b' = \End) \}$.
			We set $M'_{T2}$ to the empty set and update the backward moves as in the definition of $\mathit{NES}_\mathit{finish}$.
			The remaining elements stay the same.
	\end{itemize}
\end{definition}

\begingroup
\def\thelemma{\ref{lem:PGtoBG}}
\begin{lemma}[From Petri game to Büchi game strategies]
	If the system players have a winning strategy in the Petri game, then there is a winning strategy for Player~0 in the Büchi game.
\end{lemma}
\addtocounter{lemma}{-1}
\endgroup
\begin{proof}
	By making strategies for the system players in the Petri game minimal, this proof is analogous to the proof of \refLemma{BGtoPG}.
	We translate a minimal winning strategy for the system players in the Petri game into a winning strategy for Player~0 in the Büchi game as defined in \refDef{PGtoBG}.
	To prove that the constructed strategy is winning, we prove that the reachable cuts of every minimal winning strategy for the system players in the Petri game are equivalent as defined in \refDef{equivDTcuts} to the considered decision markings in the corresponding strategy for Player~0 in the Büchi game.
	Then, losing behavior in the strategy for the system players in the Petri game occurs if and only if losing behavior in the strategy for Player~0 in the Büchi game occurs.
	We consider both directions for a minimal winning strategy~$\sysstrat$ for the system players in the Petri game and the corresponding strategy $f$ for Player~0 in the Büchi game:
	\begin{itemize}
		\item We show that if a reachable cut $C$ exists in $\sysstrat$, then a decision marking is considered in a by $f$ reachable state of the Büchi game.
			The cut $C$ is reached via the sequence of transitions $T$.
			We only retain transitions with an environment place in their precondition in $T$.
			The choices of system players are already represented by $\sysstrat$.
			We start from the initial marking of $\sysstrat$ and obtain a sequence of transitions $w$ that is applicable to~$f$ and reaches~$C$ by recursively applying the following steps:
			If a \typeTwo case starts from the current cut (i.e., a minimal set of system players in the cut can play infinitely together without synchronizing with the environment player) and the underlying marking of the cut is the first marking to repeat in this \typeTwo case, then we add the transitions until the repetition of the marking occurs to $w$ and update the current cut accordingly.
			By \refLemma{oneTTperBranch}, at most on \typeTwo case can occur.
			Otherwise, we try to fire a transition with only system places in its precondition.
			If this is not possible, we try to fire a transition with an environment places in its precondition.
			We pick the environment transition corresponding to the first element of $T$ and afterward remove it from $T$.
			In both cases, the current cut becomes the cut reached after firing the transition.

			When none of these cases apply and there is a transition left in $T$, we need to add transitions with only system places in their precondition which are rolled back by backward moves to reach an equivalent decision marking to the cut~$C$.
			In this case, we fire a maximal sequence of transitions with only system places in their precondition such that only transitions with an environment place in the precondition are enabled.
			The enabled transitions include the remaining transition in~$T$.
			These fired transitions and all following transitions with only system places in their precondition are collected to be used as backward moves.
			We apply the recursive steps from before again.
			The first iteration removes the last transition from $T$ and afterward a finite number of transitions with only system places in their precondition are added.
			This procedure obtains $w'$.
			We add steps to $w'$ to make explicit decisions removing $\top$ according to the postcondition of the respective players and indicate the beginning of the \typeTwo case by changing the corresponding decision sets from negative to positive \typeTwo status to obtain $w$.
			By the construction of the Büchi game, $w$ is applicable to the Büchi game.
			Either $w$ directly leads to a decision marking~$\bbD$ such that $\bbD \equiv C$ holds or some of the collected backward moves and of the backward moves in the \typeTwo case are applied to the state in the Büchi game leading to a decision marking~$\bbD'$ such that $\bbD' \equiv C$ holds.
		\item We show that if a decision marking $\bbD$ is considered in a by $f$ reachable state of the Büchi game, then a reachable cut $C$ exists in $\sysstrat$ such that $ \bbD \equiv C $ holds.
			If~$\bbD$ is not considered via backward moves, then the word $w$ exists to reach the state of $\bbD$.
			The word $w^-$ is defined as $w$ without symbols $\top$.
			Firing transitions in the order of $w^-$ in the strategy leads to a cut $C$.
			By the strategy translation from \refDef{PGtoBG}, $w^-$ is applicable to $\sysstrat$ and, by the construction of the Büchi game, the precondition and postcondition of the used transitions are the same in both strategies.
			Therefore, $\bbD \equiv C$ holds.
			If $\bbD$ is reached via backward moves, then the word $w$ to reach the state of $\bbD$ and the word $w^\mathit{back}$ of applied backward moves exist.
			We define the subtraction of $w^\mathit{back}$ from $w^-$ from the end of both words.
			Firing transitions in the order of $w^- - w^\mathit{back}$ in the strategy leads to a cut $C$.
			By the strategy translation from \refDef{PGtoBG}, the backward moves can be subtracted from the word, $w^- - w^\mathit{back}$ is applicable to $\sysstrat$, and, by the construction of the Büchi game, the precondition and postcondition of the used transitions are the same in both strategies.
			Therefore, $\bbD \equiv C$ holds.
	\end{itemize}

	The same argument as in the proof of \refLemma{BGtoPG} applies to show that it suffices to only apply backward moves from one state, i.e., at most one transition can be left in $T$.

	The constructed strategy for the Büchi game is a strategy as it chooses one successor for every state of Player~0. 
	Also, by construction, all successors are added to states of Player~1. 

	As we showed that a strategy for the system players in a Petri game and the corresponding strategy for Player~0 in a Büchi game visit equivalent cuts and decision markings, we can conclude that for a minimal winning strategy for the system players in the Petri game an accepting state in the strategy for Player~0 in a Büchi game is visited infinitely often.
	This is because $F_B$ is visited infinitely often for finite plays and for infinite plays at most one \typeTwo case occurs needing only finitely many states and a state representing an mcut is visited infinitely often as the environment player is part of fired transitions infinitely often.
	Therefore, the strategy for the system players in the Petri game is winning.
\end{proof}

\begingroup
\def\thetheorem{\ref{theo:gameSolving}}
\begin{theorem}[Game solving]
	For Petri games with a bounded number of system players, one environment player, and bad markings, the question of whether the system players have a winning strategy is decidable in \emph{2-EXPTIME}.
	If a winning strategy for the system players exists, it can be constructed in double exponential time.
\end{theorem}
\addtocounter{theorem}{-1}
\endgroup
\begin{proof}
	We establish the complexity by estimating the size of the set $\TPstates$ of states in the Büchi game.
	We recall from \refDef{BGstates} that the states $\TPstates$ in the Büchi game are defined as $\TPstates = \TPstates_\mathit{BN} \cup \decsetgame \times (\plS \rightarrow \{0, \ldots, k\}) \times (\backwardrules^*)^{\maxSys} \times \{1,\ldots,\maxSys \}$.
	The underlying Petri net is $k$-bounded by some $k \geq 1$.
	We use binary encodings to identify places.
	This results in $\log_2$ occurring in the size estimate.
	Here, we assume that $\log_2$ returns numbers greater or equal to one, i.e., $\log_2(x) = \max(1, log^{\textit{usual}}_2(x)) $ where $\log^{\textit{usual}}_2(x)$ is the usual logarithm with base two.
	We represent a decision tuple $(\id, p, b, T, K)$ by $\lceil\log_2(|\maxSys|)\rceil$ Boolean variables to identify the identifier $\id$, $\lceil\log_2(|\pl|)\rceil$ Boolean variables to identify the place $p$, two Boolean variables for the three-valued flag $b$, a Boolean variable indicating the presence of $\top$, and, for each transition $t \in \tr$, a Boolean variable indicating the presence of $t$ in the decision~$T$, and $\lceil\log_2(|\maxSys|)\rceil$ Boolean variables to identify the last mcut $K$.
	There are at most $\maxSys + 1$ players.
	The additional player is the environment player.
	The maximal number of system players $\maxSys$ includes the possibility of $k$ players in one place.
	This gives us $(\maxSys + 1) \cdot (\lceil\log_2(|\pl|)\rceil + 3 + |\tr| + 2 * \lceil\log_2(|\maxSys|)\rceil)$ Boolean variables.
	Considering all valuations of these variables, the size of the set $\decsetgame$ of decision markings can be bounded by $A = 2^{(\maxSys + 1) \cdot (\lceil\log_2(|\pl|)\rceil + 3 + |\tr| + 2 * \lceil\log_2(|\maxSys|)\rceil)}$.

	The marking to repeat in the \typeTwo case can be represented with $\maxSys \cdot \lceil\log_2(|\plS|)\rceil$ Boolean variables.
	Thus, the size of the set of markings to repeat in the \typeTwo case can be bounded by $B= 2^{\maxSys \cdot \lceil\log_2(|\plS|)\rceil}$.

	Backward moves are pairs of decision markings.
	We need $2 \cdot \maxSys \cdot (\lceil\log_2(|\plS|)\rceil + 3 + |\tr|)$ Boolean variables to represent a backward move.
	The number of backward moves for each of the $\maxSys$ sequences can be bound by the triangular number of the maximum number of system players $\maxSys \cdot (\maxSys + 1) / 2 $ times the number of markings~$2^{|\pl|}$.
	The maximal length of a loop can be bounded by the number of markings~$2^{|\pl|}$ and in each loop at least one system player learns about at least one new last mcut.
	After $(\maxSys \cdot (\maxSys + 1) / 2) \cdot |\tr|$ loops, all system players have to be maximally informed and further loops lead to a useless repetition, i.e., the collection of backward moves can be stopped.
	Therefore, the number of backward moves in the sets of backward moves can be bounded by $C =  2^{(\maxSys \cdot (\maxSys + 1) / 2) \cdot 2^{|\pl|} \cdot 2 \cdot \maxSys \cdot (\lceil\log_2(|\plS|)\rceil + 3 + |\tr|)}$.

	The size of the set $\TPstates$ of states in the Büchi game can be bounded by $A \cdot B \cdot C$.
	This  equals a double exponential number of states in the size of the Petri game dominated by the factor $C$.
	As Büchi games can be solved in polynomial time~\cite{DBLP:conf/soda/ChatterjeeH12}, the total time required to construct and solve the Büchi game is double exponential in the size of the Petri game.

	Büchi games are memoryless determined.
	The winning strategy can therefore be represented by a finite graph~$G_f$ whose size is bounded by the size of the Büchi game.
	We construct a finite deadlock-avoiding winning strategy for the system players in the Petri game following the construction from \refDef{BGtoPG}, using $G_f$ instead of the infinite strategy tree $T_f$.
	In the \typeTwo case, we represent the strategy for the system players finitely.
	When the marking to repeat in the \typeTwo case is reached for the second time, we do not add an infinite unrolling of the strategy for the system players in the \typeTwo case.
	Instead, we change the flow between transitions and places where the infinite unrolling would be added to the first occurrence of the place in the \typeTwo case.
	We remove all thereby unreachable places and represent the \typeTwo case with a finite strategy.
	When $G_f$ reaches a state for the second time, i.e., a loop occurs with the environment player, we change the flow from transitions to places, where the strategy would continue as above, such that they reach the cut corresponding to the first occurrence of the state in the Büchi game.
	We remove all thereby unreachable places and represent a loop with the environment player with a finite strategy.
\end{proof}

\begin{figure}[t]
	\centering
	\begin{subfigure}[t]{0.49\textwidth}
	\centering
	\begin{tikzpicture}[label distance=-0.5mm]
		\node [sysplace] (s1) [tokens=1, label=left:{\tiny$s_1$}] {};
		\node [sysplace] (s2) [tokens=1, right of=s1, label=left:{\tiny$s_2$}] {};
		\node [envplace] (e1) [tokens=1, right of=s2, label=right:{\tiny$e_1$}] {};
		\node [sysplace] (s3) [tokens=1, right of=e1, label=right:{\tiny$s_3$}] {};
		\node [sysplace] (s4) [below of=s1, below of=s1, label=left:{\tiny$s_4$}] {};
		\node [sysplace] (s5) [right of=s4, label=left:{\tiny$s_5$}] {};
		\node [envplace] (e2) [right of=s5, label=right:{\tiny$e_2$}] {};
		\node [sysplace] (s6) [right of=e2, label=right:{\tiny$s_6$}] {};
		\node [sysplace] (s7) [below of=s4, below of=s4, label=left:{\tiny$s_7$}] {};
		\node [sysplace] (s8) [right of=s7, label=left:{\tiny$s_8$}] {};
		\node [envplace] (e3) [right of=s8, label=right:{\tiny$e_3$}] {};
		\node [sysplace] (s9) [right of=e3, label=right:{\tiny$s_9$}] {};
		\node [sysplace] (s10) [right of=s9, label=right:{\tiny$s_{10}$}] {};
		\node [transition] 	(t1)  [below of=s1, label=left:{\tiny$t_1$}] {}
			edge [pre]  (s1)
			edge [post] (s4);
		\node [transition] 	(t2)  [below of=s2, label=left:{\tiny$t_2$}] {}
			edge [pre]  (s2)
			edge [post] (s5);
		\node [transition] 	(t3)  [below of=e1, xshift=0.5cm, label=right:{\tiny$t_3$}] {}
			edge [pre]  (e1)
			edge [pre]  (s3)
			edge [post] (e2)
			edge [post] (s6);
		\node [transition] 	(t4)  [below of=s5, label=left:{\tiny$t_4$}] {}
			edge [pre]  (s4)
			edge [pre]  (s5)
			edge [pre]  (e2)
			edge [post] (s7)
			edge [post] (s8)
			edge [post] (e3);
		\node [transition] 	(t5)  [below of=s6, label=right:{\tiny$t_5$}] {}
			edge [pre]  (s6)
			edge [post] (s9);
		\node [transition] 	(t6)  [below of=s6, right of=s6, label=right:{\tiny$t_6$}] {}
			edge [pre]  (s6)
			edge [post] (s10);
	\end{tikzpicture}
	\subcaption{As $\{s_1, s_5, e_2, s_9\}$ and $\{s_2, s_4, e_2, s_{10}\}$ are bad markings, no deadlock-avoiding strategy for the system players can avoid the bad markings.}
	\label{fig:badmarking}
	\end{subfigure}%
	~
	\begin{subfigure}[t]{0.49\textwidth}
	\centering
	\begin{tikzpicture}[label distance=-0.5mm]
		\node [sysplace] (s1) [tokens=1, label=left:{\tiny$s_1$}] {};
		\node [envplace] (e1) [right of=s1, label=left:{\tiny$e_2$}] {};
		\node [envplace] (e2) [below of=e1, below of=e1, label=right:{\tiny$e_3$}] {};
		\node [envplace] (e3) [tokens=1, right of=e2, label=right:{\tiny$e_1$}] {};
		\node [envplace] (e4) [below of=e3, below of=e3, xshift=-0.5cm, label=right:{\tiny$e_4$}] {};
		\node [envplace] (e5) [left of=e2, left of=e2, xshift=-0.5cm, label=left:{\tiny$e_5$}] {};
		\node [sysplace] (s2) [tokens=1, below of=e2, left of=e2, xshift=-0.5cm, label=above:{\tiny$s_2$}] {};
		\node [sysplace] (s3) [left of=s2, label=left:{\tiny$s_1'$}] {};
		\node [sysplace] (s4) [below of=s2, label=above:{\tiny$s_1''$}] {};
		\node [transition] 	(t1)  [below of = s1, xshift=0.5cm, label=left:{\tiny$t_3$}] {}
			edge [pre]  (s1)
			edge [pre]  (e1)
			edge [post] (e2)
			edge [post] (s3);
		\node [transition] 	(t2)  [right of = t1, label=left:{\tiny$t_1$}] {}
			edge [pre]  (e3)
			edge [post] (e1);
		\node [transition] 	(t3)  [below of = e2, label=left:{\tiny$t_5$}] {}
			edge [pre]  (e2)
			edge [post] (e4);
		\node [transition] 	(t4)  [below of = e3, label=left:{\tiny$t_2$}] {}
			edge [pre]  (e3)
			edge [post] (e4);
		\node [transition] 	(t5)  [left of = e2, label=below:{\tiny$t_4$}] {}
			edge [pre]  (e2)
			edge [post] (e5);
		\node [transition] 	(t6)  [below of = s3, label=left:{\tiny$t_6$}] {}
			edge [pre]  (s3)
			edge [post] (s4);
		\node [transition] 	(t7)  [right of = t6, right of = t6, label=above:{\tiny$t_7$}] {}
			edge [pre]  (e4)
			edge [post] (s4);
		\draw[thick, shorten >=1pt,<->,shorten <=1pt] (s2) -- (t6);
		\draw[thick, shorten >=1pt,<->,shorten <=1pt] (s2) -- (t7);
	\end{tikzpicture}
	\subcaption{Despite not having bad markings, the respective Büchi game observes via backward moves that no winning strategy for the system players exists.}
	\label{fig:ndet}
	\end{subfigure}
	\caption{Two Petri games illustrate the necessity backward moves.}
	\label{fig:PGexamples}
\end{figure}
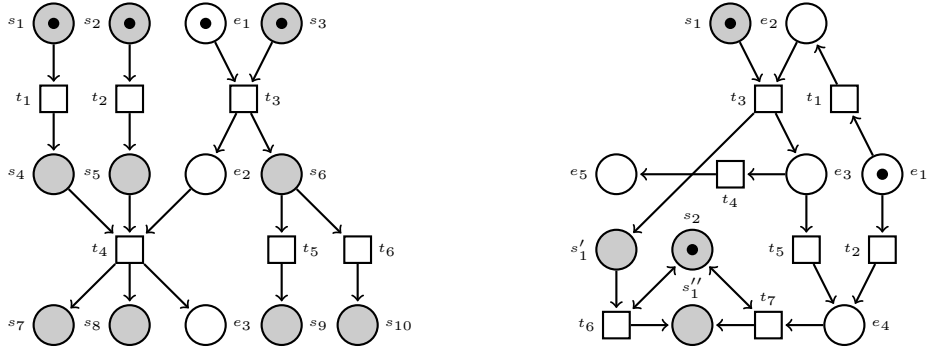

\subsection{Necessity of Backward Moves}

We illustrate the necessity of backward moves in our reduction with two examples.
With \refFig{PGexamples}, we show how the Büchi games encoding two example Petri games are won by Player~1 as both Petri games do not have a winning strategy for the system players.
In both cases, the use of backward moves is essential.

In \refFig{badmarking}, both $t_1$ and~$t_2$ fire before $t_3$ fires in the Büchi game as $t_1$ and $t_2$ have only system places in their precondition.
Only after $t_3$ fires, Player~0 can decide between $t_5$ and~$t_6$.
To reach one of the bad markings $\{s_1, s_5, e_2, s_9\}$ and $\{s_2, s_4, e_2, s_{10}\}$, it is required that only one of the transitions $t_1$ and $t_2$ fired when deciding between~$t_5$ and $t_6$.
Backward moves allow from either reachable marking $\{s_4,s_5,e_2,s_9\}$ or $\{s_4,s_5,e_2,s_{10}\}$ to apply the backward move for $t_1$ or $t_2$ to prove that the bad markings cannot be avoided.

In \refFig{ndet}, the two markings $\{s_1',s_2,e_5\}$ (by $t_1$, $t_3$, and $t_4$ firing) and $\{s_1,s_2,e_4\}$ (by $t_2$ firing) are reachable or the system players do not avoid deadlocks.
Thus, Player~0 has to allow for the system player in $s_2$ both $t_6$ and $t_7$ to avoid deadlocks.
This decision is nondeterministic for the marking $\{s_1',s_2,e_4\}$ (by $t_1$, $t_3$, and $t_5$ firing).
When making decisions for the environment player as late as possible at mcuts, transition~$t_6$ fires before $t_5$ and marking $\{s_1',s_2,e_4\}$ is not reached.
From $\{s_1'',s_2, e_4\}$, a backward move for $t_6$ is applicable resulting in a check of marking $\{s_1',s_2,e_4\}$ where the nondeterministic decision of system player $s_2$ is found.
So, no winning strategy for the system players exists for this Petri game.

\subsection{System Scheduling vs.\ Environment Scheduling}
\label{sec:SysVsEnvScheduling}

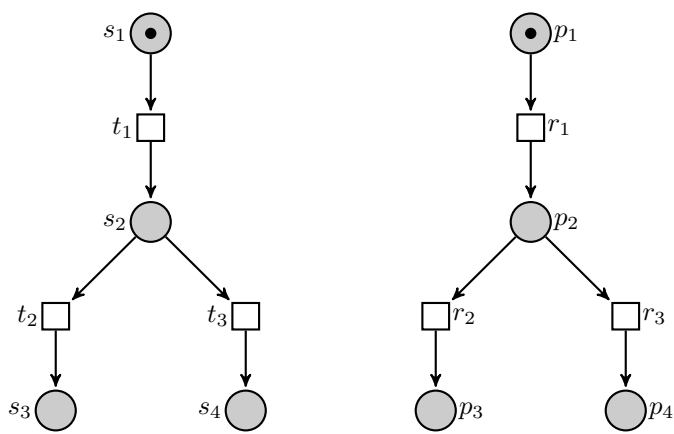
\begin{figure}[!t]
	\centering
	\begin{tikzpicture}[node distance=1.25cm, label distance=-1mm, scale=1, every node/.style={transform shape}, >=stealth']
		\node [sysplace] (s1) [tokens=1, label=left:{$s_1$}] {};
		\node [sysplace] (s2) [below of=s1, below of=s1, label=left:{$s_2$}] {};
		\node [sysplace] (s3) [below of=s2, below of=s2, left of=s2, label=left:{$s_3$}] {};
		\node [sysplace] (s4) [below of=s2, below of=s2, right of=s2, label=left:{$s_4$}] {};
		\node [transition] 	(t1)  [below of=s1, label=left:{$t_1$}] {}
			edge [pre]  (s1)
			edge [post] (s2);
		\node [transition] 	(t2)  [below of=s2, left of=s2, label=left:{$t_2$}] {}
			edge [pre]  (s2)
			edge [post] (s3);
		\node [transition] 	(t3)  [below of=s1, right of=s2, label=left:{$t_3$}] {}
			edge [pre]  (s2)
			edge [post] (s4);
		\node [sysplace] (p1) [right of=s1, right of=s1, right of=s1, right of=s1, tokens=1, label=right:{$p_1$}] {};
		\node [sysplace] (p2) [below of=p1, below of=p1, label=right:{$p_2$}] {};
		\node [sysplace] (p3) [below of=p2, below of=p2, left of=p2, label=right:{$p_3$}] {};
		\node [sysplace] (p4) [below of=p2, below of=p2, right of=p2, label=right:{$p_4$}] {};
		\node [transition] 	(r1)  [below of=p1, label=right:{$r_1$}] {}
			edge [pre]  (p1)
			edge [post] (p2);
		\node [transition] 	(r2)  [below of=p2, left of=p2, label=right:{$r_2$}] {}
			edge [pre]  (p2)
			edge [post] (p3);
		\node [transition] 	(r3)  [below of=p1, right of=p2, label=right:{$r_3$}] {}
			edge [pre]  (p2)
			edge [post] (p4);
	\end{tikzpicture}
	\caption{A simple extract of a Petri game is depicted. It consists of two system players that can each make a decision between two transitions after firing one transition. All markings of the whole Petri game containing one of the following three pairs of places $\{s_1, p_3\}$, $\{s_3, p_1\}$, and $\{s_4, p_4\}$ are bad markings.}
	\label{fig:envscheduling}
\end{figure}

	In the presented construction, Player~0 in the Büchi game decides both the decisions of the system players in the Petri game and the order in which concurrent transitions with only system places in their precondition are fired between states after an mcut and states corresponding to the next mcut. 
	We call the fact that Player~0 determines the order of these concurrent transitions the \emph{system scheduling}.
	The natural question arises what happens if Player~1 determines this order while Player~0 still decides the decisions of the system player and whether we can thereby leave out backward moves from our construction.
	We call this idea the \emph{environment scheduling}.
	
	We give a counterexample to the environment scheduling without backward moves being a suitable replacement of the system scheduling and backward moves.
	This somewhat surprising result is caused by allowing Player~0 to make different decisions for the system players in the Petri game depending on which scheduling of concurrent transitions is chosen by Player~1 without the possibility to roll fired transitions back via backward moves.
	We consider the extract of the Petri game in \refFig{envscheduling}.
	It consists of two system players and has the three bad markings $\{s_1, p_3\}$, $\{s_3, p_1\}$, and $\{s_4, p_4\}$ in this extract.
	The third bad marking $\{s_4, p_4\}$ requires that the two system players not both allow transition~$t_3$ and transition~$r_3$, which would lead to this bad marking.
	When the two system players not both allow transition~$t_3$ and transition~$r_3$, the two system players cannot always avoid the one of the other two bad markings.
	When allowing transition~$r_2$, firing transitions~$r_1$ and~$r_2$ before the system player in place~$s_1$ moves leads to the bad marking $\{s_1, p_3\}$.
	Analogously, when allowing transition~$t_2$, firing transitions~$t_1$ and~$t_2$ before the system player in place~$p_1$ moves leads to the bad marking~$\{s_3, p_1\}$.
	Therefore, the system players cannot win a Petri game where the extract of the Petri game in \refFig{envscheduling} is reachable.
	
	In \refFig{envschedulingtree}, we give a visualization of the, by Player~0 chosen, decisions of the system players in answer to the scheduling decisions by Player~1 for the environment scheduling. 
	The vertices of the tree symbolize the decisions of the system players from \refFig{envscheduling}.
	Initially, the one system player in place~$s_1$ allows transition~$t_1$ and the other system player in place~$p_1$ allows transition~$r_1$.
	After the environment scheduling decided between firing transition~$t_1$ or transition~$r_1$ as the only two enabled and allowed transitions, the moved system player makes its next decision.
	As soon as one of the players moves, it becomes acceptable for this player to allow either transition~$t_3$ or transition~$r_3$.
	When the system player from places starting with~$s$ is moved first, then it decides for transition~$t_3$ (cf.\ upper branch of \refFig{envschedulingtree}).
	When it is moved later on, then it decides for transition~$t_2$ (cf.\ lower branch of \refFig{envschedulingtree}).
	An analog statement holds for the system player from places starting with~$r$ and transitions~$r_3$ and~$t_2$.
	
	One can clearly see that a bad marking is never reached in \refFig{envschedulingtree}, which shows that the environment scheduling without backward moves determines the wrong winner for the extract of a Petri game shown in \refFig{envscheduling}.
	Only when using backward moves irrespective of whether the system scheduling or the environment scheduling is used, we can determine that no winning strategy can exist for the extract of a Petri game shown in \refFig{envscheduling}.

\begin{figure}[t]
	\centering
	\begin{tikzpicture}[shorten >=1pt, node distance=1.5cm, scale=1, every node/.style={transform shape},
	state/.style={
		rectangle, minimum size=9mm, rounded corners=1mm, draw=black, align=center
	}]
		\node[state](q_0)                {$s_1:\{t_1\}$\\ $p_1:\{r_1\}$}; 
		\node[state](q_1) [right of=q_0, above of=q_0, yshift=+0.75cm] {$s_2:\{t_3\}$\\ $p_1:\{r_1\}$};
		\node[state](q_2) [right of=q_0, below of=q_0, yshift=-0.75cm] {$s_1:\{t_1\}$\\ $p_2:\{r_3\}$};
		\node[state](q_3) [right of=q_1, right of=q_1, above of=q_1] {$s_4:\emptyset$\\ $p_1:\{r_1\}$};
		\node[state](q_4) [right of=q_3, right of=q_3] {$s_4:\emptyset$\\ $p_2:\{r_2\}$};
		\node[state](q_5) [right of=q_4, right of=q_4] {$s_4:\emptyset$\\ $p_3:\emptyset$};
		\node[state](q_6) [right of=q_1, right of=q_1] {$s_2:\{t_3\}$\\ $p_2:\{r_2\}$};
		\node[state](q_7) [right of=q_6, right of=q_6] {$s_4:\emptyset$\\ $p_2:\{r_2\}$};
		\node[state](q_8) [right of=q_7, right of=q_7] {$s_4:\emptyset$\\ $p_3:\emptyset$};
		\node[state](q_9) [right of=q_6, right of=q_6, below of=q_6] {$s_2:\{t_3\}$\\ $p_3:\emptyset$};
		\node[state](q_10) [right of=q_9, right of=q_9] {$s_4:\emptyset$\\ $p_3:\emptyset$};
		\node[state](Q_3) [right of=q_2, right of=q_2, above of=q_2] {$s_1:\{t_1\}$\\ $p_4:\emptyset$};
		\node[state](Q_4) [right of=Q_3, right of=Q_3] {$s_2:\{t_2\}$\\ $p_4:\emptyset$};
		\node[state](Q_5) [right of=Q_4, right of=Q_4] {$s_3:\emptyset$\\ $p_4:\emptyset$};
		\node[state](Q_6) [right of=q_2, right of=q_2] {$s_2:\{t_2\}$\\ $p_2:\{r_3\}$};
		\node[state](Q_7) [right of=Q_6, right of=Q_6] {$s_3:\emptyset$\\ $p_2:\{r_3\}$};
		\node[state](Q_8) [right of=Q_7, right of=Q_7] {$s_3:\emptyset$\\ $p_4:\emptyset$};
		\node[state](Q_9) [right of=Q_6, right of=Q_6, below of=Q_6] {$s_2:\{t_2\}$\\ $p_4:\emptyset$};
		\node[state](Q_10) [right of=Q_9, right of=Q_9] {$s_3:\emptyset$\\ $p_4:\emptyset$};
		\path[->, >=stealth'] (q_0) edge                node [above] {$t_1$} (q_1)
                  		edge 				node [below] {$r_1$} (q_2)
                  (q_1) edge                node [above] {$t_3$} (q_3)
                  		edge                node [above] {$r_1$} (q_6)
                  (q_3) edge                node [above] {$r_1$} (q_4)
                  (q_4) edge                node [above] {$r_2$} (q_5)
                  (q_6) edge                node [above] {$t_3$} (q_7)
                  		edge 				node [above] {$r_2$} (q_9)
                  (q_7) edge                node [above] {$r_2$} (q_8)
                  (q_9) edge                node [above] {$t_3$} (q_10)
                  (q_2) edge                node [above] {$r_3$} (Q_3)
                  		edge                node [above] {$t_1$} (Q_6)
                  (Q_3) edge                node [below] {$t_1$} (Q_4)
                  (Q_4) edge                node [below] {$t_2$} (Q_5)
                  (Q_6) edge                node [below] {$t_2$} (Q_7)
                  		edge 				node [below] {$r_3$} (Q_9)
                  (Q_7) edge                node [below] {$r_3$} (Q_8)
                  (Q_9) edge                node [below] {$t_2$} (Q_10)
                  ;
    	\draw[<-,>=stealth'] (q_0) -- node {} ++(-1.25cm,0); 
	\end{tikzpicture}
	\caption{A tree is depicted that has the decisions of the system players from \refFig{envscheduling} as vertices in answer to the fired transitions chosen by the environment scheduling.}
	\label{fig:envschedulingtree}
\end{figure}
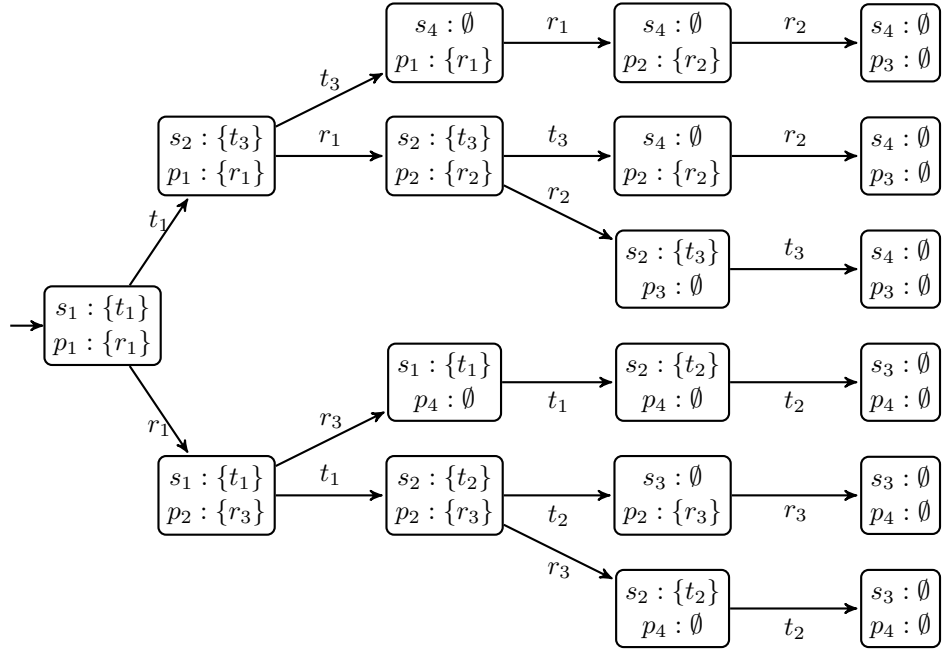

\section{Undecidability in Petri Games with Good Markings}
\label{sec:undec}

In this section, we give the formal reduction from the Post correspondence problem to Petri games with good and bad markings and from Petri games with good and bad markings to Petri games with good markings.
First, we recall the undecidability result for the synthesis of distributed systems in the synchronous setting of Pnueli and Rosner~\cite{DBLP:conf/focs/PnueliR90, DBLP:journals/ipl/Schewe14}.
Second, we present the first reduction by formally defining the two system players, the environment player, the good markings, and the bad markings.
Note that we define a Petri game based on a safe Petri net which allows us to use sets of places instead of multisets over places for simpler notation.
Third, we prove the first reduction correct.
Fourth, we present the second reduction by formally defining the translation of Petri games with both good and bad markings to Petri games with good markings.
Fifth, we prove the second reduction correct.

We recap the definition of the Post correspondence problem~\cite{post1946variant}.
The \emph{Post correspondence problem} (PCP) is to determine, for a finite alphabet $\Sigma$ and two finite lists $r_0, r_1, \ldots, r_n$ and $v_0, v_1, \ldots, v_n$ of non-empty words over $\Sigma$, if there exists a non-empty sequence $i_1, i_2, \ldots, i_l \in \{0,1,\ldots,n\}$ such that $r_{i_1}r_{i_2}\ldots r_{i_l} = v_{i_1}v_{i_2}\ldots v_{i_l}$.
This problem is undecidable.

Throughout the section, we use unique variables for parts of our reduction.
The number of indices is $n$ and $i$ is used to iterate over the possible indices.
To differentiate between the two system player, we use $k\in\{1,2\}$.
The word of an index is identified with $w_i$ for both system players, with $r_i$ specifically for the first system player and with $v_i$ specifically for the second system player.
We use $j$ to iterate over words letter-by-letter.
Labels of transitions and letters are identified by $l$.
We use $x$ to either identify an index counter or its value for both system players, $y$ specifically for an index counter or its value for the first system player, and $z$ specifically for an index counter or its value for the second system player.
The same applies to the letter counter with $a$, $b$, and $c$.

\subsection{Undecidability in the Synchronous Setting}

Synthesis of distributed systems in the synchronous setting of Pnueli and Rosner is undecidable~\cite{DBLP:conf/focs/PnueliR90}.
In the proof, the specification forces two synchronous system players $P_0$ and~$P_1$ with incomplete information on the environment $\mathit{Env}$ (cf.\ \refFig{distributed-architecture-pnueli-rosner}) to simulate a Turing machine by outputting sequences of tape configurations and to eventually reach a halting tape configuration.
The environment can change the distance between the simulation at the two systems players and can check the synchronously output tape configurations for correctness.
Each system player has no information about the distance to the simulation of the other one.
Thus, simulating the Turing machine at both players is the only possible winning strategy and synthesis of synchronous distributed systems can decide whether this strategy reaches a halting tape configuration.
The specification of the two players can be encoded in LTL~\cite{DBLP:journals/ipl/Schewe14}.

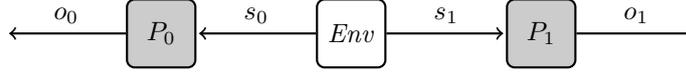
\begin{figure}[t]
	\centering
	\begin{tikzpicture}[shorten >=1pt, node distance=2.5cm, on grid,
	state/.style={
		rectangle ,minimum size=9mm, rounded corners=1mm,
		draw=black,
	},
	terminal/.style={
		rectangle ,minimum size=9mm, rounded corners=1mm,
		draw=black,
		top color=black!20, bottom color=black!20
	}]
		\node[state]   (q_0)                {$\mathit{Env}$};
		\node[terminal](q_1) [left=of q_0]  {$P_0$};
		\node[terminal](q_2) [right=of q_0] {$P_1$};
		\node[state]   (q_3) [left=of q_1, opacity=0]  {};
		\node[state]   (q_4) [right=of q_2, opacity=0] {};
		\path[->] (q_0) edge                node [above] {$s_0$} (q_1)
                  		edge 				node [above] {$s_1$} (q_2)
                  (q_1) edge					node [above] {$o_0$} (q_3)
                  (q_2) edge					node [above] {$o_1$} (q_4);
	\end{tikzpicture}
	\caption{Architecture of the undecidability proof for synthesis of synchronous distributed systems.}
	\label{fig:distributed-architecture-pnueli-rosner}
\end{figure}

\subsection{First Reduction}

We present the reduction from the Post correspondence problem (PCP) to Petri games with good and bad markings.

\subsubsection{The Two System Players}

The two system players without their \MODT counters are given in \refFig{structure-system} on the left side.
We display transitions with their label which is equal to the transition when the label is unique, e.g., $\tau$ or the letters to represent words $w_i$ are only labels as they occur more than once whereas $\#_k$ is both the transition and its label as it occurs only once in each system player.
Let $k\in\{1,2\}$ differentiate between the first system player and the second system player and $w$ identify the words of the players (i.e., $w=r$ if $k=1$ or $w=v$ if $k=2$).
For each of the two system players, the system place $p^k_\mathit{start}$ with the token initially has a choice between the indices $0, \ldots, n$.
Afterward, the corresponding word $w_0, \ldots, w_n$ is output letter-by-letter finished by a transition labeled by~$\tau$ (cf.\ right side of \refFig{structure-system}).
In all cases, the system place $p^k_\mathit{choice}$ is reached which presents the choices with the same label to the system player as from $p^k_\mathit{start}$ with the additional option to terminate with transition $\#_k$.

Next, we add the index counter $\cl^k_\inde = \{0^k_\inde, 1^k_\inde, 2^k_\inde\}$ and letter counter $\cl^k_\lett = \{0^k_\lett, 1^k_\lett, 2^k_\lett\}$.
Places have the form $\mathit{ID}^k \times \cl^k_\inde \times \cl^k_\lett$ where $\mathit{ID}^k$ are the places from \refFig{structure-system}.
Transitions labeled with numbers increase the \MODT counter for indices by one and transitions labeled with letters increase the \MODT counter for letters by one.
The termination transition $\#_k$ and transitions labeled by $\tau$ leave the counters unchanged and initially the counters are set to zero.

\begin{figure}[t]
	\centering
	\begin{tikzpicture}[on grid]
		\node [sysplace, tokens=1]	(p0) [label=right:$p^k_\mathit{start}$] {};
		\node [sysplace]			 	(p00) [below=of p0, opacity=0, label={[label distance=-0.45cm]90:...}] {};
		\node [sysplace]			 	(p1) [below left=2cm and 1cm of p0, label=right:$p^k_0$] {};
		\node [sysplace]			 	(p2) [below right=2cm and 1cm of p0, label=left:$p^k_n$] {};
		\node [sysplace]			 	(p11) [below=2cm of p1
		] {};
		\node [sysplace]			 	(p22) [below=2cm of p2
		] {};
		\node [sysplace]			 	(p3) [below=6cm of p0, label=left:$p^k_\mathit{choice}$] {};
		\node [sysplace]			 	(p33) [below=of p3, opacity=0, label={[label distance=-0.45cm]90:...}] {};
		\node [sysplace]			 	(p6) [right=3cm of p3, label=right:$p^k_\mathit{term}$] {};

		\node [transition, below left=of p0, label=left:$0$]	(t1) {}
			edge [pre]  (p0)
			edge [post] (p1);
		\node [transition,  below right=of p0, label=right:$n$] (t2) {}
			edge [pre]  (p0)
			edge [post] (p2);

		\node [transition,  below left=of p3, label=below:$0$] (t3) {}
			edge [pre]  (p3)
			edge [post, bend angle=90, bend left] (p1);
		\node [transition,  below right=of p3, label=below:$n$] (t4) {}
			edge [pre]  (p3)
			edge [post, bend angle=90, bend right] (p2);

		\node [transition, below=of p1, opacity=0, label={[label distance=-0.3cm]90:...}, label=left:$w_0$]	(t5) {}
			edge [pre]  (p1)
			edge [post] (p11)
			;
		\node [transition, below=of p2, opacity=0, label={[label distance=-0.3cm]90:...}, label=right:$w_n$]	(t6) {}
			edge [pre]  (p2)
			edge [post] (p22)
			;

		\node [transition, below=of p11, label=left:$\tau$]	(t7) {}
			edge [pre]  (p11)
			edge [post] (p3);
		\node [transition, below=of p22, label=right:$\tau$]	(t8) {}
			edge [pre]  (p22)
			edge [post] (p3);

		\node [transition, right=of p3, label=above:$\#_k$]	(t9) {}
			edge [pre]  (p3)
			edge [post] (p6);

		\node [sysplace, right=7cm of p0, label=right:$p^k_i$]							(s0) {};
		\node [sysplace, below=1.8cm of s0, label=right:$p^k_{i:0:w_i[0]}$]						(s1) {};
		\node [sysplace, below=1.8cm of s1, opacity=0, label={[label distance=-0.35cm]90:...}]	(s11) {};
		\node [sysplace, below=3.6cm of s1, label=right:$p^k_{i:|w_i|-1:w_i[|w_i|-1]}$]			(s2) {};
		\node [sysplace, below=1.8cm of s2, label=right:$p^k_\mathit{choice}$]						(s3) {};
		\node [transition, below=0.9cm of s0, label=right:$w_i[0{]}$]	(t0r) {}
			edge [pre]  (s0)
			edge [post] (s1)
			;
		\node [transition, below=0.9cm of s2, label=right:$\tau$]	(t1r) {}
			edge [pre]  (s2)
			edge [post] (s3)
			;
		\node [transition, below=0.9cm of s1, label=right:$w_i[1{]}$]	(t2r) {}
			edge [pre]  (s1)
			edge [post] (s11)
			;
		\node [transition, below=0.9cm of s11, label=right:$w_i[|w_i|-1{]}$]	(t3r) {}
			edge [pre]  (s11)
			edge [post] (s2)
			;

	\end{tikzpicture}
	\caption{Let $k=1$ identify the first system player with $w=r$ and let $k=2$ identify the second system player with $w=v$. On the left, each of the two system players is depicted without \MODT counters. It is abbreviated how the words $w_0, \ldots, w_n$ are output letter-by-letter by each of the two system players. On the right, this abbreviation between places $p^k_i$ for $i \in \{0, \ldots, n\}$ and $p^k_\mathit{choice}$ is given explicitly.}
	\label{fig:structure-system}
\end{figure}
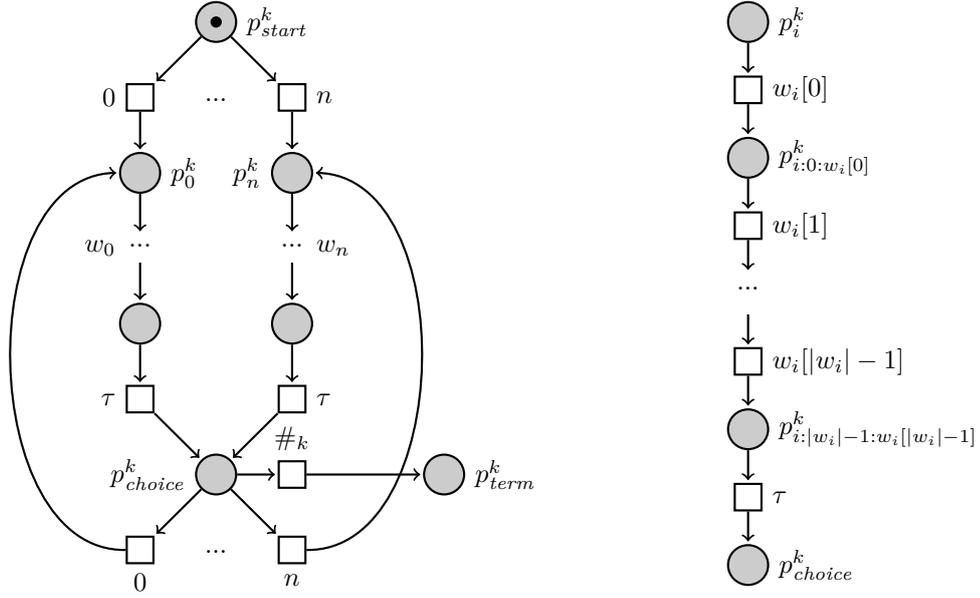

For $k \in \{1,2\}$, the Petri games for each of the two system players has the form $\pGame^k_S = (\plS^k, \emptyset, \tr^k, \fl^k, \init^k, (\emptyset, \emptyset))$ with $\plS^k$, $\tr^k$, $\fl^k$, and $\init^k$ defined as follows.
For words $w_i$ of length at least one, we introduce the notation $w_i[j]$ for $0\leq j\leq |w_i|-1$ to obtain the letters of the word.
A place $p^k_{i:j:l}$ with $0\leq i\leq n$ and $0\leq j \leq |w_i| - 1$ is reached after the $j$-th letter~$l\in\Sigma$ of word $w_i$ has been output.
We define the IDs of places as
$\mathit{ID}^k = \{p^k_\mathit{start}, p^k_\mathit{choice}, p^k_\mathit{term}, p^k_0, \ldots, p^k_n\} \cup \{ p^k_{i:j:w_i[j]} \mid 0\leq i\leq n \land 0\leq j\leq |w_i| - 1 \}$
and the set of system places as
$\plS^k = \{ (p^k_\mathit{start}, 0^k_\inde, 0^k_\lett) \} \cup \{ (\id^k, x^k, a^k) \mid \id^k \in \mathit{ID}^k \setminus \{p^k_\mathit{start}\} \land x^k \in \cl^k_\inde \land a^k \in \cl^k_\lett \}$.
The initial marking $\init^k$ is $\{(p^k_\mathit{start}, 0^k_\inde, 0^k_\lett)\}$.

We introduce the set of \textit{labels} $\labels$ for transitions as $\labels = \{\#_k, \tau\} \cup \{0, \ldots, n\} \cup \Sigma$.
Further, we introduce the notation $p$ $\firable{l}^k$ $p'$ to define a unique transition with label $l \in \labels$ from place~$p$ to place~$p'$ as transition $t_{p \firable{l} p'}^k\in\tr^k$ with the corresponding arcs $(p, t_{p\firable{l}p'}^k), (t_{p\firable{l}p'}^k, p') \in \fl^k$.
In the following, we use this notation to define the transitions $\tr^k$ and the flow $\fl^k$ step-by-step.
The outgoing transitions of the initial place and their flow are defined by
\[\{ (p^k_\mathit{start}, 0^k_\inde, 0^k_\lett) \firable{i}^k (p^k_i, 1^k_\inde, 0^k_\lett) \mid 0\leq i\leq n \}.\]
The outgoing transitions of the choice place and their flow are defined for indices by
\[
\{ (p^k_\mathit{choice}, x^k_\inde, a^k_\lett) \firable{i}^k (p^k_i, ((x+1) \bmod 3)^k_\inde, a^k_\lett) \mid 0\leq i\leq n \land 0\leq x, a \leq 2\}
\]
and for termination by
\[\{(p^k_\mathit{choice}, x^k_\inde, a^k_\lett) \firable{\#_k}^k (p^k_\mathit{term}, x^k_\inde, a^k_\lett) \mid 0\leq x, a\leq 2 \}.\]
The letter-by-letter output is outlined on the right in \refFig{structure-system} and defined as follows:
The first transitions of words are defined as 
\begin{align*}
	\{(p^k_i, x^k_\inde, a^k_\lett) \firable{w_i[0]}^k (p^k_{i:0:w_i[0]}, x^k_\inde, (&(a+1)\bmod 3)^k_\lett) \mid\\& 0\leq i \leq n \land |w_i| > 0 \land 0\leq x, a \leq 2 \}.
\end{align*}
The transitions for the remainder of the words are defined as
\begin{align*}
	\{(p^k_{i:j-1:w_i[j-1]}, x^k_\inde, a^k_\lett)&\firable{w_i[j]}^k (p^k_{i:j:w_i[j]}, x^k_\inde, ((a+1)\bmod 3)^k_\lett)\mid\\& 0\leq i \leq n \land |w_i| > 0 \land 0\leq x, a \leq 2 \land 1\leq j \leq |w_i| - 1 \}.
\end{align*}
We add a transition labeled by $\tau$ to return to place $p^k_\mathit{choice}$ by
\begin{align*}
	\{(p^k_{i:|w_i|-1:w_i[|w_i|-1])}, x^k_\inde, a^k_\lett)\firable{\tau}^k (&p^k_\mathit{choice}, x^k_\inde, a^k_\lett) \mid\\& 0\leq i \leq n \land |w_i| > 0 \land 0\leq j, m \leq 2 \}.
\end{align*}

\subsubsection{The One Environment Player}

The environment player is depicted in \refFig{structure-environment} and makes two decisions in one step:
It decides whether it wishes to check the output indices ($t_{i \land *}$) or letters ($t_{l \land * }$).
Furthermore, the environment player decides whether it suspects the first ($t_{* \land 1}$) or second ($t_{* \land 2}$) system player to terminate untruthfully or whether the termination of the system players is okay ($t_{* \land \mathit{okay}}$).
The result is stored in six pairs of the form $\{\inde,\lett\} \times \{\mathit{first,\mathit{okay,\mathit{second}}}\} $.
As all places and transition are unique in \refFig{structure-environment}, this formally defines the Petri game $\pGame_E$ for the environment player.

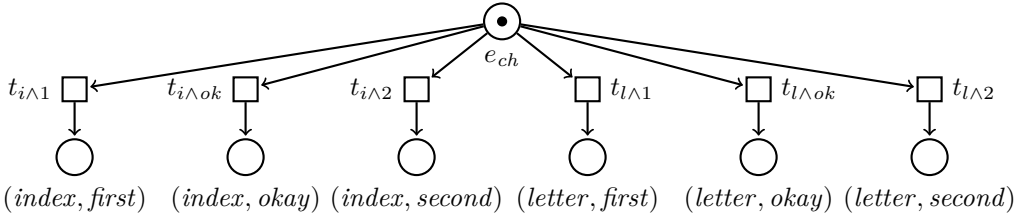
\begin{figure}
	\centering
	\begin{tikzpicture}[on grid, scale=0.9] 
		\parseNodeDistance
		\node [envplace, tokens=1, label=below:$e_{\mathit{ch}}$]			 						(e0) {};
		\node [envplace, below=2*\nodeDisX of e0, xshift=-0.25cm, left of=e0, label=below:{$(\inde,\mathit{second})$}]	 				(e3) {};
		\node [envplace, left=2.25*\nodeDisX of e3, xshift=-0.25cm, label=below:{$(\inde,\mathit{okay})$}]	 				(e2) {};
		\node [envplace, left=2.25*\nodeDisX of e2, xshift=-0.25cm, label=below:{$(\inde,\mathit{first})$}]	 				(e1) {};
		\node [envplace, right=2.25*\nodeDisX of e3, xshift=0.25cm, label=below:{$(\lett,\mathit{first})$}]	 				(e4) {};
		\node [envplace, right=2.25*\nodeDisX of e4, xshift=0.25cm, label=below:{$(\lett,\mathit{okay})$}]	 				(e5) {};
		\node [envplace, right=2.25*\nodeDisX of e5, xshift=0.25cm, label=below:{$(\lett,\mathit{second})$}]	 				(e6) {};
		\node [transition, above=of e1, label=left:$t_{i\land 1}$]	(t0) {}
			edge [pre]  (e0)
			edge [post] (e1)
			;
		\node [transition, above=of e2, label=left:$t_{i\land \mathit{ok}}$]	(t1) {}
			edge [pre]  (e0)
			edge [post] (e2)
			;
		\node [transition, above=of e3, label=left:$t_{i\land 2}$]	(t2) {}
			edge [pre]  (e0)
			edge [post] (e3)
			;
		\node [transition, above=of e4, label=right:$t_{l\land 1}$]	(t3) {}
			edge [pre]  (e0)
			edge [post] (e4)
			;
		\node [transition, above=of e5, label=right:$t_{l\land \mathit{ok}}$]	(t4) {}
			edge [pre]  (e0)
			edge [post] (e5)
			;
		\node [transition, above=of e6, label=right:$t_{l\land 2}$]	(t5) {}
			edge [pre]  (e0)
			edge [post] (e6)
			;
	\end{tikzpicture}
	\caption{The one environment player is depicted.}
	\label{fig:structure-environment}
\end{figure}

\subsubsection{Good Markings and Bad Markings}

The good and bad markings cover the two system players and the environment player.
They are implicitly ordered as the first system player, second system player, and then the environment player.
We use the notation $*$ to indicate that all places of the player at the respective position are part of the good or bad markings.
This notation can be applied to one element of the pair representing the decision by the environment player meaning that only the result of the other decision is important.
We ensure that the environment player makes its decision first and therefore can ignore the place $e_{\mathit{ch}}$ for $*$.

We start by defining good markings which indicate that both system players terminated with solutions of the same length (modulo three) as
\[\goodmarkings^\mathit{finish} = \{ \{ (p^1_\mathit{term}, x^1_\inde, a^1_\lett), (p^2_\mathit{term}, x^2_\inde, a^2_\lett), (*, \mathit{okay}) \} \mid 0\leq x, a \leq 2 \}.\]

Next, we define good markings to focus on linear firing sequences for indices as
\begin{align*}
	\goodmarkings^\inde =\{&\{ (p^1_{i_1}, y^1_\inde, b^1_\lett), (p^2, z^2_\inde, c^2_\lett), (\inde, *) \},\\&\{ (p^1, y^1_\inde, b^1_\lett), (p^2_{i_2}, z^2_\inde, c^2_\lett), (\inde, *) \} \mid \\ &~p^1 \in \mathit{ID}^1 \setminus \{ p^1_\mathit{term} \} \land p^2 \in \mathit{ID}^2 \setminus \{ p^2_\mathit{term} \} \land 0 \leq i_1, i_2 \leq n~\land \\ &~(y\parallel z) \in \{ (0\parallel 1), (2\parallel 0), (1\parallel 2) \} \land 0 \leq b, c \leq 2 \}
\end{align*}
and for letters as
\begin{align*}
\goodmarkings^\lett =\{&
\{ (p^1_{i_1:j_1:w_{i_1}[j_1]}, y^1_\inde, b^1_\lett), (p^2, z^2_\inde, c^2_\lett), (\lett, *) \},\\&
\{ (p^1, y^1_\inde, b^1_\lett), (p^2_{i_2:j_2:w_{i_2}[j_2]}, z^2_\inde, c^2_\lett), (\lett, *) \} \mid\\&
p^1 \in \mathit{ID}^1 \setminus \{ p^1_\mathit{term} \} \land p^2 \in \mathit{ID}^2 \setminus \{ p^2_\mathit{term} \}~\land\\&
0\leq i_1, i_2 \leq n \land0\leq j_1 \leq |w_{i_1}| - 1 \land 0\leq j_2 \leq |w_{i_2}| - 1~\land\\&
0 \leq y, z \leq 2 \land (b\parallel c) \in \{ (0\parallel 1), (2\parallel 0), (1\parallel 2) \} \}.
\end{align*}
The pairs of numbers for the respective \MODT counters are motivated in \refFig{linear_scheduling}.
The index counter only increases at places $p^k_{i_k}$ whereas the letter counter only increases at places $p^k_{i_k:j_k:w_{i_k}[j_k]}$.
Therefore, the place of one system player is fixed.

We introduce good markings to identify situations where the environment player claimed that one system player does not terminate but, in fact, this system player terminated as
\begin{align*}
 \goodmarkings^\mathit{term} =\:&\{\{ (p^1_\mathit{term}, y^1, b^1), *, (*, \mathit{first}) \} \mid y^1 \in \cl^1_\inde \land b^1 \in \cl^1_\lett \}~\cup\\& \{ \{ *, (p^2_\mathit{term}, z^2, c^2), (*, \mathit{second}) \} \mid z^2 \in \cl^2_\inde \land c^2 \in \cl^2_\lett \}.
\end{align*}

We define further good markings to only consider firing sequences starting with the environment player's decisions.
As the system players never learn about the decision of the environment, this only simplifies our proof later.
We define the good marking
\begin{align*}
\goodmarkings^\mathit{envfirst} =\:&
\{\{ (p^1_i, 1^1_\inde, 0^1_\lett), * , e_\mathit{ch} ) \} \mid 0\leq i \leq n \}~\cup\\&
\{ \{ *, (p^2_i, 1^2_\inde, 0^2_\lett) , e_\mathit{ch} ) \} \mid 0\leq i \leq n \}.
\end{align*}

We define bad markings for different output indices or letters at the same counter positions depending on the choice of the environment player as
\begin{align*}
\badmarkings^\inde = \{&\{ (p^1_{i_1}, x^1_\inde, b^1), (p^2_{i_2}, x^2_\inde, c^2), (\inde, *)\}  \mid\\& 0 \leq i_1,i_2 \leq n \land i_1 \neq i_2 \land 0 \leq x \leq 2 \land b^1 \in \cl^1_\lett \land c^2 \in \cl^2_\lett \}
\end{align*}
and
\begin{align*}
\badmarkings^\lett = \{&\{(p^1_{i_1:j_1:w_{i_1}[j_1]}, y^1, a^1_\lett), (p^2_{i_2:j_2:w_{i_2}[j_2]}, z^2, a^2_\lett), (\lett, *) \} \mid\\& 0 \leq i_1,i_2 \leq n \land 0 \leq j_1 \leq |w_{i_1}| - 1 \land 0 \leq j_2 \leq |w_{i_2}| - 1~\land\\& w_{i_1}[j_1] \neq w_{i_2}[j_2]\land y^1 \in \cl^1_\inde \land z^2 \in \cl^2_\inde \land 0 \leq a \leq 2 \}.
\end{align*}
When one system player terminated, the other system player can make at most one legal step and afterward should terminate as well.
Therefore, we define markings where one system player terminated and the other one makes more than one additional step as bad with

\begin{align*}
\badmarkings^{\mathit{term}_\inde}= \{&
\{ (p^1_\mathit{term}, y^1_\inde, b^1) , (p^2_i, z^2_\inde, c^2), (\inde, *) \},\\&
\{ (p^1_i, y^1_\inde, b^1) , (p^2_\mathit{term}, z^2_\inde, c^2), (\inde, *) \}
\mid\\&0 \leq i \leq n \land (y\parallel z) \in \{ (0\parallel 1), (2\parallel 0), (1\parallel 2) \} \land c^1 \in \cl^1_\lett \land  z^2 \in \cl^2_\lett \}
\end{align*}
and
\begin{align*}
\badmarkings^{\mathit{term}_\lett} =\:&
\{\{ (p^1_\mathit{term}, y^1, b^1_\lett) , (p^2_{i:j:v_i[j]}, z^2, c^2_\lett), (\lett, *) \} \mid\\&~
0 \leq i \leq n \land 0 \leq j \leq |v_i| - 1 \land (b\parallel c) \in \{ (0\parallel 1), (2\parallel 0), (1\parallel 2) \}~\land\\&~ y^1 \in \cl^1_\inde \land  z^2 \in \cl^2_\inde \}~\cup\\&
\{ \{ (p^1_{i:j:r_i[j]}, y^1, b^1_\lett) , (p^2_\mathit{term}, z^2, c^2_\lett), (\lett, *) \} \mid\\&~
0 \leq i \leq n \land 0 \leq j \leq |r_i| - 1 \land (b\parallel c) \in \{ (0\parallel 1), (2\parallel 0), (1\parallel 2) \}~\land\\&~ y^1 \in \cl^1_\inde \land z^2 \in \cl^2_\inde \}.
\end{align*}
Both system players terminating with different counter values for index or letter counter is also bad behavior but from there no good marking is reachable anymore and we do not need to define it as bad marking.

The constructed Petri game has the form $\pGame = \pGame^1 || \pGame^2 || \pGame^E$ where $||$ defines the parallel composition as the disjoint union over system and environment places, transitions, flows and initial markings with the good markings $\goodmarkings = \goodmarkings^\mathit{finish} \cup \goodmarkings^\inde \cup \goodmarkings^\lett \cup \goodmarkings^\mathit{term} \cup \goodmarkings^\mathit{envfirst}$ and the bad markings $\badmarkings = \badmarkings^\inde \cup \badmarkings^\lett \cup \badmarkings^{\mathit{term}_\inde} \cup \badmarkings^{\mathit{term}_\lett} $.

\subsection{Correctness of the First Reduction}

It is an easy check to see that our usage of $*$ results in disjoint good and bad markings.
Our reduction further only requires two system players and a single environment player.

\begin{lemma}[Linear firing sequence]\label{lem:PCPlinearScheduling}
	Bad markings for different output indices or letters are reached without a good marking being reached before iff the letters are output at the same position.
\end{lemma}
\begin{proof}
	For a strategy with different output indices or letters at position $q$ and $q+3$,
	a good marking is reached at one of the positions $(q\parallel q+1)$, $(q+2\parallel q)$, or $(q+1\parallel q+2)$ as indicated in \refFig{linear_scheduling}.
	Therefore, positions $q$ and $q+3$ are never compared and no bad marking can be reached for inequality without reaching a bad marking before.
\end{proof}

\begin{lemma}[Same number of output indices/letters]\label{lem:PCPsizeWinStrategy}
	For the constructed Petri game $\pGame$, strategies outputting a different number of indices or letters at the two system players are not winning.
\end{lemma}
\begin{proof}
	Assume the strategy has output $q$ indices at the first player and $q+1$ indices at the second player (letters analog).
	From the structure of the system players, we know that both strategies have to terminate after outputting the letters of the output indices.
	After both players have terminated, no good marking from $\goodmarkings^\mathit{finish}$ is reached because the index counters differ.
	No good marking was reached before for the linear firing sequence.

	Assume the strategy outputs $q$ indices at the first player and $q+2$ indices at the second player (letters analog).
	A corresponding bad marking from $\badmarkings^{\mathit{term}_\inde}$ is reached.
\end{proof}

This proves the absence of untruthful termination.

\begin{lemma}[Finite winning strategies]\label{lem:PCPfiniteWinStrategy}
	For the constructed Petri game $\pGame$, all infinite strategies are not winning.
\end{lemma}
\begin{proof}
	The only possibility for an infinite strategy is to not fire $\oneend$ or $\twoend$ at one or at both system players.
	If both system players do not fire $\oneend$ and $\twoend$, then the same or two different infinite solutions to the PCP are given which is not winning because, for the linear firing sequence on the indices or the letters, no good marking can be reached.
	If one process does not fire $\oneend$ or $\twoend$, then different solutions to the PCP are given which is not winning by \refLemma{PCPsizeWinStrategy}.
\end{proof}

\begin{lemma}[Same output of winning strategies]\label{lem:PCPsameOutput}
	For the constructed Petri game $\pGame$, strategies outputting a different index or letter at the two system players at the same position are not winning.
\end{lemma}
\begin{proof}
	Assume that the strategy differs at position $q$ on the indices $i_1$ and~$i_2$. The case for letters is analog.
	For the turn-taking run of the system players after the decision of the environment player for $\inde$ and $\mathit{okay}$, the bad marking $\{ (p^1_{i_1}, x^1_\inde, b^1) , (p^2_{i_2}, x^2_\inde, c^2 ) , (\inde, \mathit{okay}) \}$ for the same value $x$ of the \MODT index counter and arbitrary \MODT letter counters $b^1$ and $c^2$ is reached.
\end{proof}

\begin{lemma}[Strategy translation]\label{lem:PCPstrategyTranslation}
	For every instance of the PCP, there exists a solution to the instance iff there exists a winning strategy in the Petri game from our construction.
\end{lemma}
\begin{proof}
``$\Rightarrow$'': Let $i_1,i_2,\ldots,i_l$ be the solution to the PCP.
We build the corresponding strategy in the Petri game where both system players activate the transitions labeled by $i_1,w_{i_1},\tau,i_2,w_{i_2},\tau,\ldots,i_l,w_{i_l},\tau,\#_k$, where all words are output letter-by-letter, and the environment player remains unrestricted.
Initially, all three players have transitions enabled.
When the transition of the environment player is not the first transition of the considered firing sequence, good markings from $\goodmarkings^\mathit{first}$ are reached immediately and we do not need to consider these cases further.
When the environment player suspects one system player to not terminate, good markings from $\goodmarkings^\mathit{term}$ are reached eventually because the strategy terminates both system players and bad markings from $\badmarkings^{\mathit{term}_\inde}$ and $\badmarkings^{\mathit{term}_\lett}$ are avoided.
We therefore can neglect the second decision of the environment player and only need to consider the two cases further where the environment player either wants to check output indices or output letters.
For both cases, firing sequences which are not turn-taking on the considered outputs reach good markings from $\goodmarkings^\inde$ or $\goodmarkings^\lett$.
For turn-taking firing sequences, bad markings from $\badmarkings^\inde$ or $\badmarkings^\lett$ are avoided as the same indices and letters are output by the assumption that we translate a solution to the PCP.
In the end, a good marking from $\goodmarkings^\mathit{finish}$ is reached as the strategy is finite and of the same length at both system players.

``$\Leftarrow$'': Winning strategies have to be finite at both system players and have to end with a transition to the final place due to \refLemma{PCPsizeWinStrategy} and  \refLemma{PCPfiniteWinStrategy}.
We claim that the allowed sequence of indices is the same at both system players and constitutes a solution to the instance of the PCP.
If the sequence of indices is not the same, the environment player could have tested for different indices and if the sequence of indices is no solution to the instance of the PCP, the environment player could have tested for different letters.
Therefore, each possible winning strategy constitutes a solution to the Post correspondence problem.
\end{proof}

From these lemmas, it follows that the only possibility for a winning strategy is to output the same finite solution to the instance of the PCP at both system players.
Therefore, the realizability problem for Petri games with good and bad markings is undecidable as it would be able to decide the existence of such a strategy.

\begingroup
\def\thetheorem{\ref{theo:undecGoodAndBad}}
\begin{theorem}
	For Petri games with good and bad markings and at least two system and one environment player, the question if the system players have a winning strategy is undecidable.
\end{theorem}
\addtocounter{theorem}{-1}
\endgroup
\begin{proof}
	This follows from \refLemma{PCPlinearScheduling}, \refLemma{PCPfiniteWinStrategy}, \refLemma{PCPsizeWinStrategy}, \refLemma{PCPsameOutput}, and \refLemma{PCPstrategyTranslation}.
\end{proof}

\subsection{Second Reduction}

We present the reduction from Petri games with good and bad markings to Petri games with only good markings.
Environment places are added such that all players become environment players after each transition in the original Petri game.
Transitions are added to reverse the change before the next transition in the original Petri game.
For each bad marking, a transition is added from the additional environment places corresponding to the bad marking leading to a new sink place which is not part of any good marking.
The environment players can fire the transition to the sink place when a bad marking is reached in the original Petri game.
If no such transition is enabled, then the environment players have to return to system places and the next decision of the system players follows.
Notice that this construction does not introduce additional synchronization between the players but enables firing orders between the players such that all players are environment players and a transition to the sink place is enabled if a bad marking would be reached in the original Petri game.

As proven in the first reduction, Petri games with good and bad markings and at least two system and one environment player are undecidable.
Therefore, Petri games with good markings and at least three players, out of which one is always an environment player and two  can change between being a system and environment player are undecidable.
Without changing the type of players, this corresponds to Petri games with good markings and at least two system and three environment players where two pairs of system and environment player change between being active and inactive.

In the following, we use the symbols $\pi$ and $\tau$ to introduce and reference new and unique places and transitions.
Given a Petri game with good and bad markings $\pGame^A = (\plS^A,\plE^A,\tr^A,\fl^A,\init^A,(\goodmarkings^A,\badmarkings^A))$, we define the Petri game with good markings $\pGame^B = (\plS^B,\plE^B,\tr^B,\fl^B,\init^B,\goodmarkings^B)$ as
\begin{align*}
	\plS^B =&~\plS^A \cup \{ \pi'_p \mid p \in \init^A \cap \plS \},\\
	\plE^B =&~\plE^A \cup \{ \pi'_p \mid p \in \init^A \cap \plE \} \cup \{\pi_p \mid p \in \plS^A \cup \plE^A\} \cup \{\pi_{\mathit{sink}}\},\\
	\tr^B  =&~\tr^A  \cup \{\tau_t \mid t \in \tr^A \} \cup \{\tau_M \mid M \in \badmarkings^A \} \cup \{\tau_\init\},\\
	\fl^B  =&~\{(p,t) \mid (p,t) \in \fl^A \} \cup \{ (t,\pi_p) \mid (t,p) \in \fl^A \}~\cup\\
	        &~\{ (\pi_p, \tau_t) \mid (p,t) \in \fl^A \} \cup \{ (\tau_t, p) \mid (t,p) \in \fl^A\}~\cup\\
	        &~\{ (\pi_p, \tau_M) \mid M \in \badmarkings^A \land p \in M \} \cup \{ (\tau_M, \pi_{\mathit{sink}}) \mid M \in \badmarkings^A \}~\cup\\
	        &~\{ (\pi'_p,\tau_\init) \mid p \in \init^A \} \cup \{(\tau_\init, p) \mid p \in \init^A \},\\
	\init^B =&~\{\pi'_p \mid p \in \init^A\}\text{, and }
	\goodmarkings^B =\goodmarkings^A.
\end{align*}

Each original place $p$ is preceded by an environment place $\pi_p$.
They are connected by the transition $\pi_t$.
For the original successor places $p'$ of $p$, the places $\pi_{p'}$ are reached from $p$ by transition $t$.
To encode bad markings, we add a sink place $\pi_{\mathit{sink}}$ which is not part of any good marking.
For each bad marking $M$, the sink place is reachable from the places $\pi_p$ for the places $p \in M$.
The Petri game starts from places $\pi'_p$ for places $p$ from the initial marking.
The original initial marking is reached after firing transition $\tau_\init$.
Notice that this transition can fire at most once and is necessary in case the initial marking is a bad marking.

\subsection{Correctness of the Second Reduction}

\begingroup
\def\thetheorem{\ref{theo:undecGoodAndThree}}
\begin{theorem}
	For Petri games with good markings and at least three players, out of which one is an environment player and each of the other two changes between system and environment player, the question if the system players have a winning strategy is undecidable.
\end{theorem}
\addtocounter{theorem}{-1}
\endgroup
\begin{proof}
	Each winning strategy for $\pGame^A$ reaches no bad markings by the definition of being winning.
	Therefore, it can be translated into a winning strategy for $\pGame^B$ by adding places~$\pi_p$ and transitions $\tau_t$ as in the construction above.
	Reaching no bad markings implies that the sink place is never reached.
	Thus, the constructed strategy for $\pGame^B$ is a winning strategy because the maximal firing sequences of all maximal plays are the same when removing places $\pi_p$ and transitions $\tau_t$ and because places $\pi_p$ are not contained in good markings.

	Each winning strategy for $\pGame^B$ never reaches the sink place because it is not contained in any good marking and otherwise the strategy would not be winning.
	Therefore, it can be translated into a winning strategy for $\pGame^A$ by removing places~$\pi_p$ and transitions $\tau_t$ to reverse the construction above.
	Never reaching the sink place implies that no bad marking is reached.
	Thus, the constructed strategy for $\pGame^A$ is a winning strategy because the maximal firing sequences of all maximal plays are the same when removing places $\pi_p$ and transitions~$\tau_t$ and because places $\pi_p$ are not contained in good markings.
\end{proof}

\end{document}